\documentclass[a4paper,12pt]{amsart}
\usepackage[lmargin=2.4cm,rmargin=2.25cm,bmargin=3cm,tmargin=2.5cm]{geometry}

\usepackage{psfrag,epsf}
\usepackage{amsfonts}  
\usepackage{nicefrac}    
\usepackage{bbm, caption, subcaption, latexsym}
\usepackage{wrapfig}
\usepackage[table]{xcolor}
\graphicspath{{./sections/figures/}}

\usepackage[utf8]{inputenc}
\usepackage{amsmath} 
\usepackage{amssymb}
\usepackage{amsthm}
\usepackage{textcomp}
\usepackage{enumerate}
\usepackage[sort,numbers]{natbib}

\usepackage[T1]{fontenc}

\usepackage{array}
\usepackage{verbatim}
\usepackage{hyperref}
\usepackage{graphicx}
\usepackage{booktabs}
\usepackage{algorithm}
\usepackage{algpseudocode}
\usepackage{stmaryrd} 
\usepackage{xcolor}
\hypersetup{
  colorlinks,
  linkcolor={red!50!black},
  citecolor={blue!50!black},
  urlcolor={blue!80!black}
}

\usepackage{tcolorbox}
\usepackage{soul}
\usepackage[foot]{amsaddr}

\newtheorem{thm}{Theorem}
\newtheorem{lem}[thm]{Lemma}
\newtheorem*{lem*}{Lemma}

\DeclareMathOperator*{\argmax}{arg\,max}

\theoremstyle{definition}

\theoremstyle{remark}

\begin{document}

\title[Trace-class neural networks priors for Bayesian learning with MCMC]{Trace-class Gaussian priors for Bayesian learning of neural networks with MCMC}

\author{Torben Sell}
\address[TS]{School of Mathematics, University of Edinburgh}

\author{Sumeetpal S.~Singh}
\address[SSS]{Department of Engineering, University of Cambridge}

\begin{abstract} 
This paper introduces a new neural network based prior for real valued functions on $\mathbb R^d$ which, by construction, is more easily and cheaply scaled up in the domain dimension $d$ compared to the usual Karhunen-Lo\`eve function space prior. The new prior is a Gaussian neural network prior, where each weight and bias has an independent Gaussian prior, but with the key difference that the variances decrease in the width of the network in such a way that the resulting function is \emph{almost surely} well defined in the limit of an infinite width network. We show that in a Bayesian treatment of inferring unknown functions, the induced posterior over functions is amenable to Monte Carlo sampling using Hilbert space Markov chain Monte Carlo (MCMC) methods. This type of MCMC is popular, e.g. in the Bayesian Inverse Problems literature, because it is stable under \emph{mesh refinement}, i.e. the acceptance probability does not shrink to $0$ as more parameters of the function’s prior are introduced, even \emph{ad infinitum}. In numerical examples we demonstrate these stated competitive advantages over other function space priors. We also implement examples in Bayesian Reinforcement Learning to automate tasks from data and demonstrate, for the first time, stability of MCMC to mesh refinement for these type of problems.
\end{abstract} 


\maketitle
\sloppy

\section{Introduction}

Generating samples from probability measures on function spaces is both a challenging computational problem and a very useful tool for many applications, including mathematical modelling in bioinformatics \citep{quarteroni2017cardiovascular}, data assimilation in reservoir models \citep{iglesias2013evaluation}, and velocity field estimation in glaciology \citep{minchew2015early}, amongst many others. This paper addresses the problem of defining a computationally and statistically favourable function space prior.

In Bayesian inference on separable Hilbert spaces \citep{stuart2010inverse}, many posterior measures $\mu$ are absolutely continuous with respect to their prior $\mu_0$ (often a Gaussian measure, see \cite{knapik2011bayesian} and \cite{dashti2013map}, but not always, see \cite{dashti2011besov}, \cite{hosseini2017well}, and \cite{hosseini2017well2}), with the likelihood acting as the Radon-Nikodym derivative $d\mu/d\mu_0\propto \mathcal L$. Samples from a Gaussian prior on a separable Hilbert space have a convenient expansion as the weighted sum of an infinite countable basis, weighted with independent Gaussian random variables (see \eqref{KL}), which is known as the Karhunen-Lo\`eve (KL) expansion. The posteriors come with a variety of theoretical results, such as concentration inequalities and contraction rates, see e.g. \cite{agapiou2013posterior, nickl2020consistency, knapik2011bayesian, van2008rates}. Truncating the KL expansion then reduces the problem of sampling from infinite-dimensional measures to sampling from a finite-dimensional parameter space. This truncated approximation to the true posterior gets better by including more terms of the expansion. 
The practical applicability of these Gaussian priors is, however, restricted to inferring unknown functions with low-dimensional domain, as the orthogonal basis required for the KL expansion results in the complexity scaling exponentially with the dimension of the unknown function's domain. 

Another approach to define function space priors are Bayesian Neural Networks (BNNs) \citep{neal1995bayesian, neal2012bayesian} which currently enjoy a resurgence of interest, e.g. in the machine learning community. A BNN is a random function obtained by placing a prior distribution over the weights and biases of a Neural Network (NN), with the default choice being a centered Gaussian prior on the weights with variances that scale as $\mathcal O(1/N^{(l)})$, where $N^{(l)}$ is the number of nodes in layer $l$. Some authors argue for heavy-tailed priors on the parameters, which was initially investigated in \cite{neal1995bayesian}. Although some theoretical results exist \citep{matthews2018gaussian}, popular criticisms include the lack of interpretability of the resulting BNNs, and recent work \citep{wenzel2020good} has highlighted inter alia that novel priors are needed. Sampling approaches include Hamiltonian Monte Carlo \citep{neal1995bayesian}, and more advanced integrators \citep{leimkuhler2019partitioned}. However, inference is often limited to finding the maximum-a-posteriori (MAP) estimate of the posterior \citep{welling2011bayesian}, and the $\mathcal O(1/N^{(l)})$ scaling implies one cannot easily add nodes to a layer to obtain more accurate estimates: one would either have to adjust the prior variances for all nodes within the amended layer, thereby changing the prior, or not adjust the prior which results in exploding functions \citep{matthews2018gaussian}. Other function space priors include Deep Neural Networks and Deep Gaussian Processes \citep{damianou2013deep,dunlop2018deep}, and in \cite{dunlop2018deep} inference is done using similar function space MCMC techniques to the ones we employ.

To calculate expectations with respect to the Bayesian posterior of the unknown function, computational methods are required as the relevant integrals are usually not analytically tractable. Two popular sampling algorithms for posteriors defined on Hilbert spaces are the preconditioned Crank-Nicolson (pCN) algorithm and its likelihood-informed counterpart the preconditioned Crank-Nicolson Langevin (pCNL) algorithm, which arise from clever (and in a way optimal) discretisations of certain stochastic differential equations \citep{cotter2013mcmc}. These samplers are asymptotically exact and have a dimension-independent mixing rate in the sense that their proposal step size does not depend on the number of terms in the KL truncation \citep{hairer2014spectral, eberle2014error}. This stands in stark contrast to the well-known dimensional-dependent scaling of popular MCMC algorithms such as the Random Walk Metropolis-Hastings Algorithm and the Metropolis Adjusted Langevin Algorithm \citep{roberts1998optimal, roberts2001optimal}. 
Modifications of pCN include geometric \citep{beskos2017geometric} and likelihood-informed \citep{cui2016dimension} versions. Although the computational cost can be reduced provided one knows which basis functions are informed by the data, they cannot circumvent the costly scaling in the domain dimension. This is presumably one reason why these methods have rarely been used for inferring unknown functions with domains larger than dimension two (i.e. $\mathbb{R}^2$) in reported examples in the literature.

This paper introduces a new neural network based prior, coined \emph{trace-class} neural network priors, which allows for scalable (in the domain dimension) Bayesian function space inference. Hilbert space MCMC algorithms are then used to sample from the resulting posteriors, and owing to their stability under mesh-refinement, enhances the practical utility of our framework. In addition to comparisons with reported examples in the literature, we also demonstrate our technique's usefulness on a challenging $17$-dimensional Bayesian reinforcement learning example where the aim is to learn the \emph{value function} (a function on $\mathbb R^{17}$) that can automate a task demonstrated by an expert --- we combine the noisy expert data with a trace-class NN prior, through a suitably defined likelihood, to yield a Bayesian formulation.

The main contributions of this paper are as follows:
\begin{itemize}
    \item We introduce a new \emph{trace-class} Gaussian prior for neural networks, which is both well defined for infinite width NNs and has a degree of smoothness, and demonstrate its practical utility. The prior is independent, centred, and Gaussian across the NN's weights and biases but is non-exchangeable over the weights within each layer and has a summable variance sequence. The latter, which gives it the trace-class property, ensures it is a valid prior for an infinite width network, while the former results in parameters being better identified from an inference perspective. We further show that this prior is appropriate for use with Hilbert space MCMC methods (Theorem \ref{thm:prior}). The practical implications of this is that it is valid for the infinite-width limit of the NN and not just finite-dimensional projections of it (e.g. like the Random Walk Metropolis-Hastings algorithm), enjoys a dimension-independent mixing rate and, owing to the inherent scalability of neural networks to its number of inputs, is suitable for applications with high-dimensional state spaces.
    \item We propose a suitable likelihood for Bayesian Reinforcement Learning (BRL) for inferring the unknown continuous state value function that best describes an observed state-action data sequence. Theorem \ref{thm:bound_likelihood_derivative} and Lemma \ref{lemma:likelihood} justify the use of this likelihood with Gaussian prior measures on function spaces, and with our proposed neural network prior. This likelihood is also potentially of interest to the machine learning community in its own right.
    \item We apply Hilbert space MCMC methods to infer the unknown optimal value function in two continuous state control problems, using both our new prior and likelihood function. These exercises motivate the need for NN function priors that are, unlike a canonical orthogonal basis prior for that domain, scalable in the domain dimension, and for the first time demonstrates dimension-independent mixing of MCMC for Bayesian Inverse Reinforcement Learning.
\end{itemize}

The rest of this paper is organised as follows: In Section \ref{function_space_MCMC} we introduce the general inference problem, describe the canonical orthogonal basis for functions on $\mathbb R^d$, describe MCMC methods on an infinite-dimensional Hilbert space including their construction and the assumptions under which these methods are well-defined. Section \ref{sec:NN_priors} introduces the trace-class neural network prior and states one of our main theoretical results, showing that the proposed prior satisfies the necessary assumptions to be used with a Hilbert space MCMC algorithm. In Section \ref{RL} we formulate the Bayesian Reinforcement Learning (BRL) problem and introduce the likelihood to be used for inferring continuous state value functions from state-action data. We then show that the likelihood satisfies the assumptions needed to be admissible in a Hilbert space MCMC setting. Finally, Section \ref{numerics} provides numerical results for the proposed prior and the likelihood for different control problems. Proofs can be found in the appendix.

\subsection{Notation}
We use curly letters ($\mathcal X$ and $\mathcal A$) for spaces and sets. Subscripts denote both the temporal and spatial variables, but it will be clear from the context which one is being referred to. $\Phi$ denotes the Gaussian cumulative distribution function (cdf), $\phi$ the Gaussian probability density function (pdf). $\varphi$ is used for basis functions, $\zeta$ denotes an activation function. The likelihood function we will write as $\mathcal L$, the log-likelihood as $\ell$, and $T$ is the number of data points used in the likelihood. $\ell^2$ will also denote the space of square-summable sequences. The space of square-integrable functions, with respect to the Lebesgue measure, from $\mathcal X\subseteq\mathbb R^d$ to $\mathbb R$ is denoted $L^2(\mathcal X,\mathbb R)$ or simply $L^2$. For the control problem, $\mathcal T$ denotes the deterministic state dynamics, mapping a state-action pair $(x,a)$ to the next state $x'$. The value function is denoted with the letter $v$.

\section{Problem
Formulation\label{function_space_MCMC}}

The objective is to sample from a target distribution $\mu$ defined over an infinite-dimensional separable Hilbert space. The targets of interest in this work are Bayesian posterior distributions arising from a Gaussian prior measure $\mu_0$ and a likelihood which can be evaluated point wise. One such likelihood is the Gaussian likelihood that arises from observations of a solution to a PDE with additive Gaussian noise given in Section \ref{subsec:gwf1}, which is a standard likelihood in the Bayesian Inverse problems literature \citep{stuart2010inverse}. The other likelihood we will work with is one for continuous state control problems which is introduced in Section \ref{RL}. In what follows, we will assume that the posterior has a density with respect to the prior, in which case the Radon-Nikodym derivative is well defined and is proportional to the likelihood. The posterior density with respect to the prior is given by $\frac{d\mu}{d\mu_0}(u)=\frac1Z\exp(\ell(y|u))$, where $y$ are observations, $\ell$ is the log-likelihood, and $Z=\int\exp(\ell(y|u))\mu_0(du)>0$ is the normalisation constant.

For an infinite-dimensional separable Hilbert space $\mathcal H$, say $\mathcal H=L^2(\mathcal X,\mathbb R)$ to frame the discussion in this section (and later in Section \ref{sec:NN_priors} the sequence space $\mathcal H=\ell^2$), there exists an orthonormal basis $\{\varphi_i\}_{i=1}^\infty$ such that any element $u\in\mathcal H$ can be obtained as the limit
$u(x)=\lim_{N\rightarrow\infty}\sum_{i=1}^Na_i\varphi_i(x)$, where $a_i=\langle u,\varphi_i\rangle_{\mathcal H}$ with $\langle u,\varphi_i\rangle_{\mathcal H}$ denoting the inner product on $\mathcal H$. Let the prior $\mu_0=\mathcal N(0,\mathcal C)$ be a Gaussian measure on $\mathcal H$. If the operator $\mathcal C$ is trace-class with orthonormal eigenvalue-eigenfunction pairs $(\lambda_i^2,\varphi_i(x))$, $i=1,2,\dots$, one can sample from $\mu_0$ by sampling a sequence of $\xi_i\sim\mathcal N(0,\lambda_i^2)$ and by then defining 
\begin{align}\label{KL}
    u(x)=\sum_{i=1}^\infty \xi_i\varphi_i(x).
\end{align}
The sum defines $u(x)\in\mathcal H$ almost surely and is the Karhunen-Lo\'eve (KL) expansion \citep{gine2016mathematical}. One may thus think of a sample from the Gaussian measure as the sum of a sequence of $1$-dimensional Gaussians with summable variances. This allows us to truncate the series expansion such that we have $N$ active terms, with the remainder, or approximation error, tending to zero as $N$ increases:
\begin{align*}
    u(x)&=\sum_{i=1}^N\xi_i\varphi_i(x)+\sum_{i=N+1}^\infty \xi_i\varphi_i(x)\\
    \lVert u(x)-\sum_{i=1}^N\xi_i\varphi_i(x)\rVert^2&=\sum_{i=N+1}^\infty \lVert\xi_i\varphi_i(x)\rVert^2=\sum_{i=N+1}^\infty \xi_i^2<\infty\quad\mathrm{a.s.}.
\end{align*}
Other more elaborate truncation schemes are discussed in \cite{cotter2013mcmc}, but we will focus on a fixed number of terms for computational and notational convenience. For some applications, $\varphi_i$ for large $i$ can be interpreted as high-oscillating functions which may not be discernible by the observation operator, see the example in Section \ref{subsec:gwf1} or Figure \ref{fig:KL_prior}, where the large $i$ coefficients are responsible for the oscillating function in the left panel, and forced to $0$ on the right. Note that, given some $u\in \mathcal H$, we can let $u'$ be $u$ with $i$-th component set to $0$, i.e. $u'=u-\langle u,\varphi_i\rangle \varphi_i$. It follows from Assumption \ref{assumption_likelihood2} (stated later in the manuscript) that $\lim_{i \rightarrow \infty} \ell(y|u')=\ell(y|u)$, for any $u\in \mathcal H$. Following the approach of \cite[Theorem 4.6]{stuart2010inverse} 
this closeness of the likelihoods $\ell(y|u)$ and $\ell(y|u-\langle u,\varphi_i\rangle \varphi_i)$ translates to closeness of the corresponding posteriors. 

We emphasise that the above discussion holds not only for the space $\mathcal H=L^2(\mathcal X,\mathbb R)$, which is predominantly how it is applied in \cite{beskos2008mcmc,cotter2013mcmc,beskos2017geometric}, but also for $\mathcal H=\ell^2$ (with the only change being the choice of the orthonormal basis), which will be of particular importance in this paper. In infinite-dimensional spaces, one has to be careful to ensure the posterior is well defined, see \cite{stuart2010inverse} for a discussion on Gaussian priors and likelihoods given through possibly non-linear mappings, observed in Gaussian noise. We will work with the following assumptions, which we prove are satisfied for the likelihood defined in Section \ref{RL}. 
\begin{enumerate}
    \item\label{assumption_Hilbert_space} $\mu_0$ is a Gaussian prior defined on a separable Hilbert space $\mathcal H$, with a trace-class covariance operator $\mathcal C$, that is, the eigenvalues $\lambda_i^2$ corresponding to the eigenfunctions $\varphi_i$ satisfy $\sum_i\lambda_i^2<\infty$;
    \item The posterior is well-defined, i.e. the integral of the likelihood with respect to the prior is positive and finite.
\end{enumerate}

\subsection{A canonical approximation for functions on $\mathbb{R}^d$}
Consider a $d$-dimensional hypercube $\mathcal X=[0,1]^d$, the Hilbert space $\mathcal H=L^2(\mathcal X,\mathbb R)$, and a Gaussian prior measure $\mu_0$ on $\mathcal H$. A Bayesian approach entails choosing the covariance matrix $\mathcal C$ for the Gaussian prior $\mu_0$, and we discuss a standard choice below. If the problem requires it, as in Section \ref{subsec:gwf1} where a PDE is solved, it is possible to choose $\mathcal C$ such that the samples are almost surely differentiable. 

Given eigenvalues $\lambda_i$ and basis functions $\varphi_i$ for a $1$-dimensional function, one approach to scale this basis up to a $d$-dimensional domain is by taking a tensor product of the basis, see e.g. \cite{iserles2009high} for the multivariate Fourier basis, or \cite{wojtaszczyk1997mathematical} for Wavelets and other basis expansions. For the KL expansion, we thus get, for a multi-index $k=(k_1,\dots,k_d)=k_{1:d}$ with $k_i=1,\dots,N$,
\begin{align}
    u(x)=\sum_k \xi_k\varphi_k(x)=\sum_{k_1=1}^N\cdots\sum_{k_d=1}^N \left[\xi_{k_{1:d}}\prod_{j=1}^d\varphi_{k_j}(x_j)\right],
\end{align}
where $\xi_{k_{1:d}}\sim\mathcal N(0,\lambda_{k_{1:d}}^2)$ with $\lambda_{k_{1:d}}$ being a function of the respective eigenvalues $\lambda_{k_i}$ capturing the correlation between dimensions. In total, there are $N^d$ active terms, that is, the complexity is exponential in the dimension $d$. This will be computationally prohibitively expensive, even for moderately small $d$. An approximation-theoretic argument for the exponential scaling has been made by \cite{agapiou2021rates}, who showed that a Sobolev function $u$ on a $d$-dimensional domain with smoothness $\alpha$ can be approximated in $L^2$ within $\epsilon$ error using $N_\epsilon$ basis terms, where $N_\epsilon\propto\epsilon^{-d/\alpha}$.
To circumvent the exponential growth of terms in the domain dimension, one could employ the following simplifications with only mixed partials up to order two \citep{sobol1993sensitivity}
\begin{align}\label{vf_approx}
    u(x)\approx \sum_{i=1}^d u_i(x_i)+\sum_{i=1}^d\sum_{j=i+1}^d u_{i,j}(x_i,x_j),
\end{align}
with $dN+\frac{d(d-1)}{2}N^2$ coefficients to be estimated, thus still achieving a significant reduction compared to $N^d$ terms before. In our numerical work, this approximation is an obvious candidate to contrast against.

With the approximation \eqref{vf_approx} in mind, one restricts oneself to the prior on finitely many random functions $u_i$ and $u_{i,j}$, each of which themselves is sampled from a Gaussian measure $\mathcal N(0,\mathcal C_1)$, or  $\mathcal N(0,\mathcal  C_2)$, respectively. One identifies each of these functions with their Karhunen-Lo\'eve expansion
\begin{align}\label{infinite_KL}
    u_i(x_i)=\sum_{k=1}^\infty \xi_{i,k}\varphi_k(x_i),\quad u_{i,j}(x_{i,j})=\sum_{k=1}^\infty \xi_{i,j,k}\psi_k(x_i,x_j)
\end{align}
where the $\varphi_k$ and $\psi_k$ are the eigenfunctions corresponding to the eigenvalues $\lambda_{\varphi,k}^2$ and $\lambda_{\psi,k}^2$, respectively. The $\xi_{i,k}$ and $\xi_{i,j,k}$ are independent normal random variables $\xi_{i,k}\sim\mathcal N(0,\lambda_{\varphi,k}^2)$ and $\xi_{i,j,k}\sim\mathcal N(0,\lambda_{\psi,k}^2)$. As before one requires the covariance operators to be \emph{trace-class}, and truncates the expansion \eqref{infinite_KL} after a finite number of term. 

The numerical experiments using the KL function space prior in this paper are based on the following Fourier basis functions, $\varphi_k$ defined on $[0,1]$, $\psi_k=\psi_{k_1,k_2}$ defined on $[0,1]^2$ and indexed by a double index $k=(k_1,k_2)\in\mathbb N\times\mathbb N$:
\begin{align}\begin{split}\label{eq:fourier_basis}
    \varphi_{2k}(x_i)= \sin(2\pi kx_i) &\qquad\varphi_{2k+1}(x_i)= \cos(2\pi kx_i)\\
    \psi_{2k_1,2k_2}(x_i,x_j)= \sin(2\pi k_1x_i)\sin(2\pi k_2x_j)&\qquad\psi_{2k_1+1,2k_2}(x_i,x_j)= \cos(2\pi k_1x_i)\sin(2\pi k_2x_j)\\
    \psi_{2k_1,2k_2+1}(x_i,x_j)= \sin(2\pi k_1x_i)\cos(2\pi k_2x_j)&\qquad\psi_{2k_1+1,2k_2+1}(x_i,x_j)= \cos(2\pi k_1x_i)\cos(2\pi k_2x_j),\end{split}
\end{align}
for $i\neq j$, with corresponding eigenvalues
\begin{align}\label{eq:fourier_eigenvalues}
    \lambda^2_{\varphi,2k}&=\lambda^2_{\varphi,2k+1}=\frac{1}{k^\alpha}\\
    \lambda^2_{\psi,2k_1,2k_2}&=\lambda^2_{\psi,2k_1+1,2k_2}=\lambda^2_{\psi,2k_1,2k_2+1}=\lambda^2_{\psi,2k_1+1,2k_2+1}=\frac{1}{\left(\sqrt{k_1^2+k_2^2}\right)^\alpha}.
\end{align}
See Figure \ref{fig:KL_prior} for some representative draws from this prior, which is a modification from the prior used in Section 4.2 of Beskos et al. (2017). The covariance operator is of the form $-\Delta^{-\alpha}$ where $\Delta$ denotes the Laplacian, and we allow both Dirichlet (e.g. $\varphi_{2k}(0)=\varphi_{2k}(1)=0$) and Neumann boundary conditions (e.g. $\varphi'_{2k+1}(0)=\varphi'_{2k+1}(1)=0$), with opposing sides of the square $[0,1]^2$ satisfying the same boundary conditions.

\begin{figure}[htp]
\begin{center}
	\includegraphics[width=\linewidth]{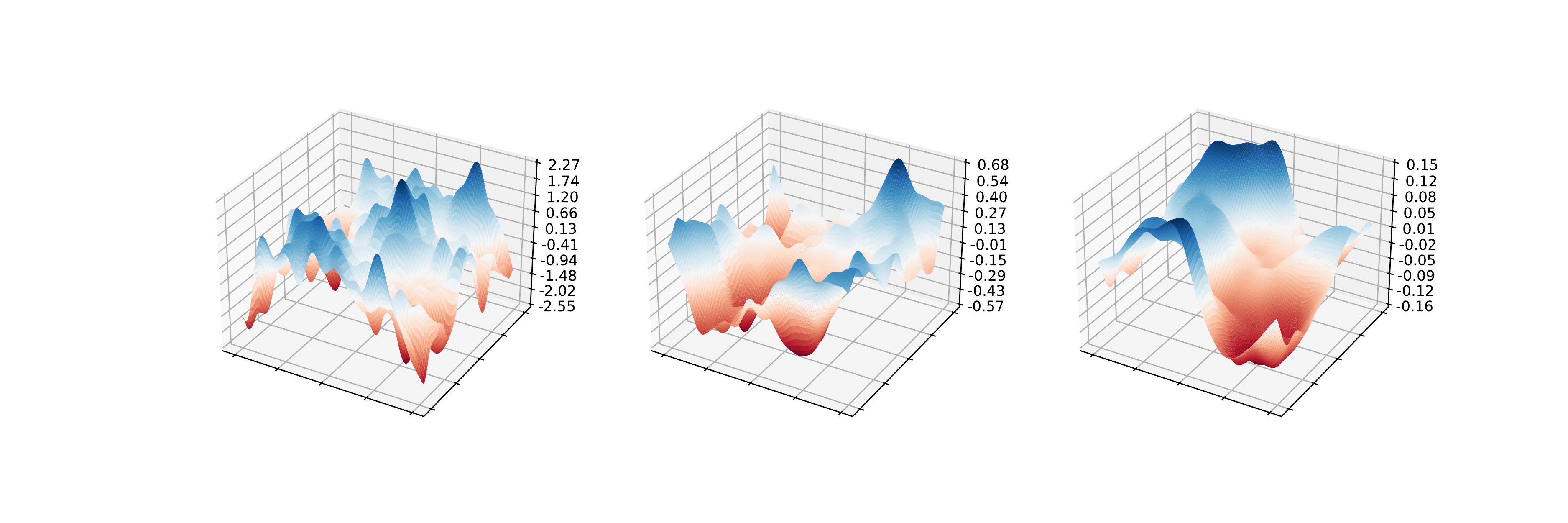}
	\caption{Three samples from the Karhunen-Lo\'eve prior; the basis functions are the two-dimensional Fourier functions. In ascending order from left to right we set $\alpha\in\{1.001,1.5,2\}$ with the eigenvalues scaling as $\lambda^2_k\propto1/(k_1^2+k_2^2)^\alpha$, for the double index $k=(k_1,k_2)$. The tuning parameter $\alpha$ controls the smoothness of the samples.}
	\label{fig:KL_prior}
\end{center}
\end{figure} 

Section \ref{sec:NN_priors} will introduce a prior which scales favourably with the domain-dimension as it does not require pre-defining an orthogonal basis.

\subsection{Metropolis-Hastings algorithms on Hilbert spaces\label{MH_pCN}}
This section recapitulates how to define `sensible' Metropolis-Hastings Markov chain Monte Carlo algorithms for inference over the $\xi_i$ in \eqref{KL}. Using Markov chains is an established approach to sample from distributions on finite-dimensional state spaces (see \cite{brooks2011handbook} for an overview of MCMC methods) and our emphasis here is  to review algorithms which can theoretically deal with arbitrarily many basis coefficients, without having to be re-tuned to avoid the usual problem of the acceptance probability degenerating as one includes more coefficients. This property, known as \emph{stability under mesh-refinement}, is not satisfied by the popular Random Walk Metropolis-Hastings Algorithm (RWMH, \cite{hastings1970monte}), or by the Metropolis Adjusted Langevin Algorithm (MALA, \cite{roberts1996exponential}).

Two algorithms which are both dimension-independent are the preconditioned Crank-Nicolson (pCN) and the preconditioned Crank-Nicolson Langevin (pCNL) algorithms, the former introduced as early as \cite{neal1998regression} and both derived and discussed in \cite{cotter2013mcmc}. Motivated by the idea of increasing dimensions translating to evaluating a function on a finer mesh, we will also refer to the dimension-independence of these algorithms as stability under mesh-refinement. Both algorithms can be seen as a discretisation of the following stochastic partial differential equation:
\begin{align}\label{eq:SPDE}
    \frac{du}{ds}=-\mathcal K(\mathcal C^{-1}u-\gamma D\ell(u))+\sqrt{2\mathcal K}\frac{dB}{ds},
\end{align}
where $D\ell$ is the Fr\'echet derivative of the log-likelihood\footnote{Note that we use the log-likelihood $\ell$ rather than the potential $\Phi=-\ell$ as the authors of \cite{cotter2013mcmc}.}, $\mathcal K$ is a preconditioner, $C$ is the covariance operator of the Gaussian prior measure, $B$ is a Brownian motion, and $\gamma$ a tuning parameter: if $\gamma=0$, the invariant distribution of \eqref{eq:SPDE} is the prior $\mu_0$, and for $\gamma=1$ the invariant distribution is the posterior $\mu$. With the choice $\mathcal K=\mathcal C$ (the preconditioned case, such that the dynamics are scaled to the prior variances), discretising \eqref{eq:SPDE} using a Crank-Nicolson scheme results in pCN (for $\gamma=0$) and pCNL (for $\gamma=1$). The resulting discretisations can be simplified to
\begin{align}\label{pCN}
    v &=\sqrt{1-\beta^2}u+\beta w,\quad w\sim\mathcal N(0,\mathcal C),\qquad\text{(pCN)}\\
    \label{pCNL}
    v &=\frac{1}{2+\delta}\left[(2-\delta)u+2\delta\mathcal C\mathcal D\ell(u)+\sqrt{8\delta}w\right],\quad w\sim\mathcal N(0,\mathcal C),\qquad\text{(pCNL)}
\end{align}
for step sizes $\beta\in(0,1]$ and $\delta\in(0,2)$, respectively. Note that due to the discretisation scheme used, pCN is prior-reversible, and using it as a proposal in a Metropolis-Hastings sampler to target the posterior, the proposal is accepted with probability $\min\{1,\exp(-\ell(u)+\ell(v))\}$. If the pCNL dynamics are used as a proposal for a MH scheme, the acceptance probability is given by $\min\{1,\exp(\rho(u,v)-\rho(v,u))\}$ where
\begin{align*}
    \rho(u,v)=-\ell(u)-\frac12\langle v-u,\mathcal D\ell(u)\rangle-\frac\delta4\langle u+v,\mathcal D\ell(u)\rangle+\frac\delta4\lVert \sqrt{\mathcal C}\mathcal D\ell(u)\rVert^2.
\end{align*}
Both pCN and pCNL are such that, for an uninformative likelihood, all moves are accepted. In practice, the likelihood Assumptions \ref{assumption_likelihood1} and \ref{assumption_likelihood2} ensure that, unlike RWMH or MALA, neither pCN nor pCNL require their step size $\beta$ or $\delta$ to go to $0$ as one includes more coefficients in the KL expansions \citep{cotter2013mcmc}. 

To conclude this section, we state the assumptions under which both pCN \cite[Thm 6.2]{cotter2013mcmc} and pCNL are well defined. Assumptions \ref{assumption_likelihood1} and \ref{assumption_likelihood2} \citep[Assumptions 6.1]{cotter2013mcmc} are needed for both pCN and pCNL, while \ref{assumption_pCNL} is only required for pCNL \citep{beskos2017geometric}:
\begin{enumerate}
\setcounter{enumi}{2}
    \item\label{assumption_likelihood1} There exist constants $K>0$, $p>0$ such that $0\leq- \ell(y|u)<K(1+\lVert u\rVert_{\mathcal H}^p)$ holds for all $u\in\mathcal H$.
    \item\label{assumption_likelihood2} For all $r>0$, $\exists K(r)>0$ such that for all $u$, $v$ with $\max(\lVert u\rVert_{\mathcal H},\lVert v\rVert_{\mathcal H})<r$, we have $\lvert\ell(y|u)-\ell(y|v)\rvert\leq K(r)\lVert u-v\rVert_{\mathcal H}$.
    \item\label{assumption_pCNL} For all $u\in\mathcal H$, $\mathcal C \mathcal D\ell(u)\in \text{Im}(\mathcal C^{1/2})$, $\mu_0$-almost surely. That is, for any draw $u$ from the prior, the \emph{preconditioned} differential operator at $u$ is in the Cameron-Martin space of the prior with probability $1$.
\end{enumerate}

\section{Trace-Class Neural Network Priors\label{sec:NN_priors}}
The Gaussian prior on $\mathcal H=L^2(\mathcal X,\mathbb R)$ exploits the isometry between the function space $L^2(\mathcal X,\mathbb R)$ and the sequence space $\ell^2$ using the Karhunen-Lo\'eve expansion \citep{gine2016mathematical}, but the computational complexity of using a basis-expansion on a high-dimensional domain is unfeasible even when using approximate function representations such as in \cite{sobol1993sensitivity}. 

Neural networks have been shown to have excellent empirical performance in high-dimensional function regression tasks. Bayesian neural networks (BNNs), capitalising on this success, randomise the neural network architecture to yield Bayesian priors for functions. BNNs are popular as they empirically show good results, scale well in the dimension of the function's domain, and more ground is being made on the supporting theory, e.g. on their approximation quality, infinite-width behaviour etc \citep{hornik1991approximation, matthews2018gaussian}. A drawback of standard BNNs is currently the limited interpretability of the posterior distributions on the parameter space, as the distribution on each weight degenerates due to the scaling of the variance proportional to the number of nodes.

We now propose a prior for the parameters that define a neural network which will generate almost surely well-defined functions for an infinite-width neural network. This is achieved by parameterising the infinite width neural network using sequences in the Hilbert space $\mathcal H=\ell^2$, the space of square-summable real valued sequences, and endow it with a trace-class Gaussian prior. This then allows inference for such neural networks to be conducted using the dimension-independent MCMC methods discussed in Section \ref{MH_pCN}. 

Through the architecture of the neural network, the prior $\mu_0$ over the parameters implicitly defines a prior on the output function of the neural network. Under mild assumptions on the network architecture, and if $\mathcal X$ is compact, the output functions, which we denote as $v$, are $\mu_0$-almost surely square-integrable over $\mathcal X$, and the prior thus naturally defines a prior over $L^2(\mathcal X,\mathbb R)$ as well. Neural network priors are also more flexible compared to the Karhunen-Lo\'eve expansion of a Gaussian measure: one neither needs to specify a covariance operator and find its eigenfunctions, nor decide on a basis which is then used to define a Gaussian prior. By giving up the orthogonality of these eigenfunctions (which allow for a rich theoretical analysis), one gains on the performance side, see our numerical comparisons in Section \ref{sec:prior_comparison}. We coin the term \emph{trace-class neural network prior} (tcNN) to emphasise that the prior leads to a well-defined function space prior if the variances of all parameters are appropriately summable. The term is well-established for Gaussian measures, where these are called trace-class if the eigenvalues of the covariance operator are summable. 

Consider a $n$-layer feed-forward fully-connected neural network illustrated in Figure \ref{fig:NN}. The width of layer $l$ is $N^{(l)}$, the input to the first layer is $x\in[0,1]^d$, the domain of the function to be approximated, and let $v(x)=f_1^{(n+1)}(x)\in\mathbb R$ denote the network's output; for notational convenience we write $N^{(0)}=d$ and $N^{(n+1)}=1$. The network is described fully by the following set of real valued weights and biases,
\begin{align}\label{wb_all}
    w= \left\{w_{i,j}^{(l)}\right\}_{i=1,j=1,l=1}^{N^{(l)},N^{(l-1)},n+1},\qquad
    b= \left\{b_i^{(l)}\right\}_{i=1,l=1}^{N^{(l)},n+1},\qquad\theta=(w,b),
\end{align}
where we have summarised $w$ and $b$ as $\theta$. Given an \emph{activation function} $\zeta:\mathbb R\rightarrow\mathbb R$, the functions of each layer are 
\begin{align}\begin{split}\label{def:NNfunctions}
    f_i^{(1)}(x)&=b_i^{(1)}+\sum_{j=1}^dw_{i,j}^{(1)}x_j,\qquad i=1\dots N^{(1)}\\
    f_i^{(l)}(x)&=b_i^{(l)}+\sum_{j=1}^{N^{(l-1)}}w_{i,j}^{(l)}\zeta(f_j^{(l-1)}(x)),\qquad i=1\dots N^{(l)},~l=2\dots n\\
    v(x)=f_1^{(n+1)}(x)&=b_1^{(n+1)}+\sum_{j=1}^{N^{(n)}}w_{1,j}^{(n+1)}\zeta(f_j^{(n)}(x)).
    \end{split}
\end{align}

\begin{figure}
\begin{center}
\includegraphics[width=0.65\linewidth]{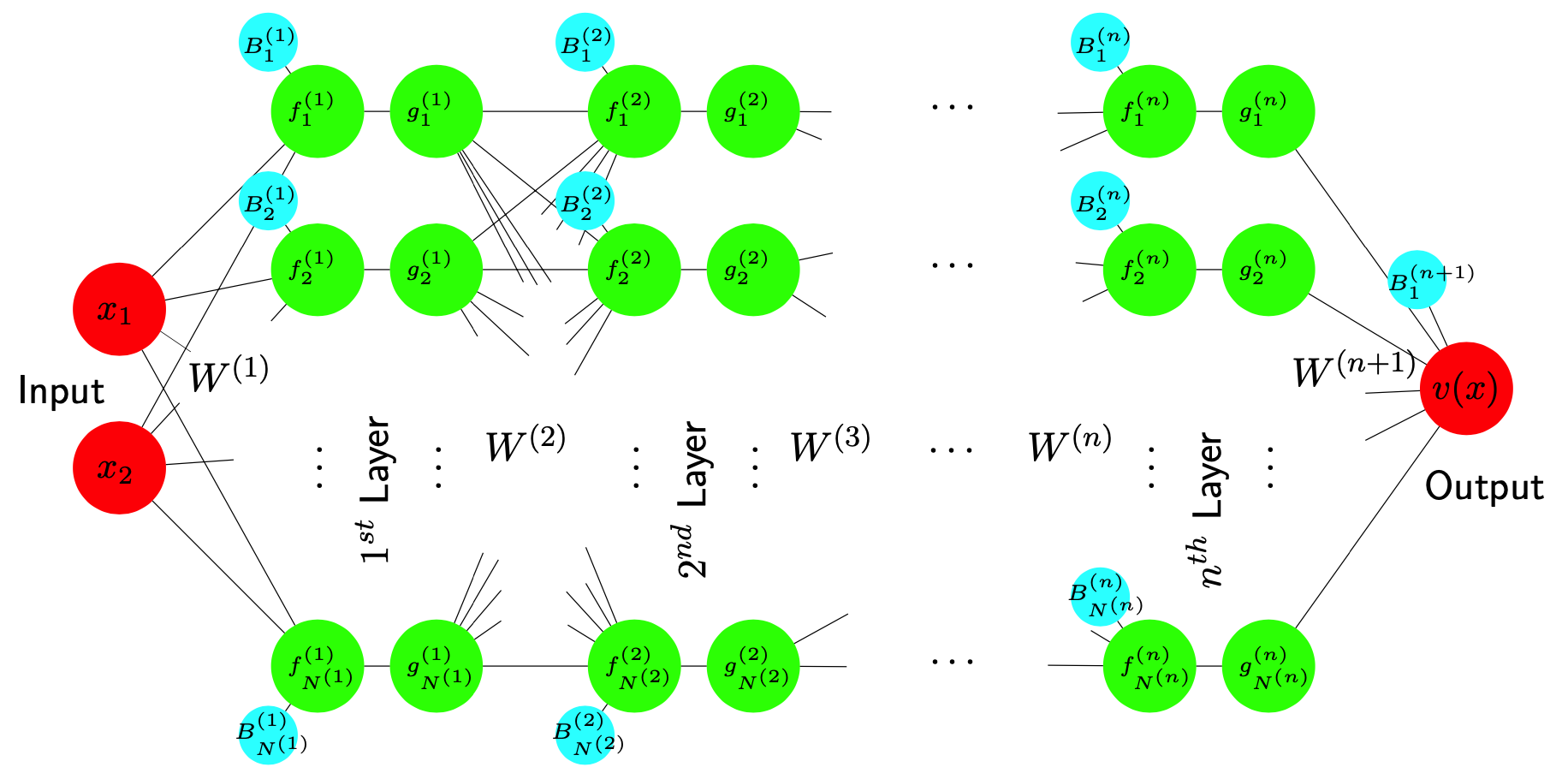}
\caption{A $n$-layer feed-forward neural network defining a function $v:\mathbb R^2\rightarrow\mathbb R$. Note that $g_i^{(l)}=\zeta(f_i^{(l)})$.}\label{fig:NN}
\end{center}
\end{figure}

The prior $\mu_0$ is now defined as follows: the individual weights and biases in each layer $l$ are independent and normally distributed, and we emphasise here that the novelty is to choose the variances not uniformly, but to decrease them as one moves into the tail nodes of each layer: 
\begin{align}\label{prior_def}
    W_{i,j}^{(1)}\sim\mathcal N\left(0,\frac{\sigma_{w^{(1)}}^2}{i^\alpha}\right),\quad
    W_{i,j}^{(l)}\sim\mathcal N\left(0,\frac{\sigma_{w^{(l)}}^2}{(ij)^\alpha}\right)~\text{for}~l=2\dots n+1,\quad
    B_{i}^{(l)}\sim\mathcal N\left(0,\frac{\sigma_{b^{(l)}}^2}{i^\alpha}\right),
\end{align}
where indices $i$, $j$, and $l$ are defined in \eqref{wb_all}, $\alpha>1$ is a fixed constant, and $\sigma_{w^{(l)}}^2>0$ for each $l$ (to avoid degeneracy of the prior). The reader should note that the prior is invariant with respect to permutation of the input variables, thus avoiding preferential treatment of any of the inputs.

The tuning parameter $\alpha$ controls how quickly the magnitude of the weights decrease in the direction of the tail nodes and is empirically seen to control how `variable' the sampled function is.
If $\alpha>1$ we refer to the prior as \emph{trace-class}, coining the term \emph{trace-class neural network priors}. If one believes that potentially many nodes with large weights are needed, one should choose $\alpha$ close to $1$. See Figure \ref{fig:NN_prior} for three representative draws from the neural network prior.  As the next theorem will show, this allows indeed to define an infinitely wide network by taking $N^{(l)}=\infty$, and the variances can be summarised in a diagonal covariance operator $\mathcal C$; this prior is well-defined on an infinite-dimensional Hilbert space (isometric to $\ell^2$), and can thus be used in the algorithms from Section \ref{function_space_MCMC}. In practice, one truncates the number of nodes within each layer as for the priors described before, or one may randomly switch nodes on and off similarly to the random truncation prior used in \cite{cotter2013mcmc}.

We now define the infinite width limit of the network. Given an infinite sequence of weights and biases for the first layer, distributed according to the prior \eqref{prior_def}, i.e. $\left\{\left(B_i^{(1)}, W_{i,1}^{(1)},\ldots,W_{i,d}^{(1)}\right): i\in \mathbb{N}\right\}$, all the functions of the first layer, $\{f_i^{(1)}: i \in \mathbb{N}\}$, are clearly well defined.
We define all the functions of the second layer corresponding to an infinite-width first layer to be the following almost sure limits, assuming they exist:
\begin{equation}
F_i^{(2)}(x)= \lim_{N^{1}\rightarrow \infty} 
B_i^{(2)}+\sum_{j=1}^{N^{(1)}}W_{i,j}^{(2)}\zeta(f_j^{(1)}(x)).
\label{eq:infNN_functions}
\end{equation}
Assuming the random functions $\{F_i^{(2)}: i \in \mathbb{N}\}$ are well defined, proceeding iteratively, all the functions $\{F_i^{(l)}: i \in \mathbb{N}\}$ of subsequent layers, $l=3,\ldots,n$ and the output layer $F_1^{(n+1)}$ can be defined similarly. The functions in each layer of a finite width network are denoted with lower case to clearly distinguish them from their infinite width versions. For the output layer,  
the finite network gives $v(x)$ or $f_{n+1}^{(1)}(x)$ while the infinite network gives $V(x)$ or $F_{n+1}^{(1)}(x)$.
In what follows, we will often write $v(x)=v_\theta(x)$ to emphasise the dependence of the function samples on the weights and biases. In order to simplify the presentation of the main results, we list a set of properties which will be shown to hold for our BNN prior:

\begin{figure}[htp]
\begin{center}
	\includegraphics[width=\linewidth]{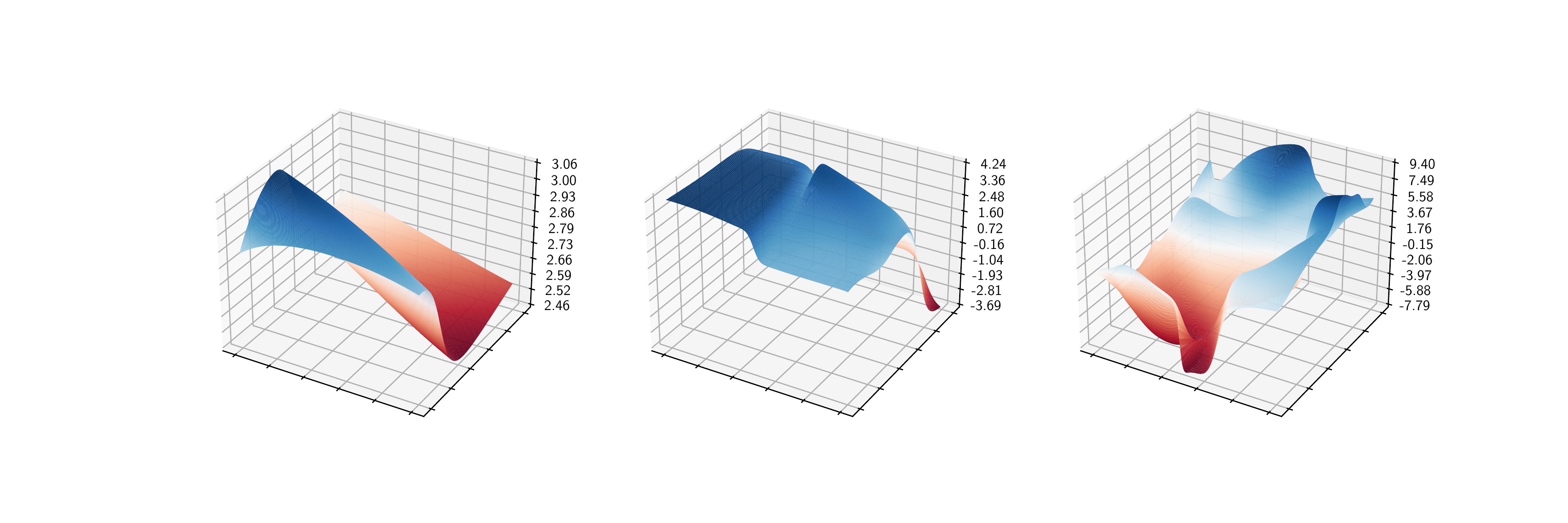}
	\caption{Three samples from the trace-class neural network prior defined on $[0,1]^2$, for a network with $3$ fully-connected layers and tanh activation functions; Tuning parameters are set to $\alpha=1.5$, $\sigma_{w^{(l)}}^2=\sigma_{b^{(l)}}^2=3$ for all $l=1\dots n+1$ (left), $\alpha=1.5$, $\sigma_{w^{(l)}}^2=\sigma_{b^{(l)}}^2=30$ (centre), $\alpha=1.0001$, $\sigma_{w^{(l)}}^2=\sigma_{b^{(l)}}^2=30$ (right). The tuning parameter $\alpha$ controls the complexity of the prior functions, the variances in the layers control the overall variance. Note the difference in the magnitudes on the $z$-axis.}
	\label{fig:NN_prior}
\end{center}
\end{figure} 

\begin{enumerate}
  \setcounter{enumi}{5}
    \item\label{assumption_finite_variance} $\forall x\in [0,1]^d$ one has $\lvert f_i^{(l)}(x)\rvert<\infty$ $\mu_0$-almost surely, $\mathbb Ef_i^{(l)}(x)=0$. Furthermore, there $\exists~\sigma_l^2$ such that $\mathbb E\left[(f_i^{(l)}(x))^2\right]<\sigma_l^2/i^\alpha$ for all $x\in[0,1]^d$. Here, both expectations are taken with respect to the prior on the parameters $\theta$ of the neural network. In particular this property holds for $v(x)=f_1^{(n+1)}(x)$. 
    \item\label{assumption_cts_vf} $\exists~c_l\geq0$ such that $\forall x,y\in [0,1]^d$, $\mathbb E\left[\left(f_i^{(l)}(x)-f_i^{(l)}(y)\right)^2\right]\leq c_l\lVert x-y\rVert^2/i^\alpha$, with the expectation again taken with respect to the prior. In particular this gives $\mathbb E\left[v(x)-v(y)\right]^2\leq c_{n+1}\lVert x-y\rVert^2$.
    \item\label{assumption_information_passes} $\mu_0$-almost surely, $v$ is differentiable almost everywhere.
\end{enumerate}
The first declared property ensures the output functions are appropriately finite in value and moments while the second property ensures a degree of smoothness. We now state a theorem which shows that the proposed prior satisfies the declared properties. To this end, we use an activation function\footnote{As will be clear from the proof of Theorem \ref{thm:prior}, one may use different activation functions at different layers, which will then all have to satisfy this assumption.} $\zeta:\mathbb R\rightarrow\mathbb R$ which satisfies the following condition, which will imply that $\lvert\zeta(x)\rvert<\lvert x\rvert$ for all $x\in\mathbb R$, and that $\zeta$ is differentiable almost everywhere, with the derivative being essentially bounded by $1$:
\begin{enumerate}
  \setcounter{enumi}{8}
    \item $\zeta$ is Lipschitz continuous with Lipschitz constant $1$ and $\zeta(0)=0$ \label{assumption_bounded_activation}.\footnote{The generalisation to arbitrary Lipschitz constants and the implication $\exists~c>0$ such that $\forall x\in\mathbb R$: $\lvert\zeta(x)\rvert<c\lvert x\rvert$ is straightforward.} 
\end{enumerate}

\begin{thm}\label{thm:prior}
Under Assumption \ref{assumption_bounded_activation}, the functions of the layers of the \emph{finite-width} neural network satisfy Properties \ref{assumption_finite_variance}, \ref{assumption_cts_vf}, and \ref{assumption_information_passes}. In addition, if $\alpha>1/2$, the functions on every layer of the infinite-width neural network (see \eqref{eq:infNN_functions}) exist almost surely and satisfy Properties \ref{assumption_finite_variance} and \ref{assumption_cts_vf}, when the functions $f_i^{(l)}(x)$ and $v(x)$ therein are replaced with  $F_i^{(l)}(x)$ and $V(x)$ defined as in \eqref{eq:infNN_functions}. In addition, if the prior is \emph{trace-class} (i.e. $\alpha>1$), Property \ref{assumption_Hilbert_space} is satisfied.
\end{thm}
The proof can be found in Appendix \ref{proof:thm:prior}.

\subsection{Identifiability Issues and Remedies}
It is well-known that the output function of a standard neural network does not depend on the labeling of functions within each layer. However, unlike a prior that has uniform variances within each layer, swapping nodes $f_i^{(l)}$ and $f_{i+1}^{(l)}$ (effectively by swapping their corresponding weights and biases) will lead from $\theta$ to a new $\theta'$ such that the prior weights change, and thus avoid the label-switching problem. To facilitate faster mixing by allowing jumps between these different configurations, we propose Algorithm \ref{alg:prior_swaps}, which can be found in Appendix \ref{app:alg}. The algorithm is well defined for finite widths networks, in which case the acceptance ratio is given by $a(\theta,\vartheta)=\mu_0(\vartheta)/\mu_0(\theta)$, but not for infinite width networks, see Lemma \ref{lem:node_swap} in the Supplementary Material; this exemplifies the extra care needed when defining MCMC moves in the infinite dimension setting. One remedy is not to swap all the weights of the two selected nodes but only blocks of them, however we did not pursue this approach.

\subsection{Illustrative Groundwater Flow Example}
Before moving on to more challenging examples, we present an illustrative example, and compare the performance of the neural network prior to the Gaussian prior presented previously. The example, taken from \cite{beskos2017geometric}\footnote{While we could not perfectly replicate their results, we aimed to stick as close to their results as possible.}, aims is to recover the permeability of an aquifer. The PDE $ -\nabla\cdot(\exp(u(x))\nabla p(x))=0$ connects the log-permeability $u$ of a porous medium to the hydraulic head function $p$ with the boundary conditions given by (for $x=(x_1,x_2)$)
\begin{align*}
    p(x)=x_1\quad\text{if }x_2=0,\qquad p(x)=1-x_1\quad\text{if }x_2=1,\qquad\frac{\partial p(x)}{\partial x_1}=0\quad\text{if }x_1\in\{0,1\}.
\end{align*}
To enforce the permeability to be positive, the prior is defined for the log-permeability $u(x)$. \\
We compare two priors. The first one is a trace-class neural network prior with $100$ nodes, Tanh activation function, and a four dimensional input space with the inputs $(x_1,x_2,\sin(x_1),\sin(x_2))$. We set the tuning parameters to $\alpha=1.001$, $\sigma_{w_1}^2=\sigma_{b_1}^2=100$, $\sigma_{w_2}^2=1/30$, and $\sigma_{b_1}^2=1/10$. The second prior is a Gaussian measure on $[0,1]^2$ with the following orthonormal basis and corresponding eigenvalues defined using double indices $i=(i_1,i_2)$:
\begin{align}
    \varphi_i(x)=2\cos\left(\pi(i_1+\frac12)x_1\right)\cos\left(\pi(i_2+\frac12)x_2\right),\quad
    \lambda_i^2=\frac{1}{(\pi^2\left((i_1+1/2)^2+(i_2+1/2)^2\right)^{1.1}}.
\end{align}
In the experiments, we truncated the basis expansion using $1\leq i_1,i_2\leq25$, which gives a similar number of parameters as we used in the neural network example.
The true $u^*$ is now defined using the same basis as $u^*(x)=\sum_iu^*_i\varphi_i(x)$ with $u^*_i=\lambda_i\sin\left((i_1-1/2)^2+(i_2-1/2)^2\right)\cdot\delta[1\leq i_1,i_2\leq10]$. The simulated data  are 33 noisy observations of the true hydraulic head function $p^*$ at various $x$ positions, $y=p^*(x)+\varepsilon$, where $\varepsilon\sim\mathcal N(0,0.01^2)$. The `true' head function $p^*(x)$ is obtained by solving the forward PDE on a $40\times40$ grid. We ran pCN using both priors, and solving the forward problem on a $20\times20$ grid. Both experiments used a similar number of iterations and stored $1000$ MCMC samples to obtain the mean estimates in Figure \ref{fig:GWF_results}. The results in Figure \ref{fig:GWF_results} are less insightful and interpretable than those we will see in the next subsection as the few observations we have are related to the target function only through the PDE. A better comparison between, and validation of, the different priors is through visual posterior predictive checks as shown in Figure \ref{fig:posterior_predictive}.

\begin{figure}[ht!]
\centering
\begin{subfigure}{.2\textwidth}
  \centering
  \includegraphics[width=\linewidth]{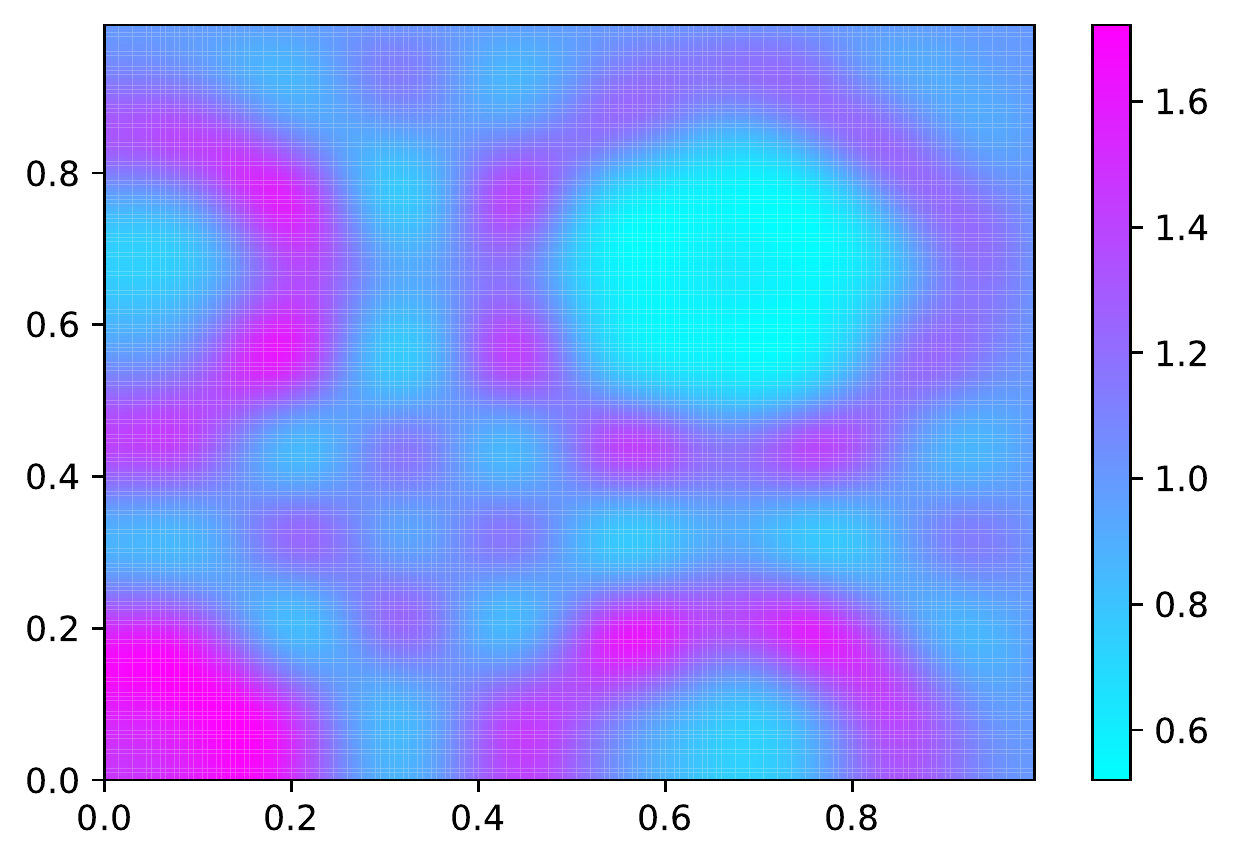}
  \caption{\centering True $\exp(u^*)$.}
\end{subfigure}
\qquad\qquad
\begin{subfigure}{.2\textwidth}
  \centering
  \includegraphics[width=\linewidth]{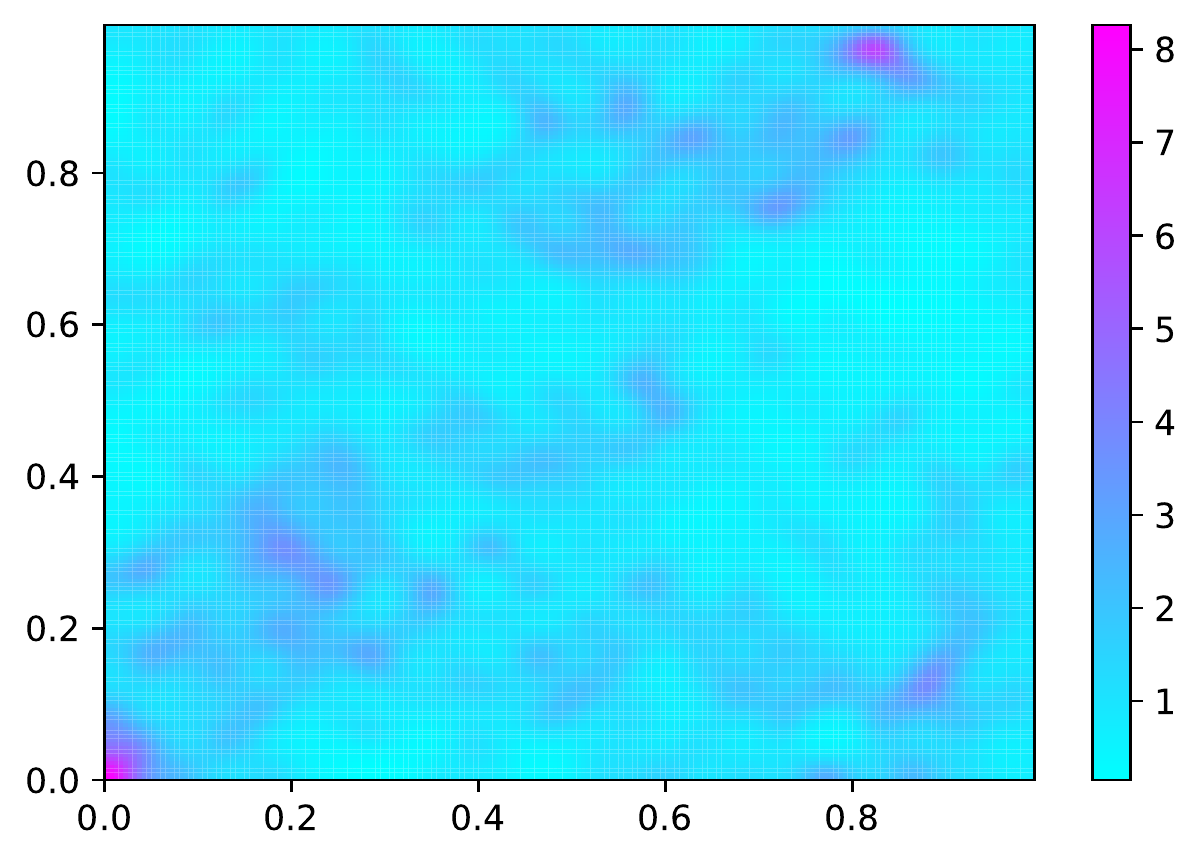}
  \caption{\centering Sample from KL-based prior.}
\end{subfigure}
\begin{subfigure}{.2\textwidth}
  \centering
  \includegraphics[width=\linewidth]{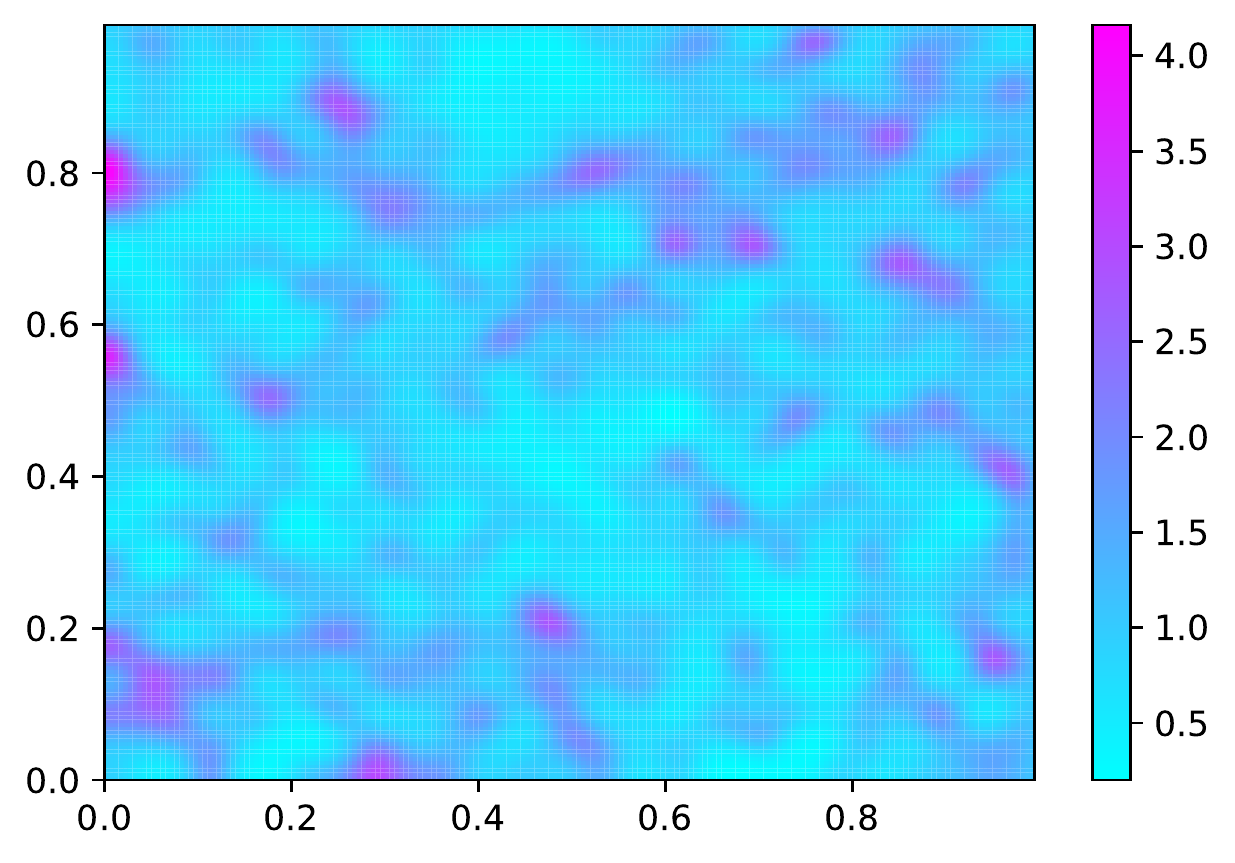}
  \caption{\centering Sample from KL posterior.}
\end{subfigure}
\begin{subfigure}{.2\textwidth}
  \centering
  \includegraphics[width=\linewidth]{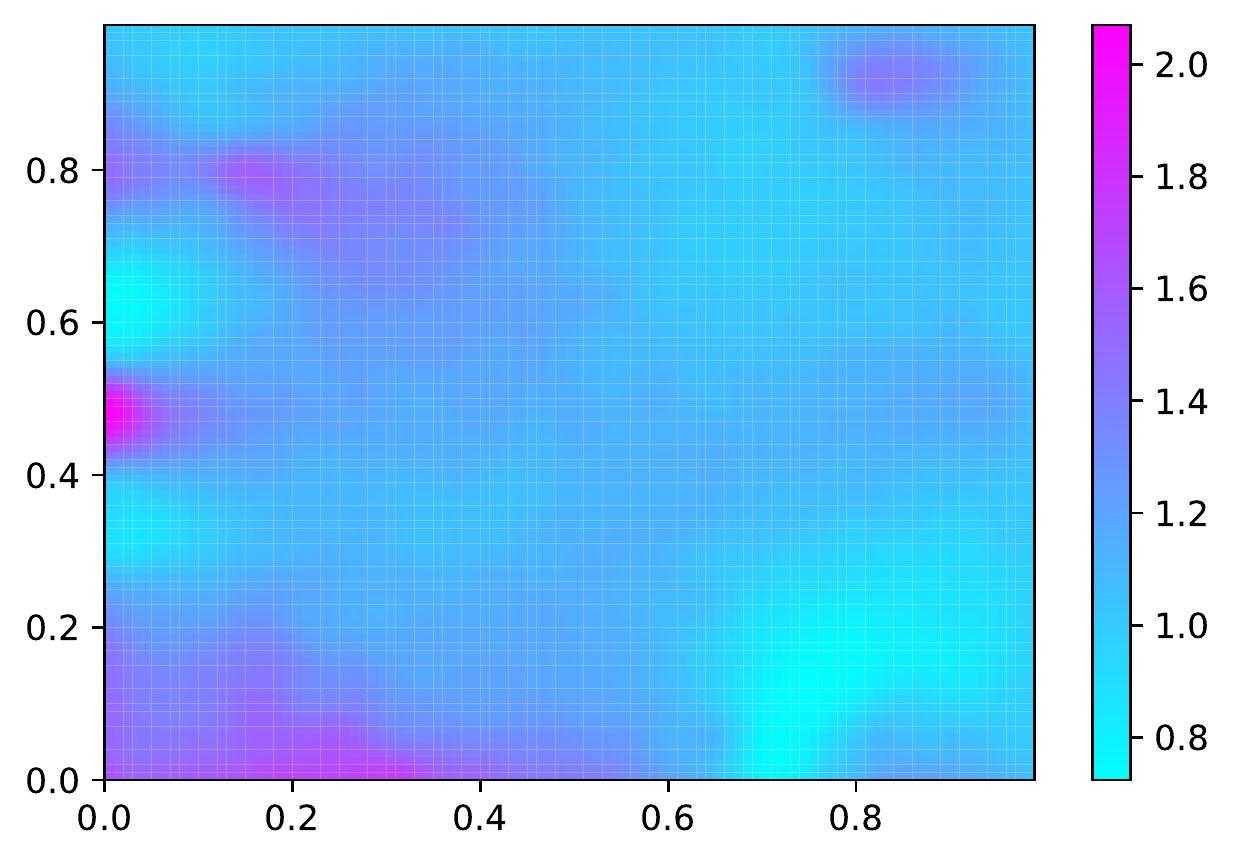}
  \caption{\centering KL posterior mean estimate.}
\end{subfigure}
\vskip\baselineskip
\begin{subfigure}{.2\textwidth}
  \centering
  \includegraphics[width=\linewidth]{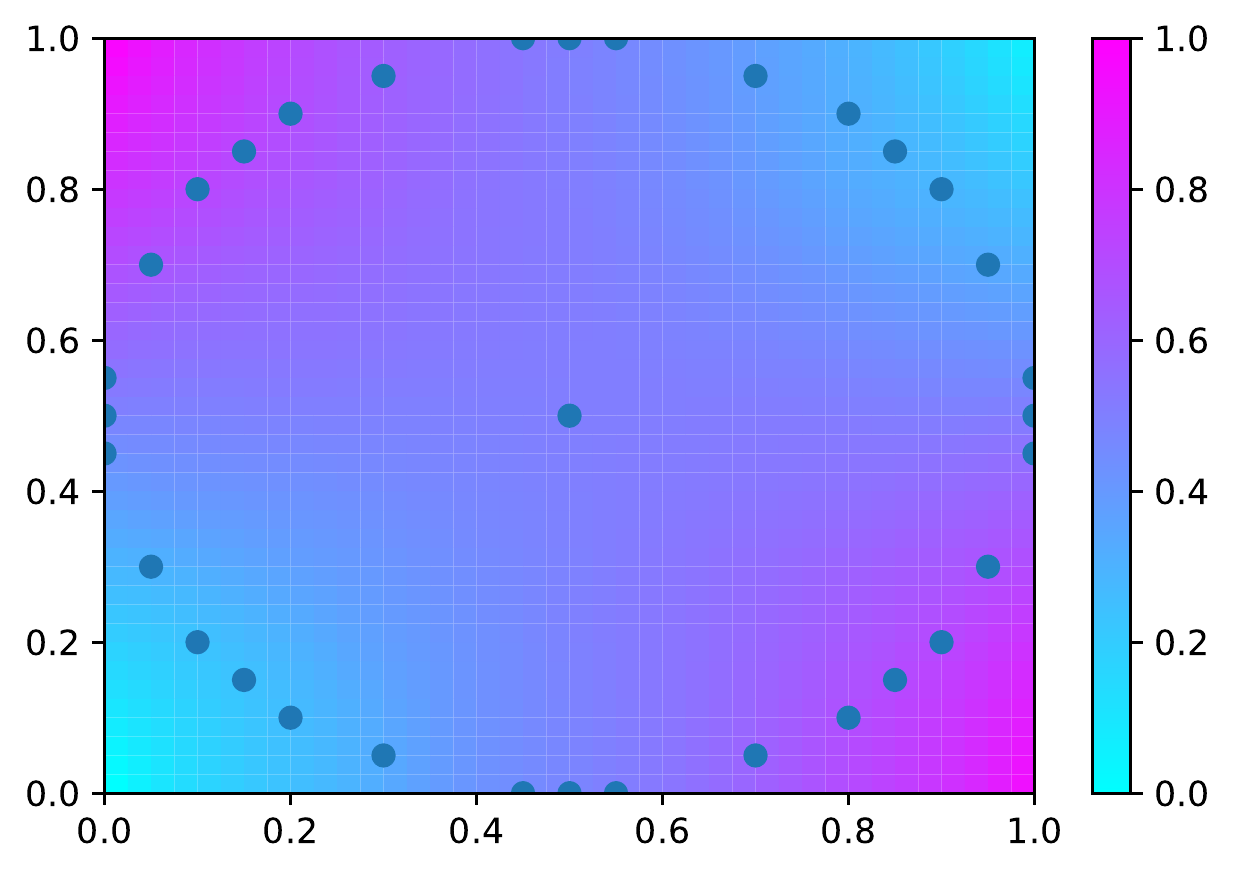}
  \caption{\centering True hydraulic head function $p^*$ and location of observations.}
\end{subfigure}
\qquad\qquad
\begin{subfigure}{.2\textwidth}
  \centering
  \includegraphics[width=\linewidth]{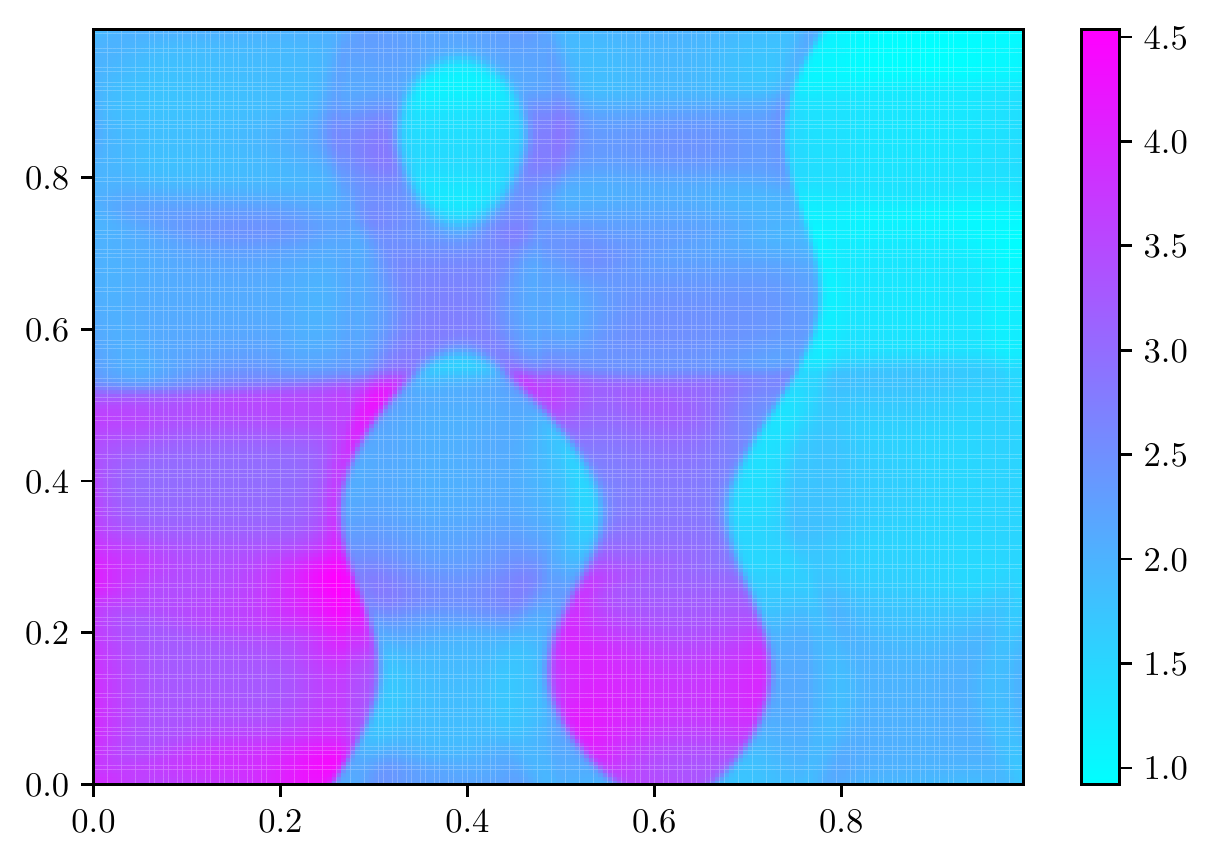}
  \caption{\centering Sample from tcNN prior.}
\end{subfigure}
\begin{subfigure}{.2\textwidth}
  \centering
  \includegraphics[width=\linewidth]{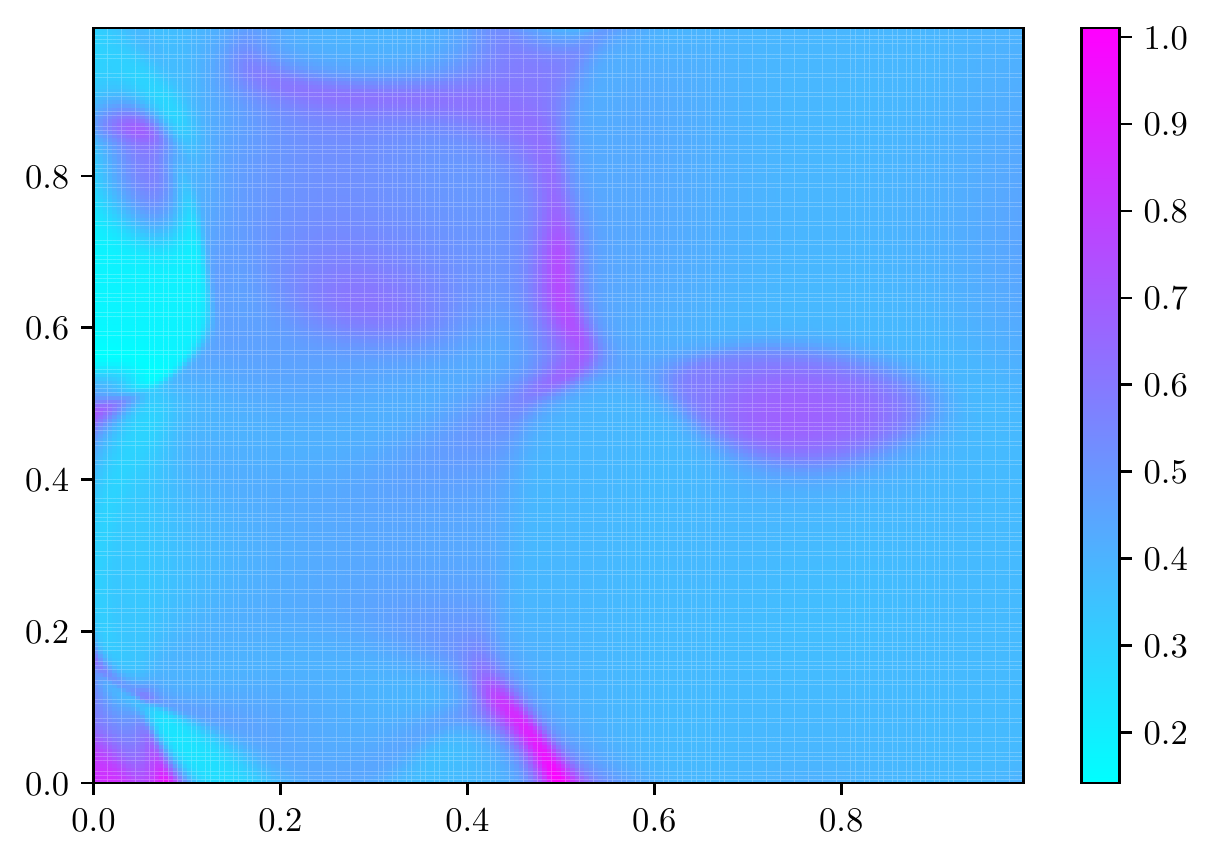}
  \caption{\centering Sample from tcNN posterior.}
\end{subfigure}
\begin{subfigure}{.2\textwidth}
  \centering
  \includegraphics[width=\linewidth]{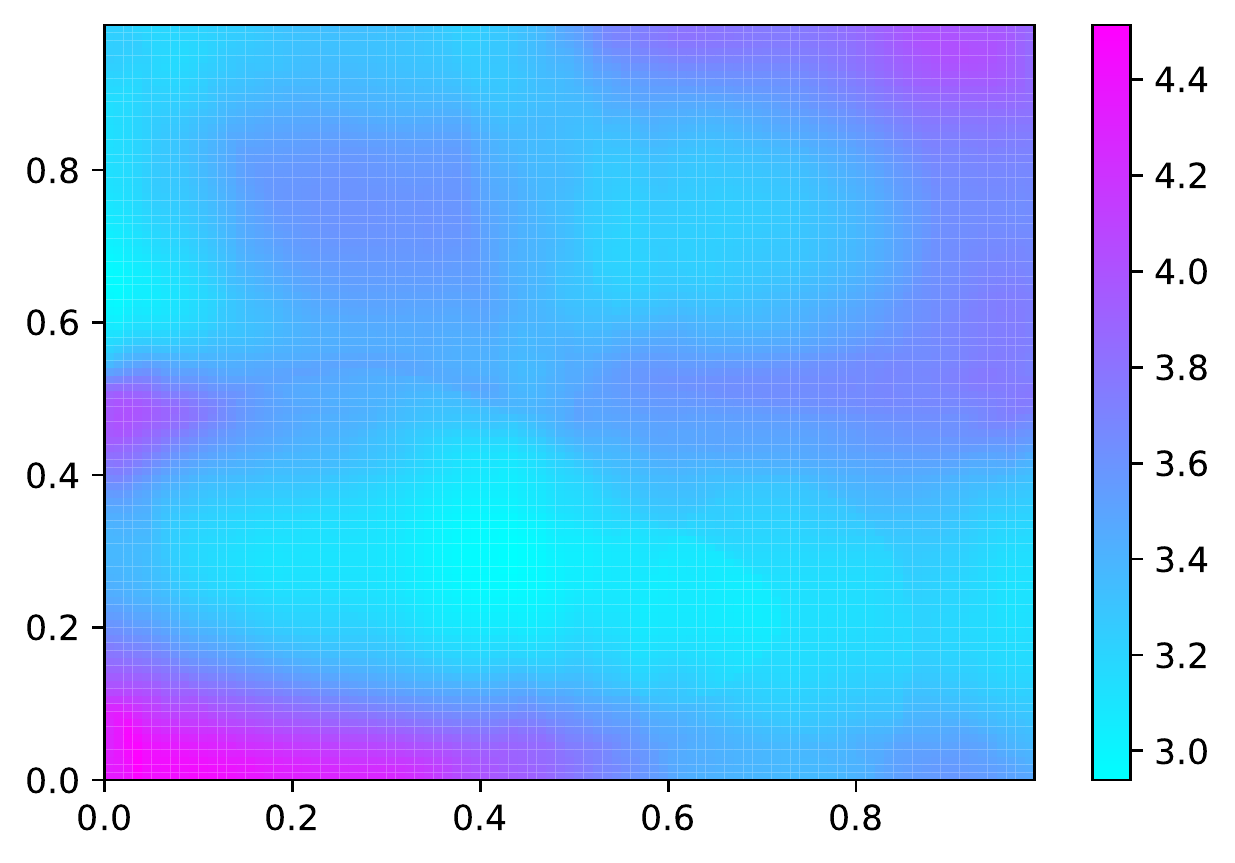}
  \caption{\centering tcNN posterior mean estimate.}
\end{subfigure}
\caption{\emph{Left column:} The true log-permeability (a), and its associated hydraulic head function (e) with the location of the $33$ observations. \emph{Top row, right columns:} A sample from the Gaussian prior (b), a sample from the associated posterior (c), and the posterior mean estimate obtained using pCN for the Gaussian prior (d). \emph{Bottom row, right columns:} A sample from the trace-class neural network prior (f), a sample from the associated posterior (g), and the posterior mean estimate obtained using pCN for the neural network prior (h).
}
\label{fig:GWF_results}
\end{figure}

\begin{figure}[ht]
\centering
\begin{subfigure}{.45\textwidth}
  \centering
  \includegraphics[width=0.75\linewidth]{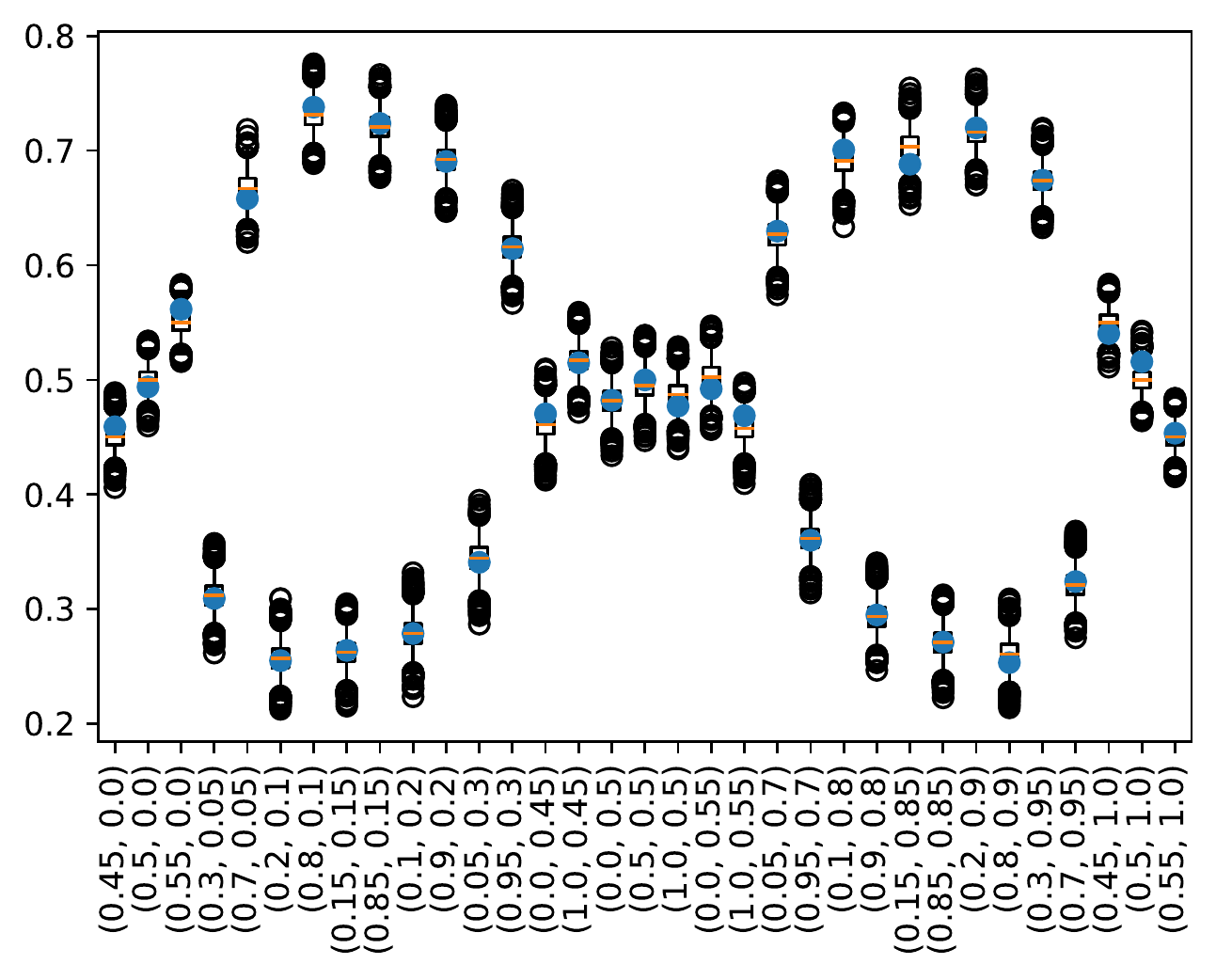}
  \caption{\centering KL posterior predictive.}
\end{subfigure}
\begin{subfigure}{.45\textwidth}
  \centering
  \includegraphics[width=0.75\linewidth]{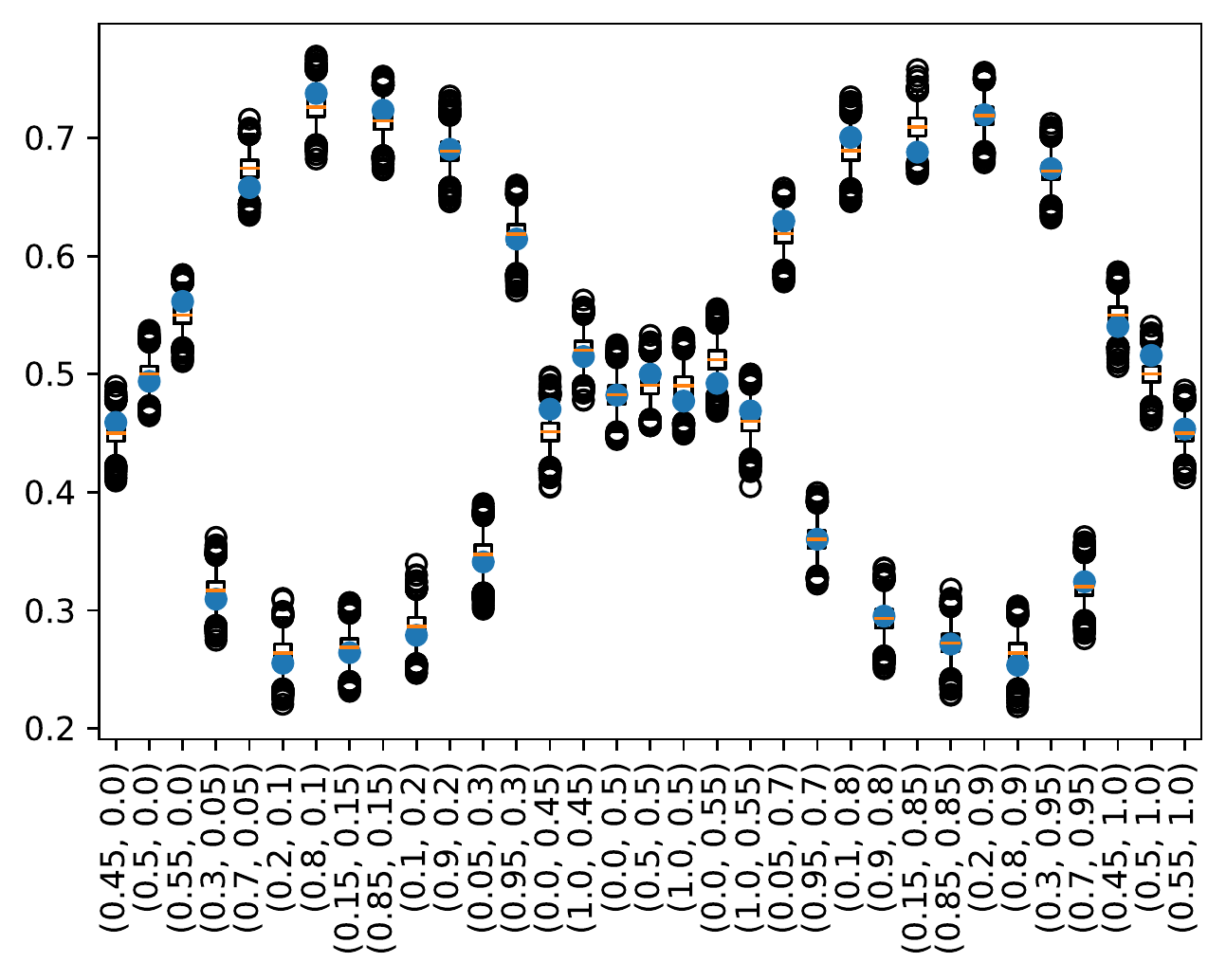}
  \caption{\centering tcNN posterior predictive.}
\end{subfigure}
\caption{Visual posterior predictive check for both the KL- and tcNN-based posteriors. The observed values at each of the $33$ observation locations (see figure \ref{fig:GWF_results}) are shown as a blue dot, the box plots are $100$ samples from the posterior predictive distribution \cite[Section 6.3]{gelman2013bayesian}. Both posteriors show similar predictive performance indicating that they arise from similarly well-suited priors.}
\label{fig:posterior_predictive}
\end{figure}

\subsection{Ability to approximate complicated functions\label{subsec:gwf1}}
To show that the trace-class neural network prior is able to visually recover relatively complicated functions, we define a function $u^*:[0,1]^2\rightarrow \mathbb R_{+}$, and observe this function on a $20\times20$ grid with independent Gaussian noise $\mathcal N(0,0.01^2)$. The true $u^*$ and the parameters of the prior we used here are the same one as in the example above. As Figure \ref{fig:approx_u} shows, the neural network prior is able to approximate the true $u^*$ when given many, in this example $400$, observations.

\begin{figure}[ht]
\centering
\begin{subfigure}{.45\textwidth}
  \centering
  \includegraphics[width=0.75\linewidth]{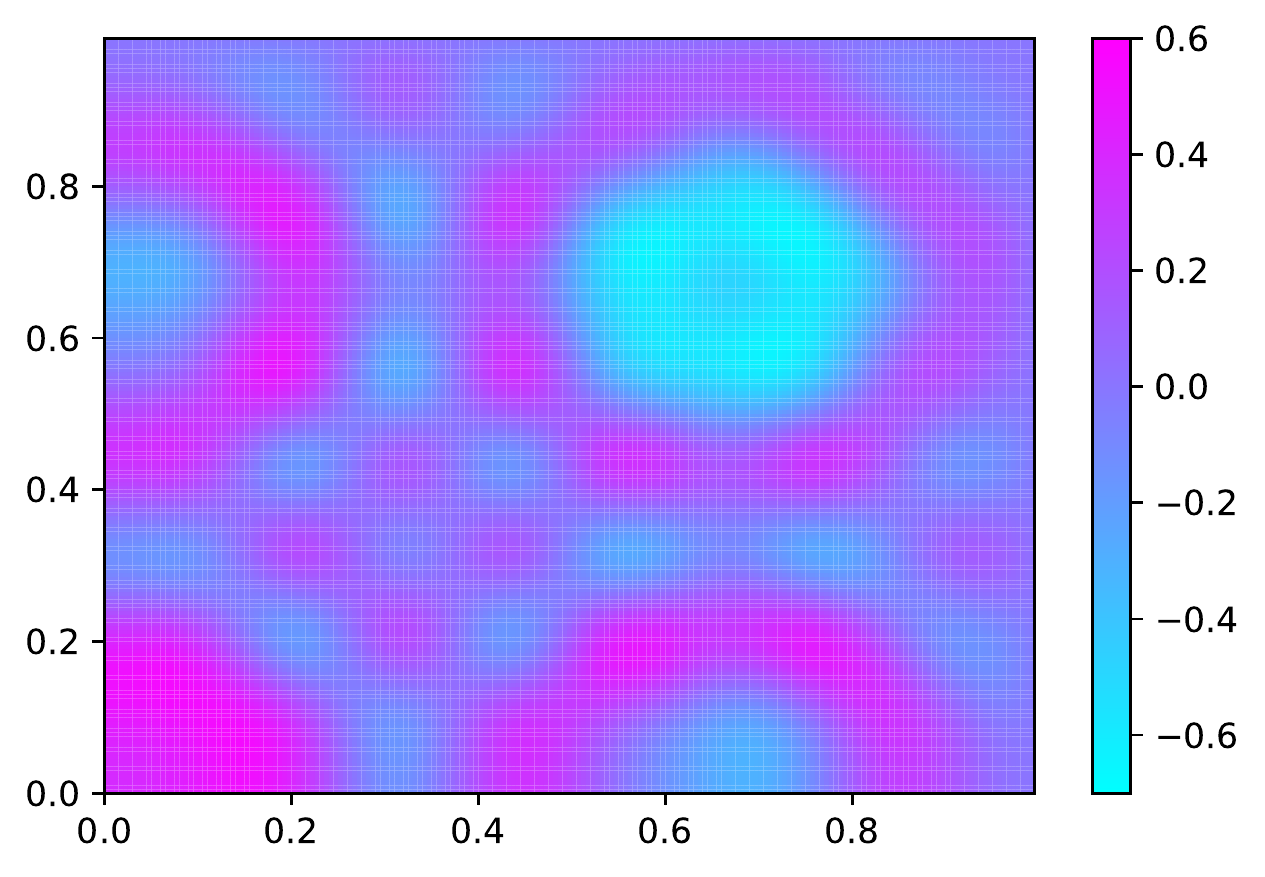}
  \caption{\centering True log-permeability $u^*$.}
\end{subfigure}
\begin{subfigure}{.45\textwidth}
  \centering
  \includegraphics[width=0.75\linewidth]{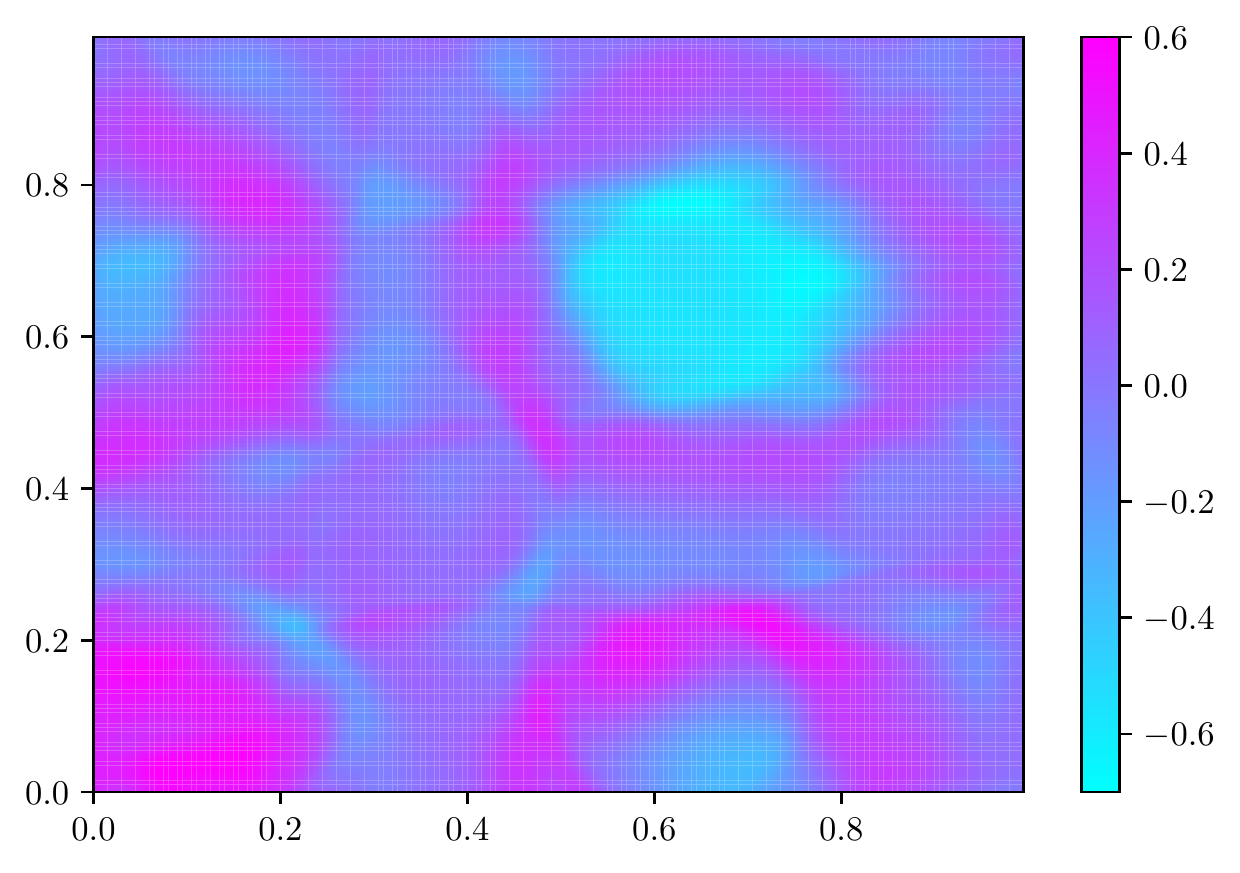}
  \caption{\centering Mean estimate for the tcNN-based posterior.}
\end{subfigure}
\caption{The neural network estimates the true $u^*(x)$ which is noisily observed on every grid point $x$ of a $20\times20$ grid. In real applications, only few observations will be available, this example simply illustrates that many observations lead to close approximations for the trace-class neural network prior. Note that the functions displayed are shown on a fine grid, a coarse $20\times20$ sub-grid was used to generate the observations.}
\label{fig:approx_u}
\end{figure}
\section{Bayesian Learning of the Optimal Value Function\label{RL}}
The solution to a stochastic optimal control problem is known as the {\it optimal value function} which can be found through Dynamic Programming (DP) (discussed in Section \ref{sec:setup}.) Reinforcement Learning is a popular and practical algorithmic approach for solving stochastic optimal control problems \citep{sutton2018reinforcement}. It finds the best control, which is a mapping from states to actions, in an online manner by using noisy estimates of the mathematical expectations to be maximised in DP. Online here refers to the use of the current best learnt control to actuate the system to its next state which is also accompanied by a corresponding stochastic reward. This interaction with the system yields a stochastic process of actions, states and rewards with which DP's mathematical expectations are estimated.

Automating a task can be made easier through the use of expert demonstrations, an approach known as Inverse Reinforcement Learning; see e.g. \citep{ramachandran2007bayesian} for more nuanced details. Given the observed state, actions and rewards from an expert, we can exploit the mathematical formalism of Markov Decision Processes to relate this ``data'' to the optimal value function of the expert. In a Bayesian approach to this problem, one defines a prior on a function space that includes all admissible value functions. The data observed from the expert's behaviour can then be used through a suitably defined likelihood \citep{ramachandran2007bayesian} to infer the expert's value function: having the expert's value function at hand allows one to mimic their behaviour and hence defines an approach for automation. For discrete state spaces, \cite{singh2013bayesian} provide a method to quantify the uncertainty of the estimated value function. Here, we will generalise those ideas to continuous state spaces by using the priors introduced in the previous section.

\subsection{Setup\label{sec:setup}}
A \emph{Markov Decision Process} is defined by a controlled Markov chain  $\{X_t\}_{t\in\mathbb N}$ called the \emph{state process}, the \emph{control process} $\{A_t\}_{t\in\mathbb N}$, and an optimality criterion. The state process takes values in a bounded set $\mathcal X\subset\mathbb R^d$, for simplicity we will 
assume the $d$-dimensional hypercube $\mathcal X=[0,1]^d$. The control process is $\mathcal A$-valued, where $\mathcal A=\{1,\dots,M\}$ is a finite set. Given states $X_{1:t}=x_{1:t}$ and actions $A_{1:t}=a_{1:t}$ up to time $t$, the next state $X_{t+1}$ is
\begin{align}\label{eq:state_transition}
    X_{t+1}|(X_{1:t}=x_{1:t},A_{1:t}=a_{1:t})\sim p(\cdot|x_t,a_t),
\end{align}
where for any state-action pair $(x_t,a_t)$, $p(\cdot|x_t,a_t)$ is a probability density. In some applications, the state dynamics are deterministic, and thus there exists a map $\mathcal T$ such that $X_{t+1}=\mathcal T(x_t,a_t)$. The action process depends on a \emph{policy} $\mu:\mathcal X\rightarrow\mathcal A$ which is a deterministic mapping from the state space into the action space: $A_t|(X_{1:t}=x_{1:t},A_{1:t-1}=a_{1:t-1})\sim \delta_{\mu(x_t)}(\cdot)$. As there are many possible mappings $\mu:\mathcal X\mapsto \mathcal A$, we assume the agent executes a policy that is in some way optimal. To be more precise, let $r:\mathcal X\rightarrow \mathbb R$ be the \emph{reward function}, then the \emph{accumulated reward} given a policy $\mu$ and an initial state $X_1=x_1$ is
\begin{align*}
    C_\mu(x_1)=\mathbb E_\mu\left[\sum_{t=1}^\infty\beta^tr(X_t)|X_1=x_1\right],
\end{align*}
where $\beta\in(0,1)$ is a discount factor. The discount factor serves two purposes: it ensures that the expectation is well defined, and also that early actions are more important (in terms of the reward it adds to the total) than later ones, see \cite{kaelbling1996reinforcement} for a more detailed discussion. A policy $\mu^*$ is \emph{optimal} if $C_{\mu^*}(x_1)\geq C_\mu(x_1)$ for all $(\mu,x_1)$ and the optimal policy can be found through the solution of Bellman's fixed-point equation \citep{bellman1952theory}. The function $v:\mathcal X\rightarrow\mathbb R$, which is the fixed-point solution to
\begin{align*}
    v(x)=\max_{a\in\mathcal A}\left[r(x)+\beta\int_{\mathcal X}p(x'|x,a)v(x')dx'\right],
\end{align*}
is called the \emph{optimal value function} \citep{bertsekas1995dynamic} and the corresponding optimal policy is
\begin{align}
    \mu^*(x)=\argmax_{a\in\mathcal A}\left[\int_{\mathcal X}p(x'|x,a)v(x')dx'\right],
    \label{eq:opt_pol}
\end{align}
that is, the optimal action at any state is the one that maximises the expected value function at the next state.

\subsection{Likelihood definition}
The above decision making process gives optimal actions, but a human expert may occasionally pick non-optimal ones. To model imperfect action selections, noise is added to \eqref{eq:opt_pol}. At each time step the chosen action is a random variable given by
\begin{align}\label{random_action}
    A_t(x_t)=\argmax_{a\in\mathcal A}\left[\int_{\mathcal X}p(x'|x_t,a)v(x')dx'+\epsilon_t(a)\right],
\end{align}
where we assume $\epsilon_t\sim\mathcal N(0,\sigma^2 I_{M\times M})$ for some $\sigma>0$. The Gaussian choice simplifies numerical calculations, and it is reasonable to assume that the variances for different actions are independent and identically distributed, but this assumption can be relaxed. From now on, we will assume that the state dynamics are deterministic, in which case the action selections occur according to 
\begin{align}\label{random_action_2}
    A_t(x_t)=\argmax_{a\in\mathcal A}\left[v(\mathcal T(x_t,a))+\epsilon_t(a)\right].
\end{align}
Our goal from now on will be to recover the optimal value function, and quantify the uncertainty thereof, by using the Hilbert space MCMC methods and the priors discussed in Sections \ref{function_space_MCMC} and \ref{sec:NN_priors}.

The data consists of a collection of state-action pairs $y=\{y_t\}_{t=1}^T=\{(x_t,a_t)\}_{t=1}^T$ and the aim is to infer the value function (and thus the policy through \eqref{eq:opt_pol}) that leads to the actions $a_t$ for the current state $x_t$. Using the noisy action selection procedure \eqref{random_action}, the likelihood is
\begin{align}\label{likelihood}
    \mathcal L(y|v,\sigma)=\prod_{t=1}^Tp(a_t|x_t,v,\sigma)=\prod_{t=1}^Tp(a_t|v_t,\sigma),
\end{align}
where the second equality follows by defining the vector $v_t$ to contain the relevant evaluations of the value function to calculate the likelihood at $y_t$, i.e. using equation \eqref{random_action_2}, the $k$-th entry of $v_t$ is the evaluation of the value function $v(\cdot)$ at the location $\mathcal T(x_t,k)$, corresponding to starting at $x_t$ and taking action $a=k \in \mathcal{A}$.

For a single observation $y_t=(x_t,a_t)$, we now drop the subscript $t$ to simplify notation, and assume wlog that the optimal action is action $a=1$, permuting the labels if necessary. The probability $p(a=1|v,\sigma)$ (where $v$ is a vector and $p(a|v,\sigma)$ is a probability mass function) can be computed using \eqref{random_action} by 
\begin{align}\label{likelihood_uncomputable}
    p(a=1|v,\sigma)&=\int\mathbbm 1_{\{u\in\mathbb R^d: u_1\geq u_j\forall j\neq 1\}}\mathcal N(u;v,\sigma^2 I_{M\times M})du.
\end{align}
To compute this probability, we make use of the fact that the value of the integral is the same as the probability $\mathbb P({U_1>U_j,~\forall j\neq 1})$, where $U_k\sim\mathcal N(v(\mathcal{T}(x,k)),\sigma^2))$. This can be computed numerically using the pdf $\phi_1(\cdot)$ of $U_1$ and cdfs $\Phi_j(\cdot)$ of the remaining random variables $U_j$:
\begin{align}
    p(a=1|v,\sigma) &= \int_{-\infty}^\infty \phi_1(t)\Phi_2(t)\dots\Phi_M(t)dt= \frac{1}{\sigma}\int_{-\infty}^\infty \phi\left(\frac{t-v_1}{\sigma}\right)\prod_{j=2}^M \Phi\left(\frac{t-v_j}{\sigma}\right)dt.\label{likelihood_computable}
\end{align}
where $v_j$ is $v(\mathcal{T}(x,j))$. If the noise in \eqref{random_action} is not diagonal, this simplification cannot be made, and the integral \eqref{likelihood_uncomputable} is harder to compute. More advanced numerical methods exist to efficiently calculate such integrals using Monte-Carlo simulations \citep{genz1992numerical}.

\subsection{Likelihood gradient}

Following from \eqref{likelihood_computable} we can compute the gradient of the likelihood in a data point $(x_t,a_t)$  with respect to $v_t$. We again assume wlog that $a_t=1\in \mathcal{A}$ (by swapping the label of the first and the best action if necessary), and drop the subscript $t$, emphasising that $v_k$ is the $k$-th entry of the vector $v=(v(\mathcal{T}(x,1)),\ldots,v(\mathcal{T}(x,M)))$. The partial derivatives with respect to the $v_k$ are given by
\begin{align}
    \frac{\partial}{\partial v_1}p(a=1|v,\sigma) &=\frac{1}{\sigma} \int_{-\infty}^\infty \frac{t-v_1}{\sigma^2}\phi\left(\frac{t-v_1}{\sigma}\right)\prod_{j=2}^M \Phi\left(\frac{t-v_j}{\sigma}\right)dt\label{partial_deriv1} \\
    \frac{\partial}{\partial v_k}p(a=1|v,\sigma) &=-\frac{1}{\sigma^2} \int_{-\infty}^\infty \phi\left(\frac{t-v_1}{\sigma}\right)\phi\left(\frac{t-v_k}{\sigma}\right)\prod_{j=2, j\neq k}^M \Phi\left(\frac{t-v_j}{\sigma}\right)dt\qquad k=2\dots M\\
    &=-\frac{1}{\sigma^2}\phi\left(\frac{v_1-v_k}{\sqrt{2}\sigma}\right)\int_{-\infty}^\infty
    \phi\left(\frac{t-\frac{v_k+v_1}{2}}{\frac{\sigma}{\sqrt{2}}}\right)\prod_{j=2, j\neq k}^M \Phi\left(\frac{t-v_j}{\sigma}\right)dt,\label{partial_deriv2}
\end{align}
where the last identity follows from the product of two Gaussian pdfs. This allows us, when using the neural network prior, to compute the gradient of the log-likelihood with respect to the parameters of the neural network, $\theta$, using backpropagation. We emphasise that the vector $v=v(\theta)$ depends on these parameters, justifying the calculation of the Jacobian $\mathcal D_\theta v$. Using the chain rule, we get
\begin{align}
    &\nabla_\theta\log p(a=1|v,\sigma)=\frac{\nabla_\theta p(a=1|v,\sigma)}{p(a=1|v,\sigma)}=\frac{(\mathcal D_\theta v)^T\nabla_v p(a=1|v,\sigma)}{p(a=1|v,\sigma)}\label{single_gradient}.
\end{align}
To get the entire gradient of the log-likelihood, we simply need to sum over all data points:
\begin{align}
    \nabla_\theta \ell(y|v,\sigma)=\nabla_\theta \log\left(\prod_{t=1}^Tp(a_t|v_t,\sigma)\right)
    =\nabla_\theta\sum_{t=1}^T\log p(a_t|v_t,\sigma)
    =\sum_{t=1}^T\nabla_\theta\log p(a_t|v_t,\sigma),\label{full_likelihood_gradient}
\end{align}
where we only need to keep in mind the permutation in the actions when using \eqref{single_gradient}.

When calculating \eqref{single_gradient}, we note that $1\cdot\nabla_v p(a|v,\sigma)=0$ by translation invariance of $v$: $\mathcal L(y|v,\sigma)=\mathcal L(y|v+c,\sigma)$ for any constant function $c$, i.e. $c(x)=c(x')$ for all $x,x'\in\mathcal X$. The integrals involved in the gradient are in practice calculated numerically, and the arising errors may accumulate and cause numerical instabilities. To avoid these, one can ensure that the mean of these gradients is $0$ by using the following modification, which we observed to enhance the performance in practice:
\begin{align}
    \eqref{single_gradient}=\frac{\sum_{k=1}^M((\mathcal D_\theta v)^T)_k(\frac{\partial}{\partial v_k} p(a=1|v,\sigma)-\sum_{k=1}^M\frac{\partial}{\partial v_k} p(a=1|v,\sigma))}{p(a=1|v,\sigma)}.\label{stable_gradient}
\end{align}

The following theorem justifies the use of this likelihood in the function space MCMC setting, see \cite[Chapter 12]{wainwright2019high} for a definition of reproducing kernel Hilbert spaces (RKHS):
\begin{thm}\label{thm:bound_likelihood_derivative}
The log-likelihood $\ell(y|v,\sigma)=\log \mathcal{L}(y|v,\sigma)$ defined in \eqref{likelihood} satisfies Assumptions \ref{assumption_likelihood1} and \ref{assumption_likelihood2}, if $v\in\mathcal H=L^2_{\mathcal K}$, where $L^2_{\mathcal K}$ is any RKHS defined on $L^2$.
\end{thm}
The proof can be found in Appendix \ref{proof:thm:bound_likelihood_derivative}. We also note that when using the trace-class neural network prior from Section \ref{sec:NN_priors}, the statements remain true if the likelihood is seen as a function of the parameters $\theta$ of the neural network: 
\begin{lem}\label{lemma:likelihood}
The log-likelihood $\ell(y|v_\theta,\sigma)$ defined in \eqref{likelihood} satisfies Assumptions \ref{assumption_likelihood1} and \ref{assumption_likelihood2}, where now inference is over the weights and biases, $\theta\in\mathcal H=\ell^2$.
\end{lem}
\begin{proof}
The proof can be found in Appendix \ref{proof:lemma:likelihood}.
\end{proof}

We now prove under which conditions on the likelihood one may use the preconditioned Crank-Nicolson Langevin algorithm when using the trace-class neural network prior, which in particular requires the gradient-informed proposals to be in the Cameron-Martin space of the prior. We will then remark on how it applies to the noisy action selection likelihood \eqref{random_action_2}. For Theorem \ref{thm:endomorphism_likelihood_derivative} assume the log-likelihood $\ell(y|v,\sigma)$ of the mapping $x\rightarrow v(x)\in\mathbb{R}$ is of the form 
\begin{align}\label{eq:loglikecontrol}
\ell(y|v,\sigma)=\sum_{t=1}^{T}\ell(a_{t},v(x_{t}^{1}),\ldots,v(x_{t}^{M}))
\end{align}
for some function $\ell:\mathcal{A}\times\mathbb{R}^{M}\rightarrow\mathbb{R}$, where a data point $y_{t}=(a_{t},x_{t}^{1},\ldots,x_{t}^{M})$ is comprised of $a_{t}\in\mathcal{A}$ and $M$ points in the domain of $v$, i.e. $x_{t}^{i}\in\mathcal{X}$. Note that such a likelihood clearly encompases \eqref{likelihood}. In the theorem below, we further assume uniformly bounded partial derivatives of the log-likelihood w.r.t. $v(x_t^i)$ for any $t$ and $i$. Even with this assumption, to verify the assertion of Theorem \ref{thm:endomorphism_likelihood_derivative}, we need to establish the behaviour of moments of $\partial v(x)/\partial W_{i,j}^{(l)}$ and $\partial v(x)/\partial B_{i}^{(i)}$ for all $x$, weights and biases; details can be found in its proof.

\begin{thm}\label{thm:endomorphism_likelihood_derivative} Consider the infinite-width neural network with the Gaussian prior given in \eqref{prior_def} with $\alpha>1$ for its weights and biases, abbreviated to $N(0,\mathcal{C})$, and the log-likelihood \ref{eq:loglikecontrol}. Assume $\sup_{v_{1:M}\in\mathbb{R}^{M}}\left|\partial\ell(a,v_{1},\ldots,v_{M})/\partial v_{i}\right|<\infty$ for all $i$ and $a\in\mathcal{A}$. Then Condition \eqref{assumption_pCNL} holds for the pCNL implementation \eqref{pCNL} for sampling from the posterior, that is $N(\mathcal{CD}\ell(u),\mathcal{C})\simeq N(0,\mathcal{C})$ for $u\sim N(0,\mathcal{C})$ almost surely.
\end{thm}

The proof can be found in Appendix \ref{proof:thm:endomorphism_likelihood_derivative}. The proposed stochastic control likelihood given in \eqref{random_action_2} does not satisfy the assumption of the theorem since the partial derivatives are unbounded. To circumvent this, we apply a saturation function $s$ to $v(x_t^i)$, and employ \eqref{random_action_2} with $s(v(\mathcal T(x_t,a)))$ instead of $v(\mathcal T(x_t,a))$. Lastly, we note that a similar result to Theorem \ref{thm:endomorphism_likelihood_derivative} can be shown for the Hilbert space $L^2$.

\section{Numerical Illustrations\label{numerics}}
This section aims to validate the theory, and highlight the applicability of the proposed priors and methodology. In particular, Section \ref{sec:dim_independence} confirms that, empirically, as the layer width for the trace-class neural network prior grows, the acceptance probability does not go to $0$, a property known as `stability under mesh-refinement' or `dimension-independence'. Section \ref{sec:prior_comparison} compares the proposed trace-class neural network (tcNN) prior to a standard BNN prior and a KL prior, it highlights that, unlike the KL prior, the tcNN is scalable to higher-dimensional domains; and Section \ref{sec:policy_learning} shows that the posteriors can learn and mimic policies, thus justifying the use of these priors in the reinforcement learning setup. The code is available at \url{https://github.com/TorbenSell/trace-class-neural-networks}.

Throughout we use the Fourier basis \eqref{eq:fourier_basis} as the series expansion of choice when using the KL based prior, as this proved to be a good choice for reinforcement learning problems \citep{konidaris2011value}. As a tuning parameter for the corresponding eigenvalues we set $\alpha=2$ in \eqref{eq:fourier_eigenvalues}, forcing the samples to be very smooth which we found to be a sensible choice in the discussed control problems. For the tcNN prior we used fully connected layers with $\tanh$ activation functions, and set $\sigma_{b^{(l)}}^2=\sigma_{w^{(l)}}^2=2$ and $\alpha=1.5$ in all the experiments, this again results in smooth sample functions. For the standard BNN we used the same architecture and set $\alpha=0$ to get a constant variance sequence, in Section \ref{sec:dim_independence} we set $\sigma_{b^{(l)}}^2=\sigma_{w^{(l)}}^2=10/(3N^{(l-1)})$ to highlight the dependence on the layer-width, in Sections \ref{sec:prior_comparison} and \ref{sec:policy_learning} we set $\sigma_{b^{(l)}}^2=\sigma_{w^{(l)}}^2=1/3$.

\subsection{Control Problems: Setup\label{sec:exp_setup}}
We set the scene by briefly describing the setup of the control problems which we use in the experiments, a detailed description can be found in the Supplementary Material.

The first example is the popular mountain car problem. A car is to drive up a mountain slope to reach a flag, but needs to gain momentum first by driving up the opposite mountain slope, thus initially driving away from the goal; see the left panel of Figure \ref{fig:setups} for an illustration. The state space is the two-dimensional domain $\mathcal X=[-1.2,0.6]\times[-0.07,0.07]$ describing the vehicle's position and velocity, and the action space contains three possible actions: $\mathcal A=\{-1,0,1\}$, representing exerting a constant force to the left, no force, and exerting the same constant force to the right, respectively. The likelihood \eqref{likelihood_computable} arises from $T=50$ observations of state-action pairs, the data generating process is described in the Supplementary Material. The noise level in the likelihood is set to $\sigma=0.1$.

The second example is the HalfCheetah example from the MuJoCo library \citep{todorov2012mujoco}, where an agent controls a two-dimensional cheetah with the aim to `run' as fast as possible. For this example, the state space $\mathcal X$ is $17$-dimensional and the action space contains $8$ possible actions. The likelihood \eqref{likelihood_computable} arises from $T=100$ observations, we again refer to the Supplementary Material for the data generating process, and set the noise level in the likelihood to $\sigma=0.1$. The right panel of Figure \ref{fig:setups} shows the HalfCheetah.

\begin{figure}[ht!]
\centering
\begin{subfigure}{.45\textwidth}
\centering
  \includegraphics[width=.8\linewidth]{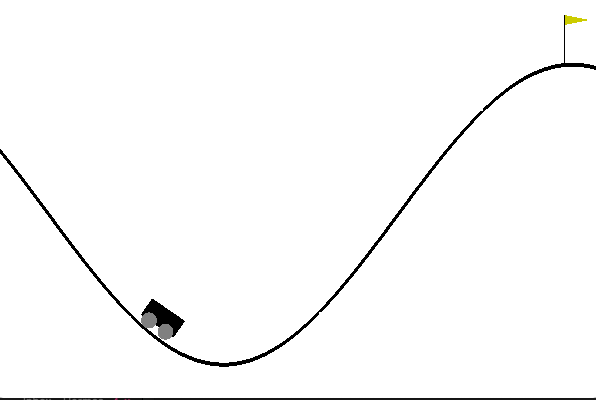}
\end{subfigure}
\quad
\begin{subfigure}{.45\textwidth}
\centering
  \includegraphics[width=.8\linewidth]{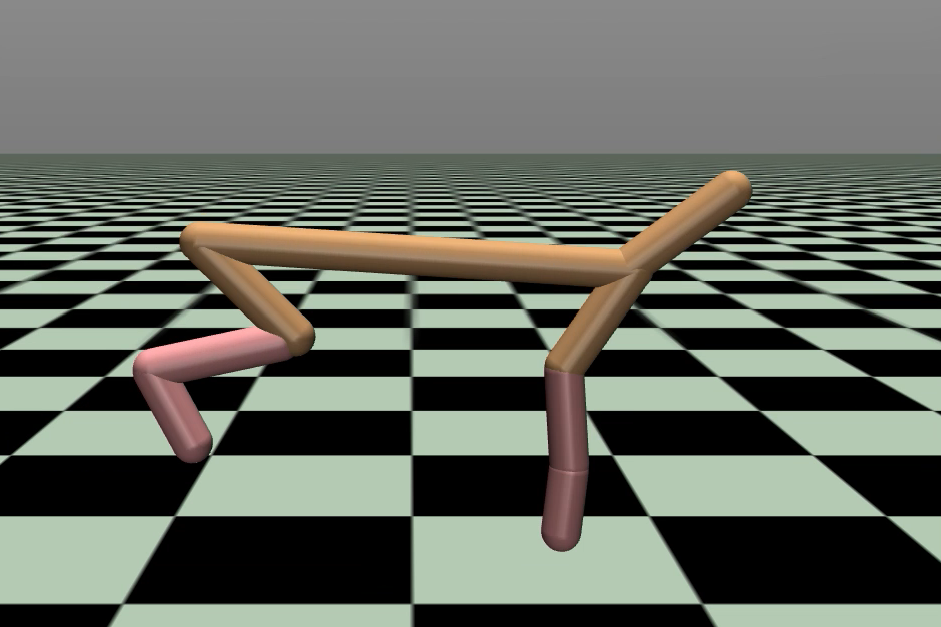}
\end{subfigure}
	\caption{\emph{Left}: The setup for the mountaincar example. The car's goal is to reach the flag in as few steps as possible. The slope on the right is too steep to simply drive up the mountain, the car therefore has to gain momentum by going up the hill on the left first. \emph{Right}: The HalfCheetah has states $x_t$ in $\mathbb{R}^{17}$. Its goal is to run to the right as quickly as possible, while not moving its body parts more than necessary.}
	\label{fig:setups}
\end{figure}

\subsection{Dimension independence of trace-class neural network prior under mesh-refinement\label{sec:dim_independence}}
We ran pCN for different network widths on the mountain car example. The network used has $l=3$ hidden layers. As stated before, the tuning parameters in the prior are set to $\sigma_{b^{(l)}}^2=\sigma_{w^{(l)}}^2=2$, and $\alpha=1.5$. Table \ref{table:NN_dim_independence} displays the acceptance probability of pCN for a fixed step size when targeting the posteriors arising from the mountain car likelihood with a trace-class neural network prior and also a standard Bayesian neural network prior. The latter is characterised by setting $\alpha=0$ in \eqref{prior_def}, resulting in a constant sequence of variances per layer. The other tuning parameters for the standard Bayesian neural network were set to $\sigma_{b^{(l)}}^2=\sigma_{w^{(l)}}^2=10/(3N^{(l)})$. The step sizes chosen were $\beta=1/10$ for the tcNN, $\beta=1/7$ for the standard BNN.

\begin{table}[ht!]
\begin{center}
    \begin{tabular}{| l | c | c | c | c | c | c | c | c | c | c |}
    \hline
    $N^{(l)}$, for all $l$ & 10 & 20 & 30 & 40 & 50 & 60 & 70 & 80 & 90 & 100 \\ \hline
    \hiderowcolors
    Acc. ratio (tcNN) & 22.8 & 24.0 & 23.5 & 22.1 & 22.2 & 23.1 & 23.9 & 23.4 & 23.0 & 23.9 \\ \hline
     Acc. ratio (BNN) & 21.2 & 15.0 & 10.9 & 8.52 & 6.81 & 5.47 & 4.25 & 3.91 & 2.97 & 2.23 \\ \hline
   Total \# of param. & $261$ & $921$ & $1981$ & $3441$ & $5301$ & $7561$ & $10221$ & $13281$ & $16741$ & $20601$ \\
    \hline
    \end{tabular}
\end{center}
\caption{Acceptance ratios in \% for both the trace-class neural network (tcNN) and standard Bayesian neural network (BNN) and total number of parameters (weights and biases) for different layer widths. $3$ fully connected layers were used, and pCN was run over $3$ hours for each choice of $N^{(l)}$. Notably the acceptance probability for the trace-class neural network proposed in this paper does not degenerate as more nodes are included per layer. Note that in the limit, only the tcNN is well-defined on the parameter space.}
\label{table:NN_dim_independence}
\end{table}

\subsection{Comparison of priors\label{sec:prior_comparison}}
To compare the trace-class neural network prior to the Karhunen-Lo\'eve prior, we used a large number of parameters for each, such that the error from truncating after finitely many nodes, or finitely many terms, is negligible. 
For both the mountaincar and the HalfCheetah example, we used the same trace-class neural network prior, with $3$ hidden layers, and $100$ nodes per layer, resulting in $20,601$ parameters to be estimated for the mountaincar example, and $22,101$ for the HalfCheetah example. For the Karhunen-Lo\'eve prior in the mountaincar example we set the truncation parameter to $k_{\mathrm{max}}=(70,70)$ for \eqref{eq:fourier_basis} with eigenvalues \eqref{eq:fourier_eigenvalues} (recall that here $\alpha=2$), resulting in a total of $19,880$ coefficients to be estimated. For the KL prior in the HalfCheetah example we used approximation \eqref{vf_approx}, and otherwise the same eigenfunctions and eigenvalues; due to the higher domain dimension $d=17$, one would have to estimate $2,667,980$ parameters. As this is too memory expensive for the computer used for the experiments, we used $k_{\mathrm{max}}=(10,10)$ in the HalfCheetah example, resulting in $54,740$ parameters to be estimated. Note that this increase in parameters to be estimated is despite the approximation \eqref{vf_approx} being used, and additionally truncating the expansions after fewer terms, highlighting the benefits of the dimension-robustness of the trace-class neural network prior.

To assess the quality of the priors, we ran pCN using $50$ (for the mountaincar) and $100$ (for the HalfCheetah) data points. For the mountain car example, we fixed five test points $z_j$, $j=1,\dots,5$ independent of the training data, and compared the posteriors by evaluating $v(z_j)$ at these new locations as estimated through MCMC runs. The top row in Figure \ref{fig:prior_compare} shows the resulting uncertainty estimates. As the value function is invariant under translations, we adjusted all samples such that they take the value $0$ at the state which the optimal action $a_{opt}$ takes one to:
\begin{align}
    v^{\text{centered}}_{1:M}=v_{1:M}-v_{a_{opt}}\cdot 1,\label{eq:normalise_v}
\end{align}
where $1$ denotes a vector of ones.
For the HalfCheetah example, we looked at one test point for illustration, see the bottom row in Figure \ref{fig:prior_compare}, and summarised the performance on another $100$ test points (independent of the training data) in the Table \ref{table:prior_compare}, where we compared how the respective samples from the posterior do, as well as how the mean of all samples from the posterior in Section \ref{sec:policy_learning} (with a smaller number of nodes for the tcNN prior, and fewer active terms in the KL prior\footnote{To calculate the mean function it is necessary to store the samples which (due to ther used computer's limited memory capacity) would not be feasible for the very wide layer prior, nor all the terms in the KL prior.}) does on predicting the correct action (last two columns). Not surprisingly, the mean function is better at picking the correct action. Details on the data generating mechanism can be found in the Supplementary Material.

{\footnotesize
\begin{table}
\begin{center}
    \begin{tabular}{| l || c | c | c || c | c | c |}
    \hline
    Decision by & KL samples & BNN samples & tcNN samples & KL mean & BNN mean & tcNN mean \\ \hline\hline
    \hiderowcolors
    Optimal & 20.1\% & 18.1\% & 32.1\% & 25\% & 20\% & 42\% \\ \hline
    Non-optimal & 79.9\% & 81.9\% & 67.9\% & 75\% & 80\% & 58\% \\ 
    \hline
    \end{tabular}
\caption{Actions picked using Equation \ref{random_action} with $v$ a posterior sample or the estimated posterior mean. The trace-class neural network prior outperforms the approximate KL and the standard BNN prior. The optimal action is computed using the same policy used to simulate data, see Section \ref{sec:exp_setup}, the test points chosen at random from a representative episode of a HalfCheetah run. A random prediction would result in a success rate of $12.5\%$.}
\label{table:prior_compare}
\end{center}
\end{table}
}

\begin{figure}[ht!]
\centering
\begin{subfigure}{.3\textwidth}
  \centering
  \includegraphics[width=\linewidth]{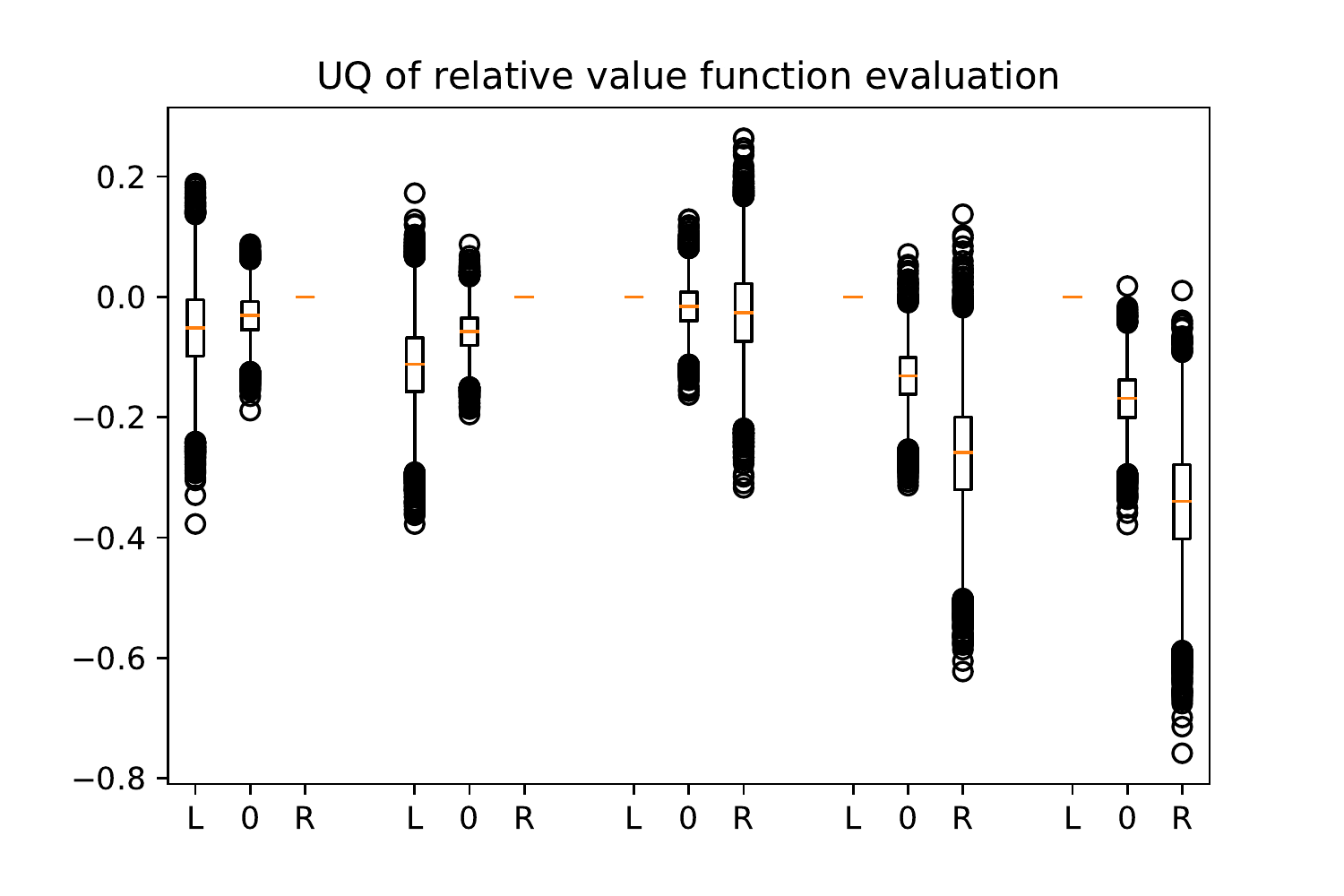}
  \caption{Mountaincar: KL based posterior}
\end{subfigure}
\begin{subfigure}{.3\textwidth}
  \centering
  \includegraphics[width=\linewidth]{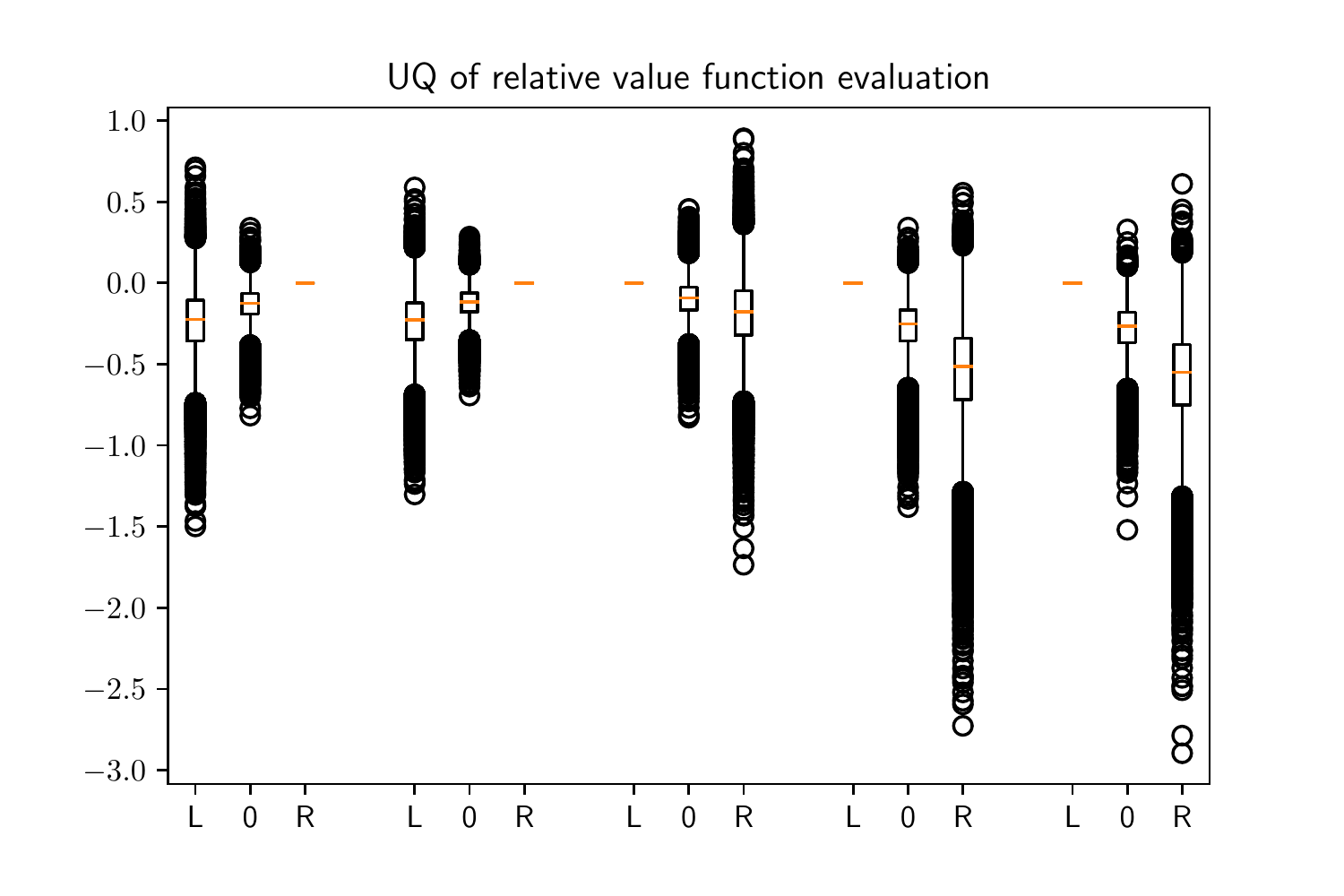}
  \caption{Mountaincar: BNN based posterior}
\end{subfigure}
\begin{subfigure}{.3\textwidth}
  \centering
  \includegraphics[width=\linewidth]{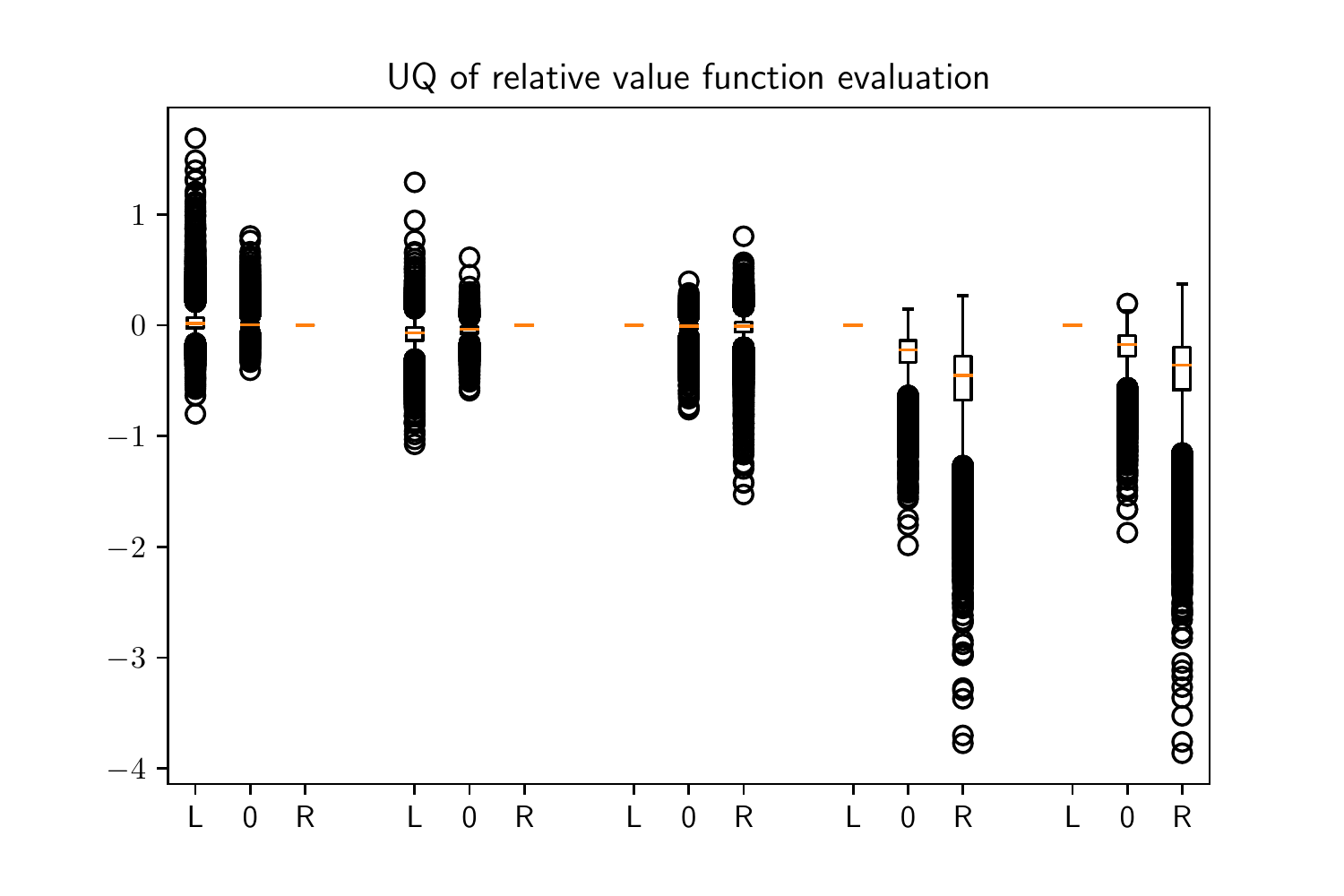}
  \caption{Mountaincar: tcNN based posterior}
\end{subfigure}
\vskip\baselineskip
\begin{subfigure}{.3\textwidth}
  \centering
  \includegraphics[width=\linewidth]{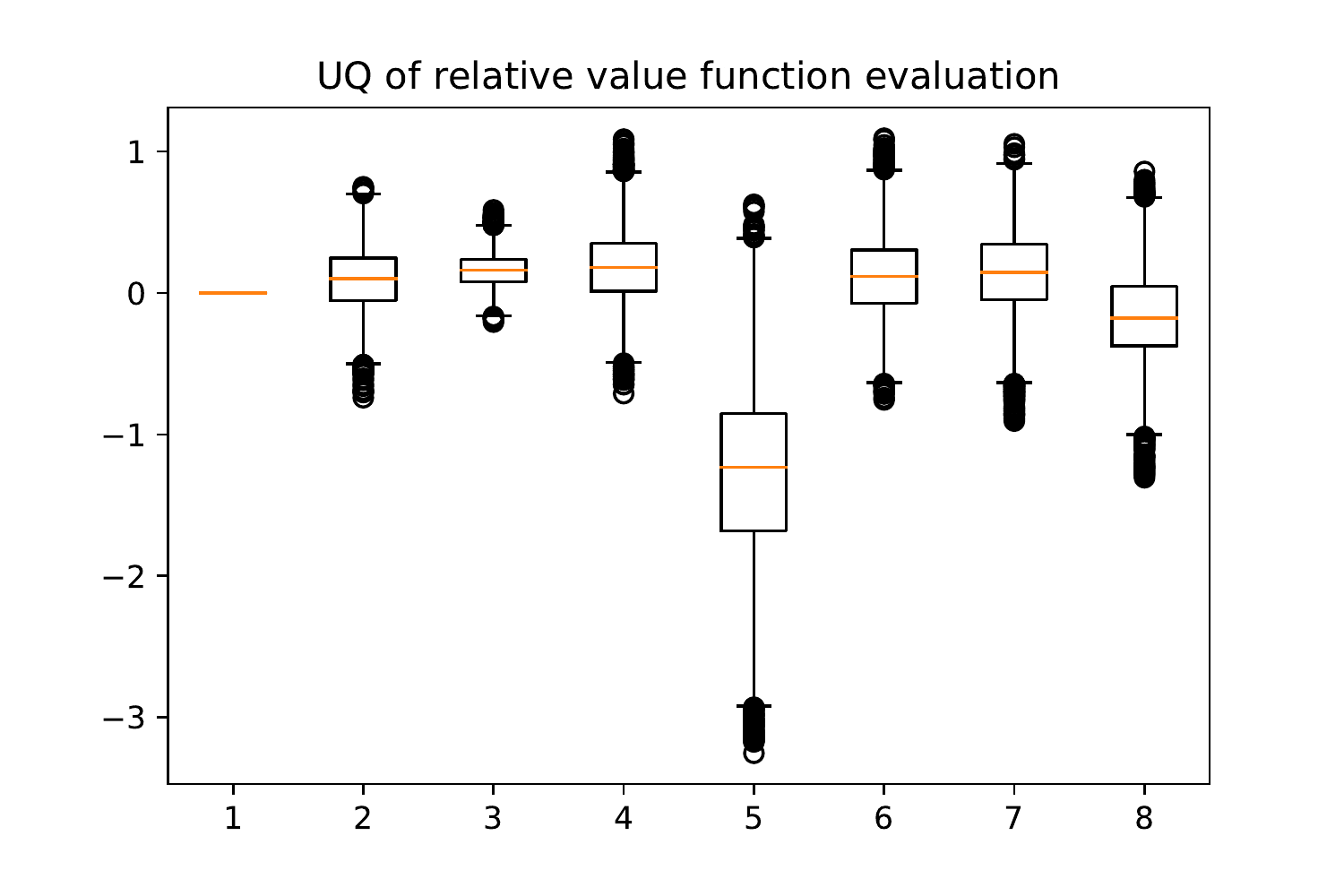}
  \caption{HalfCheetah: KL based posterior}
\end{subfigure}
\begin{subfigure}{.3\textwidth}
  \centering
  \includegraphics[width=\linewidth]{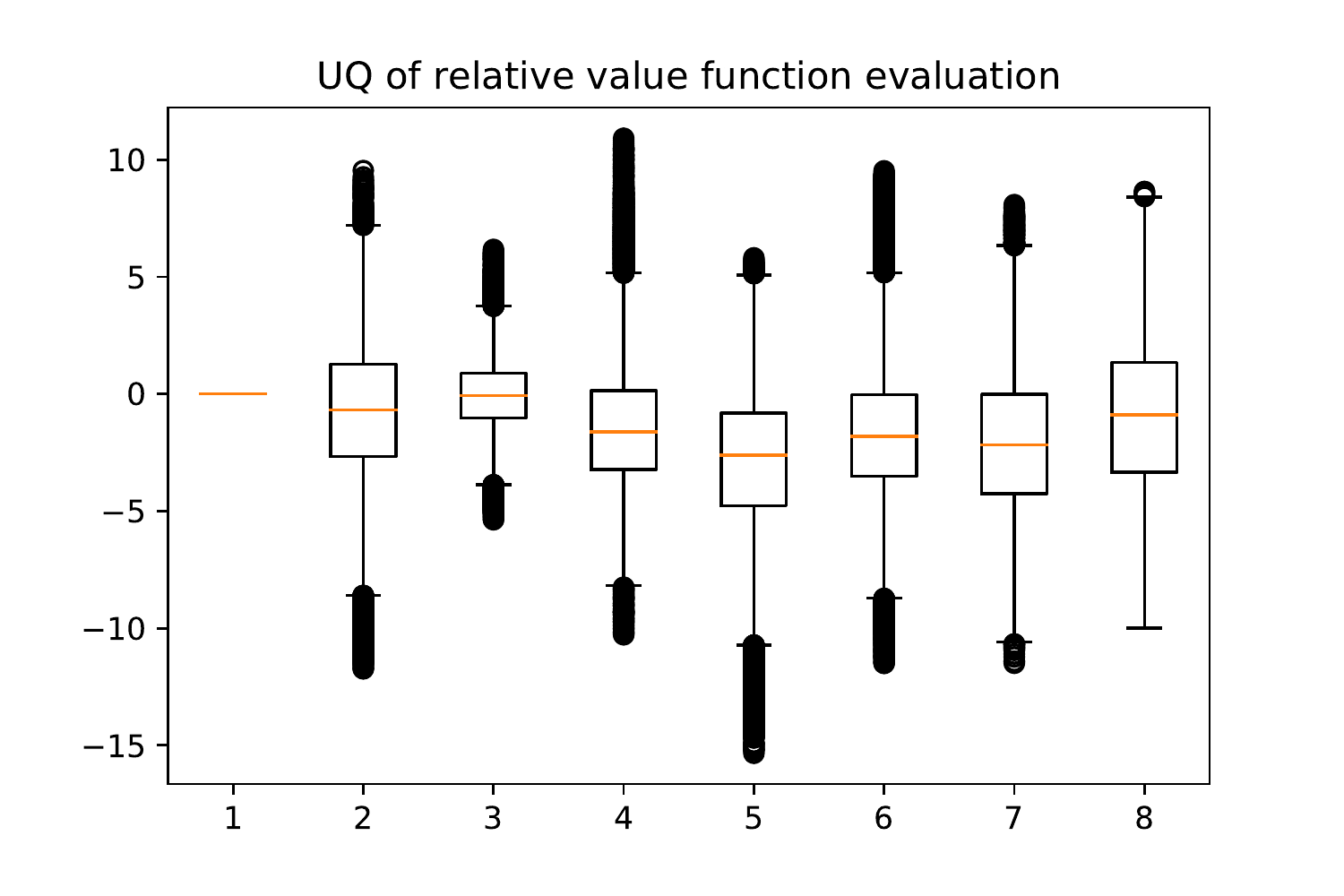}
  \caption{HalfCheetah: BNN based posterior}
\end{subfigure}
\begin{subfigure}{.3\textwidth}
  \centering
  \includegraphics[width=\linewidth]{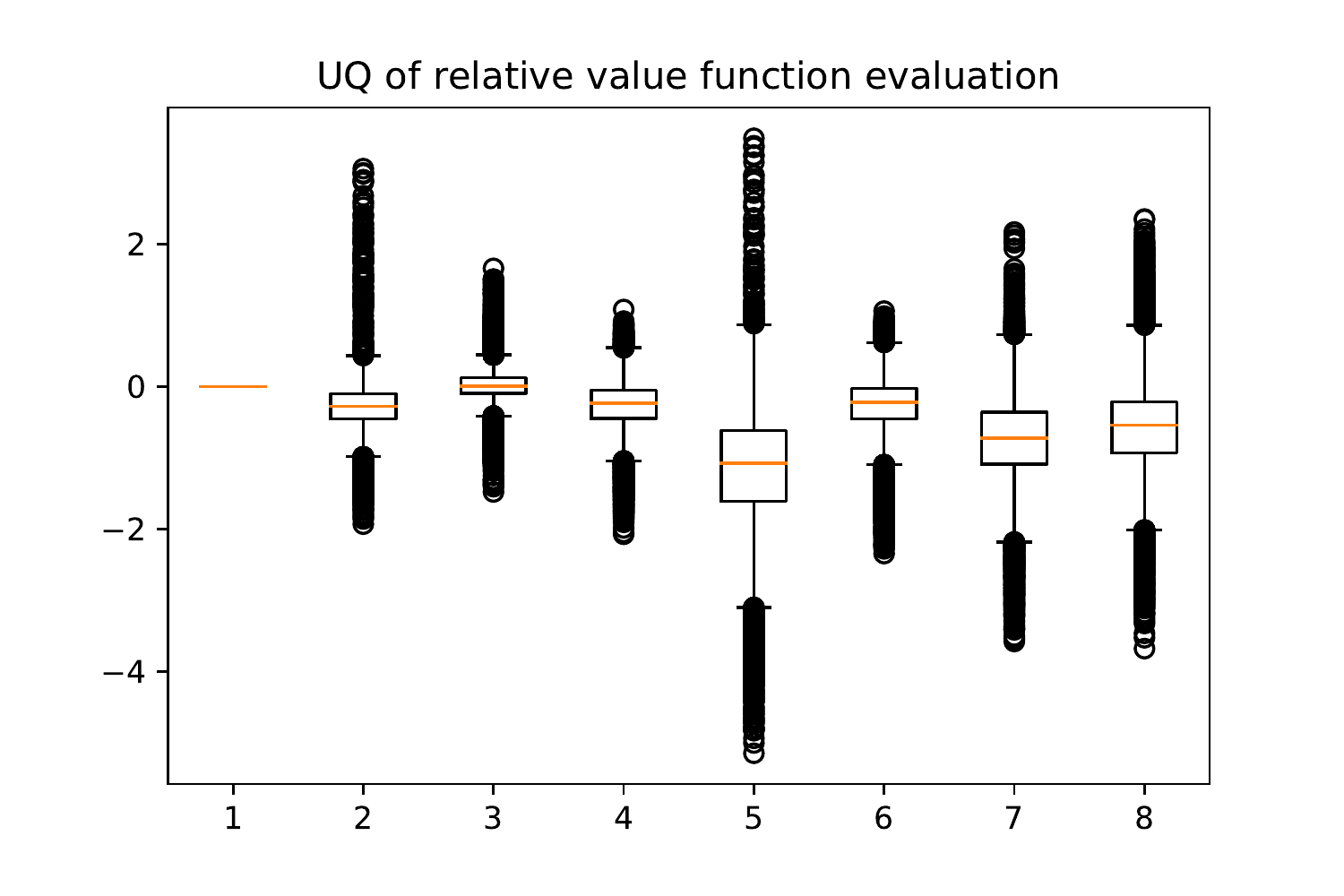}
  \caption{HalfCheetah: tcNN based posterior}
\end{subfigure}
\caption{Uncertainty quantified using estimates arising from the three different priors for the mountaincar and the HalfCheetah examples. \\
\emph{Top row}: Mountaincar example. In each plot, five different states are looked at, the estimates of the value functions are shown, standardised such that the optimal action has value $0$ always using \eqref{eq:normalise_v}. None of the posteriors can make a clear judgement as to what the optimal actions for the first three shown states are, as the boxplots illustrate the uncertainty when predicting the best action. For the fourth and fifth states, all posteriors suggest a clear decision for action `Left' as $v(\mathcal T(x,$`Left'$))>>v(\mathcal T(x,$`0'$))\vee v(\mathcal T(x,$`Right'$))$ cf. \eqref{random_action_2}. The reader should note that the KL, the BNN, and the tcNN posteriors behave similarly in that they are uncertain in the first three states, and very decisive in the last two states.\\
\emph{Bottom row}: HalfCheetah example. The optimal action is the first one in all three plots, and samples are again normalised using \eqref{eq:normalise_v} such that they take the value $0$ at the state the optimal value takes one to. The BNN and tcNN posteriors correctly estimate the optimal action, the KL posterior doesn't. 
}
\label{fig:prior_compare}
\end{figure}

\subsection{Ability to Learn Policy\label{sec:policy_learning}}
To asses if the posteriors can truly learn an agent's behaviour, we used the priors with a smaller number of parameters, and stored $1000$ samples for each posterior. We then used these samples to obtain a mean value function which was used for decision making. For the trace-class neural network prior we used $3$ layers with $10$ nodes per layer for both examples (resulting in $261$ parameters for the mountaincar example and $411$ for the HalfCheetah); for the KL prior we used $k_{\mathrm{max}}=(5,5)$ for the mountaincar example (giving a total of $224$ parameters), and $k_{\mathrm{max}}=(5,5)$ in the HalfCheetah example (a total of $8,730$ parameters). While the number of parameters can theoretically be chosen infinitely large, we truncated the layers and expansions earlier as we only had a very limited computational budget available. In general, where to truncate is an interesting model choice problem, and we found that for our problems the parameters described above yield very good approximations to a model with many more parameters. We thus chose to run the simplified model rather than a model with many more parameters, allowing many more stored MCMC posterior samples ($1000$ in this case) in the same wall-clock time. The results are summarised in Figure \ref{fig:results}.

\begin{figure}[ht!]
\centering
\begin{subfigure}{.45\textwidth}
  \centering
  \includegraphics[width=\linewidth]{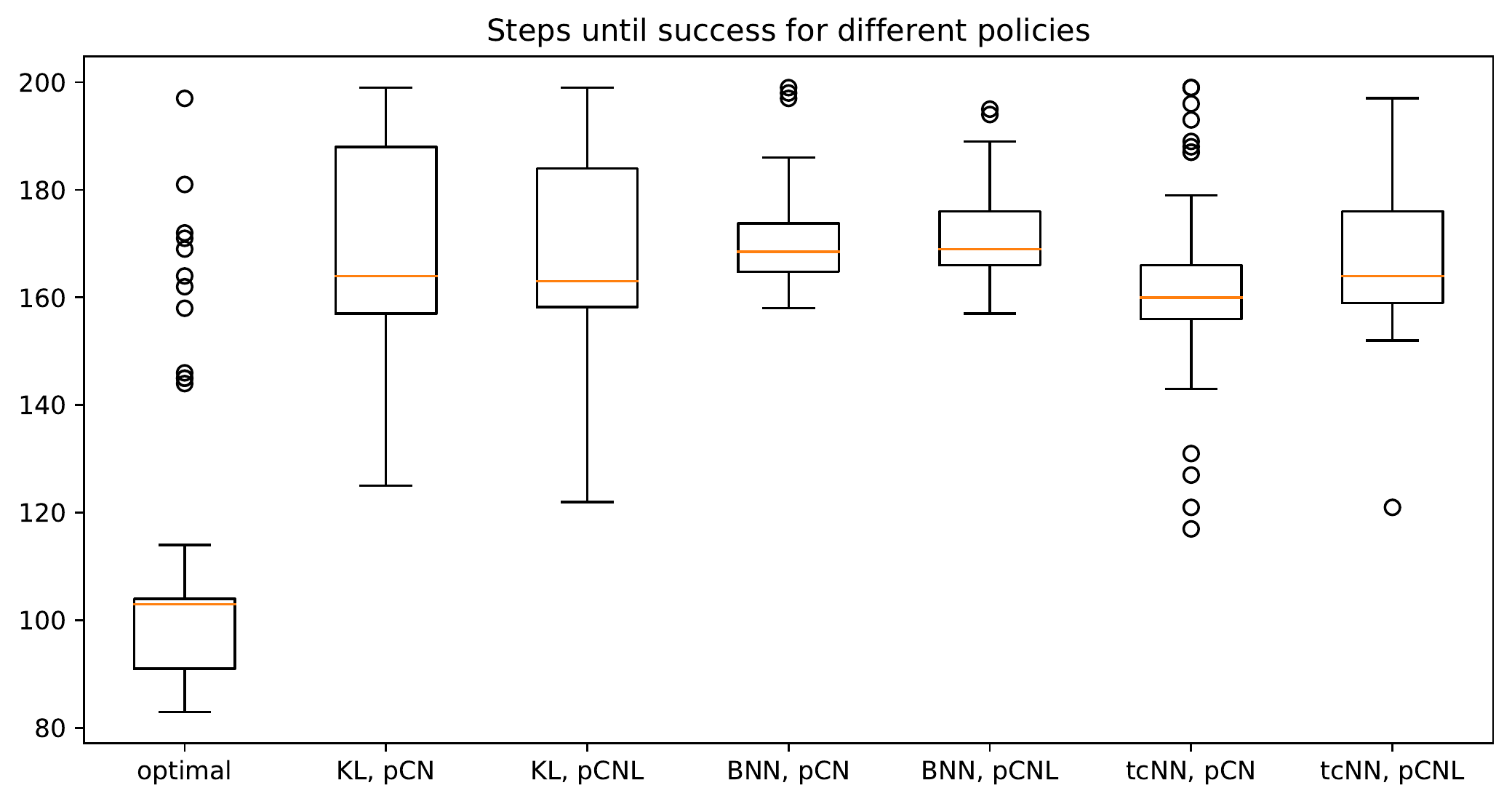}
  \caption{Mountaincar: Results}
\end{subfigure}
\quad
\begin{subfigure}{.45\textwidth}
  \centering
  \includegraphics[width=\linewidth]{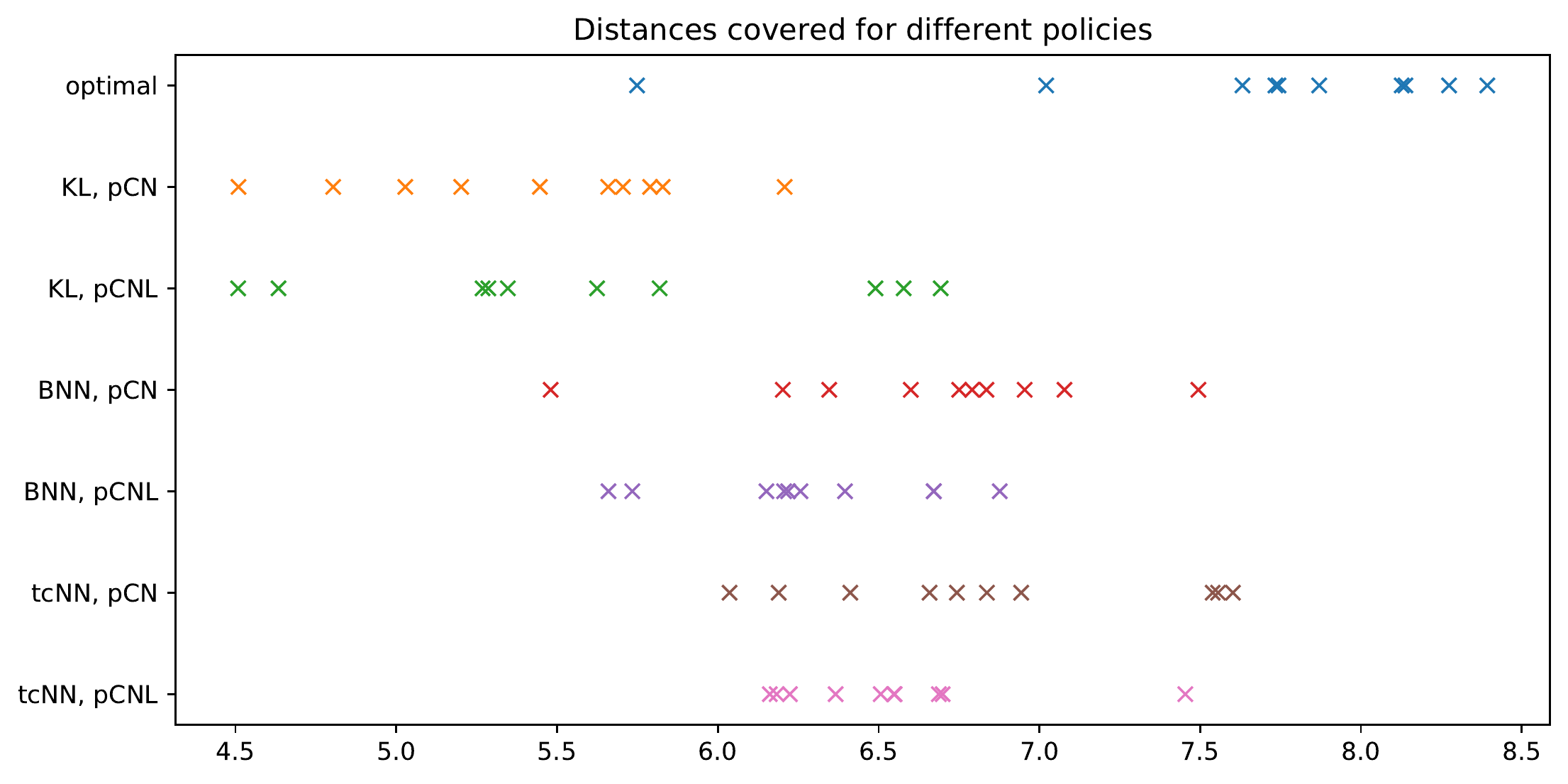}
  \caption{HalfCheetah: Results}
\end{subfigure}
    \caption{Results for the policy learning experiment.\\
\emph{Left}: Mountaincar example. The number of steps until success is shown for different posteriors. If the goal was not reached after $200$ steps, the run was counted as failure. Out of $100$ runs, the policy following the KL posterior when using pCN gave $34$ failures ($24$ when using pCNL), the standard BNN posterior gave $81$ (pCN) and $80$ (pCNL) failures, and the tcNN posteriors gave $23$ (for the posterior estimates obtained using pCN) and $25$ (pCNL). \\
\emph{Right}: HalfCheetah example. The different policies arising from the KL posterior (obtained once using pCN, once using pCNL), a standard BNN posterior, and the tcNN posterior were controlling the agent over $10$ runs with $100$ time steps. The distances covered per run are shown in the plot.}
    \label{fig:results}
\end{figure} 

\section{Conclusion and Outlook}
This paper addresses the problem of effective Bayesian inference for unknown functions with higher dimensional domains.
Unlike priors which require an orthogonal basis for the function space and scale exponentially in the domain dimension, our proposed trace-class neural network prior easily scales to higher-dimensional domains as the dependence on the domain dimension is linear. When using the pCN sampling method, this prior also satisfies the desired property of being stable under mesh-refinement, in the sense that the acceptance probability of pCN does not degenerate to $0$ when using more parameters for the neural network. Various questions remain unanswered though, and interesting directions for future work open up. For example, what are suitable generalisations of the proposed prior, e.g. heavy-tailed or hierarchical ones? What are the optimal settings for the tuning parameters $\sigma^2_{w^{(l)}}$, $\sigma^2_{b^{(l)}}$ and $\alpha$? 
Can one obtain contraction rates to ensure the concentration of the posterior samples around the true functions? A first idea here is to exploit the various generalisations of the universal approximation theorem \citep{scarselli1998universal}, and combine them with the proof methodology used in this paper.

We further introduced a likelihood suitable for Bayesian reinforcement learning where the underlying Markov decision process has a continuous state-space, and thus the unknown value function to be estimated has domain $\mathbb{R}^d$ as opposed to a discrete set. An interesting research direction is to generalise this to continuous action spaces as well.
Finally, we underscored the theory with numerical illustrations, illustrating the applicability of the prior for various control problems. It would also be interesting to evaluate the tcNN prior in other applied settings beyond control.

\section*{Acknowledgements}
Part of this research was carried out when TS received financial support from the Cantab Capital Institute for the Mathematics of Information, he is currently supported by the EPSRC New Investigator award EP/V002694/1.

\bibliographystyle{plainnat}
\bibliography{references.bib}

\begin{thebibliography}{55}
\providecommand{\natexlab}[1]{#1}
\providecommand{\url}[1]{\texttt{#1}}
\expandafter\ifx\csname urlstyle\endcsname\relax
  \providecommand{\doi}[1]{doi: #1}\else
  \providecommand{\doi}{doi: \begingroup \urlstyle{rm}\Url}\fi

\bibitem[Agapiou et~al.(2013)Agapiou, Larsson, and
  Stuart]{agapiou2013posterior}
Sergios Agapiou, Stig Larsson, and Andrew~M Stuart.
\newblock Posterior contraction rates for the bayesian approach to linear
  ill-posed inverse problems.
\newblock \emph{Stochastic Processes and Their Applications}, 123\penalty0
  (10):\penalty0 3828--3860, 2013.

\bibitem[Agapiou et~al.(2021)Agapiou, Dashti, and Helin]{agapiou2021rates}
Sergios Agapiou, Masoumeh Dashti, and Tapio Helin.
\newblock Rates of contraction of posterior distributions based on
  p-exponential priors.
\newblock \emph{Bernoulli}, 27\penalty0 (3):\penalty0 1616--1642, 2021.

\bibitem[Bellman(1952)]{bellman1952theory}
Richard Bellman.
\newblock On the theory of dynamic programming.
\newblock \emph{Proceedings of the National Academy of Sciences of the United
  States of America}, 38\penalty0 (8):\penalty0 716, 1952.

\bibitem[Bertsekas(1995)]{bertsekas1995dynamic}
Dimitri~P Bertsekas.
\newblock \emph{Dynamic Programming and Optimal Control}, volume~1.
\newblock Athena scientific Belmont, MA, 1995.

\bibitem[Beskos et~al.(2008)Beskos, Roberts, Stuart, and Voss]{beskos2008mcmc}
Alexandros Beskos, Gareth Roberts, Andrew Stuart, and Jochen Voss.
\newblock Mcmc methods for diffusion bridges.
\newblock \emph{Stochastics and Dynamics}, 8\penalty0 (03):\penalty0 319--350,
  2008.

\bibitem[Beskos et~al.(2017)Beskos, Girolami, Lan, Farrell, and
  Stuart]{beskos2017geometric}
Alexandros Beskos, Mark Girolami, Shiwei Lan, Patrick~E Farrell, and Andrew~M
  Stuart.
\newblock Geometric mcmc for infinite-dimensional inverse problems.
\newblock \emph{Journal of Computational Physics}, 335:\penalty0 327--351,
  2017.

\bibitem[Brockman et~al.(2016)Brockman, Cheung, Pettersson, Schneider,
  Schulman, Tang, and Zaremba]{1606.01540}
Greg Brockman, Vicki Cheung, Ludwig Pettersson, Jonas Schneider, John Schulman,
  Jie Tang, and Wojciech Zaremba.
\newblock Openai gym, 2016.

\bibitem[Brooks et~al.(2011)Brooks, Gelman, Jones, and
  Meng]{brooks2011handbook}
Steve Brooks, Andrew Gelman, Galin Jones, and Xiao-Li Meng.
\newblock \emph{Handbook of Markov Chain Monte Carlo}.
\newblock CRC press, 2011.

\bibitem[Cotter et~al.(2013)Cotter, Roberts, Stuart, and White]{cotter2013mcmc}
Simon~L Cotter, Gareth~O Roberts, Andrew~M Stuart, and David White.
\newblock Mcmc methods for functions: Modifying old algorithms to make them
  faster.
\newblock \emph{Statistical Science}, pages 424--446, 2013.

\bibitem[Cui et~al.(2016)Cui, Law, and Marzouk]{cui2016dimension}
Tiangang Cui, Kody~JH Law, and Youssef~M Marzouk.
\newblock Dimension-independent likelihood-informed mcmc.
\newblock \emph{Journal of Computational Physics}, 304:\penalty0 109--137,
  2016.

\bibitem[Da~Prato and Zabczyk(2014)]{da2014stochastic}
Giuseppe Da~Prato and Jerzy Zabczyk.
\newblock \emph{Stochastic Equations in Infinite Dimensions}.
\newblock Cambridge university press, 2014.

\bibitem[Damianou and Lawrence(2013)]{damianou2013deep}
Andreas Damianou and Neil Lawrence.
\newblock Deep gaussian processes.
\newblock In \emph{Artificial Intelligence and Statistics}, pages 207--215,
  2013.

\bibitem[Dashti et~al.(2011)Dashti, Harris, and Stuart]{dashti2011besov}
Masoumeh Dashti, Stephen Harris, and Andrew Stuart.
\newblock Besov priors for bayesian inverse problems.
\newblock \emph{arXiv preprint arXiv:1105.0889}, 2011.

\bibitem[Dashti et~al.(2013)Dashti, Law, Stuart, and Voss]{dashti2013map}
Masoumeh Dashti, Kody~JH Law, Andrew~M Stuart, and Jochen Voss.
\newblock Map estimators and their consistency in bayesian nonparametric
  inverse problems.
\newblock \emph{Inverse Problems}, 29\penalty0 (9):\penalty0 095017, 2013.

\bibitem[Dunlop et~al.(2018)Dunlop, Girolami, Stuart, and
  Teckentrup]{dunlop2018deep}
Matthew~M Dunlop, Mark~A Girolami, Andrew~M Stuart, and Aretha~L Teckentrup.
\newblock How deep are deep gaussian processes?
\newblock \emph{The Journal of Machine Learning Research}, 19\penalty0
  (1):\penalty0 2100--2145, 2018.

\bibitem[Eberle et~al.(2014)]{eberle2014error}
Andreas Eberle et~al.
\newblock Error bounds for metropolis--hastings algorithms applied to
  perturbations of gaussian measures in high dimensions.
\newblock \emph{The Annals of Applied Probability}, 24\penalty0 (1):\penalty0
  337--377, 2014.

\bibitem[Federer(1969)]{federer1969geometric}
Herbert Federer.
\newblock “geometric measure theory”, springer-verlag, berlin.
\newblock \emph{Heidelberg, New York}, 1969.

\bibitem[Gelman et~al.(2013)Gelman, Carlin, Stern, Dunson, Vehtari, and
  Rubin]{gelman2013bayesian}
Andrew Gelman, John~B Carlin, Hal~S Stern, David~B Dunson, Aki Vehtari, and
  Donald~B Rubin.
\newblock \emph{Bayesian Data Analysis}.
\newblock CRC press, 2013.

\bibitem[Genz(1992)]{genz1992numerical}
Alan Genz.
\newblock Numerical computation of multivariate normal probabilities.
\newblock \emph{Journal of computational and graphical statistics}, 1\penalty0
  (2):\penalty0 141--149, 1992.

\bibitem[Gin{\'e} and Nickl(2016)]{gine2016mathematical}
Evarist Gin{\'e} and Richard Nickl.
\newblock \emph{Mathematical Foundations of Infinite-dimensional Statistical
  Models}, volume~40.
\newblock Cambridge University Press, 2016.

\bibitem[Hairer et~al.(2014)Hairer, Stuart, Vollmer,
  et~al.]{hairer2014spectral}
Martin Hairer, Andrew~M Stuart, Sebastian~J Vollmer, et~al.
\newblock Spectral gaps for a metropolis--hastings algorithm in infinite
  dimensions.
\newblock \emph{The Annals of Applied Probability}, 24\penalty0 (6):\penalty0
  2455--2490, 2014.

\bibitem[Hastings(1970)]{hastings1970monte}
W~Keith Hastings.
\newblock Monte carlo sampling methods using markov chains and their
  applications.
\newblock 1970.

\bibitem[Hornik(1991)]{hornik1991approximation}
Kurt Hornik.
\newblock Approximation capabilities of multilayer feedforward networks.
\newblock \emph{Neural networks}, 4\penalty0 (2):\penalty0 251--257, 1991.

\bibitem[Hosseini(2017)]{hosseini2017well}
Bamdad Hosseini.
\newblock Well-posed bayesian inverse problems with infinitely divisible and
  heavy-tailed prior measures.
\newblock \emph{SIAM/ASA Journal on Uncertainty Quantification}, 5\penalty0
  (1):\penalty0 1024--1060, 2017.

\bibitem[Hosseini and Nigam(2017)]{hosseini2017well2}
Bamdad Hosseini and Nilima Nigam.
\newblock Well-posed bayesian inverse problems: Priors with exponential tails.
\newblock \emph{SIAM/ASA Journal on Uncertainty Quantification}, 5\penalty0
  (1):\penalty0 436--465, 2017.

\bibitem[Iglesias et~al.(2013)Iglesias, Law, and
  Stuart]{iglesias2013evaluation}
Marco~A Iglesias, Kody~JH Law, and Andrew~M Stuart.
\newblock Evaluation of gaussian approximations for data assimilation in
  reservoir models.
\newblock \emph{Computational Geosciences}, 17\penalty0 (5):\penalty0 851--885,
  2013.

\bibitem[Iserles and N{\o}rsett(2009)]{iserles2009high}
Arieh Iserles and Syvert~P N{\o}rsett.
\newblock From high oscillation to rapid approximation iii: Multivariate
  expansions.
\newblock \emph{IMA journal of numerical analysis}, 29\penalty0 (4):\penalty0
  882--916, 2009.

\bibitem[Kaelbling et~al.(1996)Kaelbling, Littman, and
  Moore]{kaelbling1996reinforcement}
Leslie~Pack Kaelbling, Michael~L Littman, and Andrew~W Moore.
\newblock Reinforcement learning: A survey.
\newblock \emph{Journal of artificial intelligence research}, 4:\penalty0
  237--285, 1996.

\bibitem[Knapik et~al.(2011)Knapik, Van Der~Vaart, van Zanten,
  et~al.]{knapik2011bayesian}
Bartek~T Knapik, Aad~W Van Der~Vaart, J~Harry van Zanten, et~al.
\newblock Bayesian inverse problems with gaussian priors.
\newblock \emph{The Annals of Statistics}, 39\penalty0 (5):\penalty0
  2626--2657, 2011.

\bibitem[Konidaris et~al.(2011)Konidaris, Osentoski, and
  Thomas]{konidaris2011value}
George Konidaris, Sarah Osentoski, and Philip Thomas.
\newblock Value function approximation in reinforcement learning using the
  fourier basis.
\newblock In \emph{Twenty-fifth AAAI conference on artificial intelligence},
  2011.

\bibitem[Leimkuhler et~al.(2019)Leimkuhler, Matthews, and
  Vlaar]{leimkuhler2019partitioned}
Benedict Leimkuhler, Charles Matthews, and Tiffany Vlaar.
\newblock Partitioned integrators for thermodynamic parameterization of neural
  networks.
\newblock \emph{arXiv preprint arXiv:1908.11843}, 2019.

\bibitem[Matthews et~al.(2018)Matthews, Rowland, Hron, Turner, and
  Ghahramani]{matthews2018gaussian}
Alexander G de~G Matthews, Mark Rowland, Jiri Hron, Richard~E Turner, and
  Zoubin Ghahramani.
\newblock Gaussian process behaviour in wide deep neural networks.
\newblock \emph{arXiv preprint arXiv:1804.11271}, 2018.

\bibitem[Minchew et~al.(2015)Minchew, Simons, Hensley, Bj{\"o}rnsson, and
  P{\'a}lsson]{minchew2015early}
Brent Minchew, Mark Simons, Scott Hensley, Helgi Bj{\"o}rnsson, and Finnur
  P{\'a}lsson.
\newblock Early melt season velocity fields of langj{\"o}kull and
  hofsj{\"o}kull, central iceland.
\newblock \emph{Journal of Glaciology}, 61\penalty0 (226):\penalty0 253--266,
  2015.

\bibitem[Neal(2012)]{neal2012bayesian}
Radford~M Neal.
\newblock \emph{Bayesian Learning for Neural Networks}, volume 118.
\newblock Springer Science \& Business Media, 2012.

\bibitem[Neal(1995)]{neal1995bayesian}
RM~Neal.
\newblock Bayesian learning for neural networks [phd thesis].
\newblock \emph{Toronto, Ontario, Canada: Department of Computer Science,
  University of Toronto}, 1995.

\bibitem[Neal(1998)]{neal1998regression}
RM~Neal.
\newblock Regression and classification using gaussian process priors.
\newblock \emph{Bayesian statistics}, 6:\penalty0 475, 1998.

\bibitem[Nickl and Giordano(2020)]{nickl2020consistency}
Richard Nickl and Matteo Giordano.
\newblock Consistency of bayesian inference with gaussian process priors in an
  elliptic inverse problem.
\newblock \emph{Inverse Problems}, 2020.

\bibitem[Quarteroni et~al.(2017)Quarteroni, Manzoni, and
  Vergara]{quarteroni2017cardiovascular}
ALFIO Quarteroni, Andrea Manzoni, and Christian Vergara.
\newblock The cardiovascular system: Mathematical modelling, numerical
  algorithms and clinical applications.
\newblock \emph{Acta Numerica}, 26:\penalty0 365--590, 2017.

\bibitem[Ramachandran and Amir(2007)]{ramachandran2007bayesian}
Deepak Ramachandran and Eyal Amir.
\newblock Bayesian inverse reinforcement learning.
\newblock In \emph{IJCAI}, volume~7, pages 2586--2591, 2007.

\bibitem[Roberts and Rosenthal(1998)]{roberts1998optimal}
Gareth~O Roberts and Jeffrey~S Rosenthal.
\newblock Optimal scaling of discrete approximations to langevin diffusions.
\newblock \emph{Journal of the Royal Statistical Society: Series B (Statistical
  Methodology)}, 60\penalty0 (1):\penalty0 255--268, 1998.

\bibitem[Roberts et~al.(1996)Roberts, Tweedie, et~al.]{roberts1996exponential}
Gareth~O Roberts, Richard~L Tweedie, et~al.
\newblock Exponential convergence of langevin distributions and their discrete
  approximations.
\newblock \emph{Bernoulli}, 2\penalty0 (4):\penalty0 341--363, 1996.

\bibitem[Roberts et~al.(2001)Roberts, Rosenthal, et~al.]{roberts2001optimal}
Gareth~O Roberts, Jeffrey~S Rosenthal, et~al.
\newblock Optimal scaling for various metropolis--hastings algorithms.
\newblock \emph{Statistical science}, 16\penalty0 (4):\penalty0 351--367, 2001.

\bibitem[Scarselli and Tsoi(1998)]{scarselli1998universal}
Franco Scarselli and Ah~Chung Tsoi.
\newblock Universal approximation using feedforward neural networks: A survey
  of some existing methods, and some new results.
\newblock \emph{Neural networks}, 11\penalty0 (1):\penalty0 15--37, 1998.

\bibitem[Singh et~al.(2013)Singh, Chopin, and Whiteley]{singh2013bayesian}
Sumeetpal~S Singh, Nicolas Chopin, and Nick Whiteley.
\newblock Bayesian learning of noisy markov decision processes.
\newblock \emph{ACM Transactions on Modeling and Computer Simulation (TOMACS)},
  23\penalty0 (1):\penalty0 4, 2013.

\bibitem[Sobol(1993)]{sobol1993sensitivity}
Ilya~M Sobol.
\newblock Sensitivity estimates for nonlinear mathematical models.
\newblock \emph{Mathematical modelling and computational experiments},
  1\penalty0 (4):\penalty0 407--414, 1993.

\bibitem[Stuart(2010)]{stuart2010inverse}
Andrew~M Stuart.
\newblock Inverse problems: A bayesian perspective.
\newblock \emph{Acta numerica}, 19:\penalty0 451--559, 2010.

\bibitem[Sutton and Barto(2018)]{sutton2018reinforcement}
Richard~S Sutton and Andrew~G Barto.
\newblock \emph{Reinforcement Learning: An Introduction}.
\newblock MIT press, 2018.

\bibitem[Tierney et~al.(1998)]{tierney1998note}
Luke Tierney et~al.
\newblock A note on metropolis--hastings kernels for general state spaces.
\newblock \emph{The Annals of Applied Probability}, 8\penalty0 (1):\penalty0
  1--9, 1998.

\bibitem[Todorov et~al.(2012)Todorov, Erez, and Tassa]{todorov2012mujoco}
Emanuel Todorov, Tom Erez, and Yuval Tassa.
\newblock Mujoco: A physics engine for model-based control.
\newblock In \emph{2012 IEEE/RSJ International Conference on Intelligent Robots
  and Systems}, pages 5026--5033. IEEE, 2012.

\bibitem[van~der Vaart et~al.(2008)van~der Vaart, van Zanten,
  et~al.]{van2008rates}
Aad~W van~der Vaart, J~Harry van Zanten, et~al.
\newblock Rates of contraction of posterior distributions based on gaussian
  process priors.
\newblock \emph{The Annals of Statistics}, 36\penalty0 (3):\penalty0
  1435--1463, 2008.

\bibitem[Wainwright(2019)]{wainwright2019high}
Martin~J Wainwright.
\newblock \emph{High-dimensional statistics: A non-asymptotic viewpoint},
  volume~48.
\newblock Cambridge University Press, 2019.

\bibitem[Welling and Teh(2011)]{welling2011bayesian}
Max Welling and Yee~W Teh.
\newblock Bayesian learning via stochastic gradient langevin dynamics.
\newblock In \emph{Proceedings of the 28th international conference on machine
  learning (ICML-11)}, pages 681--688, 2011.

\bibitem[Wenzel et~al.(2020)Wenzel, Roth, Veeling, {\'S}wi{\k{a}}tkowski, Tran,
  Mandt, Snoek, Salimans, Jenatton, and Nowozin]{wenzel2020good}
Florian Wenzel, Kevin Roth, Bastiaan~S Veeling, Jakub {\'S}wi{\k{a}}tkowski,
  Linh Tran, Stephan Mandt, Jasper Snoek, Tim Salimans, Rodolphe Jenatton, and
  Sebastian Nowozin.
\newblock How good is the bayes posterior in deep neural networks really?
\newblock \emph{arXiv preprint arXiv:2002.02405}, 2020.

\bibitem[Wojtaszczyk(1997)]{wojtaszczyk1997mathematical}
Przemyslaw Wojtaszczyk.
\newblock \emph{A Mathematical Introduction to Wavelets}, volume~37.
\newblock Cambridge University Press, 1997.

\bibitem[Xiao(2019)]{xiao2019}
Zhiqing Xiao.
\newblock \emph{Reinforcement Learning: Theory and {Python} Implementation}.
\newblock China Machine Press, 2019.

\end{thebibliography}

\newpage
\appendix
\section{NodeSwap Algorithm\label{app:alg}}
\begin{algorithm}[ht!]
\caption{}\label{alg:prior_swaps}
\begin{algorithmic}[1] \Procedure{NodeSwap}{$\theta$}\Comment{Input current iterate $\theta$}
\State $\theta'\gets\theta$ \Comment{Priming the return value}
\State $l\sim Unif(n)$\Comment{Sample random layer}
\State $i\sim Geom(\alpha^{-1})$\Comment{Sample random node}
\While{$i\geq N_{l}$}\Comment{Repeat process until we have a valid node index}
      \State $i\sim Geom(\alpha^{-1})$
   \EndWhile
   \State $\forall j:~w_{i+1,j}^{(l)~\prime}\gets w_{i,j}^{(l)}$
   \State $\forall j:~w_{i,j}^{(l)~\prime}\gets w_{i+1,j}^{(l)}$
   \State $\forall j:~w_{j,i+1}^{(l+1)~\prime}\gets w_{j,i}^{(l+1)}$
   \State $\forall j:~w_{j,i}^{(l+1)~\prime}\gets w_{j,i+1}^{(l+1)}$
   \State $b_{i+1}^{(l)~\prime}\gets b_{i}^{(l)}$
   \State $b_{i}^{(l)~\prime}\gets b_{i+1}^{(l)}$
   \State $u\sim Unif([0,1])$
   \State $a=\min(1,\mu_0(\theta')/\mu_0(\theta))$\Comment{Metropolis-Hastings acceptance probability cf. \eqref{eq:nodeswapinvariance}}
   \If{u<a}
   \State \Return $\theta'$ \Comment{Accept node swap}
   \Else{}
   \State \Return $\theta$ \Comment{Reject node swap}
   \EndIf
\EndProcedure
\end{algorithmic}
\end{algorithm}

\section{Proofs}
Before turning to the proofs of the lemmas and theorems from the main paper, consider the $n$-layer fully connected feed-forward neural network in \eqref{def:NNfunctions}. When the layers have infinite width, we delineate the domain of the sequences that define each layer separately. For layer $1<l<n+1$ let 
\begin{equation}
\label{eq:inf_NN_params}
 \mathcal{H}^{(l)}_w =
\left\{ 
(w_{i,j}^{(l)})_{i,j \in \mathbb{N}}: \sum_{i=1}^{\infty} \sum_{j=1}^{\infty} (w_{i,j}^{(l)})^2 
< \infty \right\},
\quad \mathcal{H}^{(l)}_b =
\left\{ 
(b_{i}^{(l)})_{i\in\mathbb N}
: \sum_{i=1}^{\infty} (b_i^{(l)})^2 
< \infty \right\}.
\end{equation}
(We omit the obvious modification for the sequence spaces for layer $1$ and $n+1$.) The entire network is then parameterised by 
\begin{equation}
\mathcal{H}=\mathcal{H}^{(1)}\times \cdots \times \mathcal{H}^{(n+1)}, \quad  \mathrm{where~} \mathcal{H}^{(l)}=\mathcal{H}^{(l)}_w \times \mathcal{H}^{(l)}_b.
\label{eq:inf_NN_params2}
\end{equation}
This domain is chosen because it has full measure under our Hilbert space Gaussian prior and also results in the infinite width functions in \eqref{eq:infNN_functions} being well defined almost surely. 

\subsection{Lemma 2}
\begin{lem}
\label{lem:equiv_to_l2}
Consider the $n$-layer fully connected feed-forward neural network in \eqref{def:NNfunctions}. When the layers have infinite width, their weights and biases 
can be equivalently parameterised by $\ell^2=\{(a_1,a_2,\ldots) \in \mathbb{R}^{\mathbb{N}}: \sum_{i}^{\infty} a_i^2 < \infty\}$. 
\end{lem}
\begin{proof}
$\mathcal{H}^{(l)}_w$ in \eqref{eq:inf_NN_params} is an instance of the Hilbert space $\ell^2$ since $\mathbb{N}\times \mathbb{N}$ is countable and {\it any} enumeration (e.g. the `diagonal' enumeration method) of $\mathcal{H}^{(l)}_w$ to map its elements to infinite sequences of the form $(a_1,a_2,\ldots)$  will be square summable. Similarly, $\mathcal{H}^{(l)}=\mathcal{H}^{(l)}_w \times \mathcal{H}^{(l)}_b$, the cartesian product of two $\ell^2$ spaces is again an instance of $\ell^2$ regardless of how the two sequences are merged into one. Finally, by the same arguments, 
$\mathcal{H}=\mathcal{H}^{(1)}\times \cdots \times \mathcal{H}^{(n+1)}$ is also 
an instance of $\ell^2$.
\end{proof}

\subsection{Proof of Theorem \ref{thm:prior}\label{proof:thm:prior}}
\begin{proof}[Proof of Theorem \ref{thm:prior}.]
We prove the claims in the theorem for the infinite width case and in doing so cover the finite width case; the finite-dimensional case follows by omitting the limit arguments.\\

Lemma \ref{lem:equiv_to_l2} shows that the weights and biases of the infinite width and finite depth neural network can be equivalently parameterised by $\ell^2$. As the biases and weights of each layer are independent zero mean Gaussian random variables, and the variances form a summable sequence when $\alpha>1$, the prior $\mu_0$ is a trace-class Gaussian prior on $\ell^2$ and thus Property \ref{assumption_Hilbert_space} is satisfied. 

To see Property \ref{assumption_finite_variance}, by looking at the first layer we can easily check that for fixed $x\in[0,1]^d$, $f_i^{(1)}(x)$ is a mixture of centered Gaussian distributions, and the claim follows by noting that $\mathbb EB_i^{(1)}=\mathbb EW_{i,j}^{(1)}=0$,
\begin{align}
    \mathbb E\left[(f_{i}^{(1)}(x))^2\right]&=\mathbb E\left[(B_i^{(1)})^2\right]+\sum_{j=1}^d\mathbb E\left[(W_{i,j}^{(1)})^2\right](x_j)^2\nonumber\\
    &\leq\frac{\sigma_{b_1}^2}{i^\alpha}+\frac{\sigma_{w_1}^2}{i^\alpha}d\\
    &=\frac{1}{i^\alpha}\left[\sigma_{b_1}^2+\sigma_{w_1}^2d\right].\label{sigma_1=1}
\end{align}

We use induction over $l$, and define the following random variables, for which we truncate the $i$-th function of layer $l$ after $k$ terms:
\begin{align*}
    f_{i,k}^{(l)}(x)=B_i^{(l)}+\sum_{j=1}^kW_{i,j}^{(l)}\zeta(F_j^{(l-1)}(x))\qquad l=2\dots n+1.
\end{align*}
(Note that, with slight abuse of notation, we write $F_j^{(1)}$ even for the functions on the first layer, which are defined by finitely many parameters.)

By Assumption \ref{assumption_bounded_activation},
\begin{align}
    \lvert\mathbb E\zeta(F_j^{(l-1)}(x))\rvert
    \leq \mathbb E[\lvert\zeta(F_j^{(l-1)}(x))\rvert]\leq \mathbb E[\lvert F_j^{(l-1)}(x)\rvert]<\infty\label{Ephi_finite},
\end{align}
where the last inequality holds as $F_j^{(l-1)}(x)$ is $L^2$ bounded by the induction hypothesis.

We now show that $f_{i,k}^{(l)}(x)\rightarrow F_i^{(l)}(x)$ almost surely, and in $L^2$, by applying the $L^2$ martingale convergence theorem. We thus need to show that $S_k(x):=f_{i,k}^{(l)}(x)$ is a $L^2$ bounded martingale, where we dropped the indices $i$ and $l$ for notational convenience. Indeed, with the natural filtration $(\mathcal F_k)_{k\in\mathbb N}$
\begin{align*}
    \mathbb E[S_{k+1}(x)|\mathcal F_k]=0,
\end{align*}
as $W_{i,j}^{(l)}$ and $\zeta(F_j^{(l-1)}(x))$ are independent, the expectation of the former is centered, and the latter is finite. Additionally, by exploiting the independence, Assumption \ref{assumption_bounded_activation} and \ref{Ephi_finite}, we get
\begin{align}
    \mathbb E\left[(S_k(x))^2\right]&=\mathbb E[B_i^{(l)}]^2+\sum_{j=1}^k\mathbb E\left[(W_{i,j}^{(l)})^2\right]\mathbb E\left[(\zeta(F_j^{(l-1)}(x)))^2\right]\label{S_k}\\
    &\leq\frac{\sigma_{b^{(l)}}^2}{i^\alpha}+\sigma_{w^{(l)}}^2\sum_{j=1}^k\frac{1}{(ij)^\alpha}\mathbb E\left[(F_j^{(l-1)}(x))^2\right]\leq\frac{\sigma_{b^{(l)}}^2}{i^\alpha}+\frac{\sigma_{w^{(l)}}^2\sigma_{l-1}^2}{i^\alpha}\sum_{j=1}^k\frac{1}{j^{2\alpha}}\nonumber\\
    &=\frac{1}{i^\alpha}\left[\sigma_{b^{(l)}}^2+\sigma_{w^{(l)}}^2\sigma_{l-1}^2\sum_{j=1}^k\frac{1}{j^{2\alpha}}\right].\label{sigma_l=1}
\end{align}
This series converges for $\alpha>1/2$, and we define the limit for $i=1$ as $\sigma_l^2$. Thus, $S_k$ is indeed a $L^2$ bounded martingale and trivially $\mathbb E F_i^{(l)}=0$, proving Assumption \ref{assumption_finite_variance}.

We next show Property \ref{assumption_cts_vf}. For the first layer, we use independence to get
\begin{align}
    \mathbb E\left[(f_i^{(1)}(x)-f_i^{(1)}(y))^2\right]&=\sum_{j=1}^d\mathbb E\left[(W_{i,j}^{(1)})^2\right](x_j-y_j)^2\nonumber\\
    &=\frac{\sigma_{w_1}^2}{i^\alpha}\lVert x-y\rVert^2.\label{c_1=sigma_w_1}
\end{align}

For the subsequent layers, we again use induction over $l$. We define $S_k(x)$ as before and check that
\begin{align}
    \mathbb E\left[(S_k(x)S_k(y))^2\right]=\mathbb E\left[(B_i^{(l)})^2\right]+\sum_{j=1}^k\mathbb E \left[(W_{i,j}^{(l)})^2\right]\mathbb E[\zeta(F_j^{(l-1)}(x))\zeta(F_j^{(l-1)}(y))].\label{S_k_x_y}
\end{align}
Using the induction hypothesis, Assumption \ref{assumption_bounded_activation}, \eqref{S_k} and \eqref{S_k_x_y} we get
\begin{align}
    \mathbb E\left[(S_k(x)-S_k(y))^2\right]
    &=
    \mathbb E\left[(S_k(x))^2\right]+\mathbb E\left[(S_k(y))^2\right]-2\mathbb E[S_k(x)S_k(y)]\nonumber\\
    &=2\mathbb E\left[(B_i^{(l)})^2\right] +\sum_{j=1}^k\mathbb E\left[(W_{i,j}^{(l)})^2\right]\left(\mathbb E[\zeta(F_j^{(l-1)}(x))^2]+\mathbb E[\zeta(F_j^{(l-1)}(y))^2]\right)\nonumber\\
    &\qquad-2E[S_k(x)S_k(y)]\nonumber\\
    &=\sigma_{w^{(l)}}^2\sum_{j=1}^k\frac{1}{(ij)^\alpha}\mathbb E\left[(\zeta(F_j^{(l-1)}(x))-\zeta(F_j^{(l-1)}(y)))^2\right]\nonumber\\
    &\leq\sigma_{w^{(l)}}^2\sum_{j=1}^k\frac{1}{(ij)^\alpha}\mathbb E\left[(F_j^{(l-1)}(x)-F_j^{(l-1)}(y))^2\right]\nonumber\\
    &\leq\frac{\sigma_{w^{(l)}}^2c_{l-1}}{i^\alpha}\lVert x-y\rVert^2\sum_{j=1}^k\frac{1}{j^{2\alpha}}\\
    &=\frac{1}{i^\alpha}\left[\sigma_{w^{(l)}}^2c_{l-1}\sum_{j=1}^k\frac{1}{j^{2\alpha}}\right]\lVert x-y\rVert^2,\label{c=rho}
\end{align}
such that the claim follows upon defining $c_l=\sigma_{w^{(l)}}^2c_{l-1}\sum_{j=1}^\infty1/j^{2\alpha}$, and noting that by the Fatou's lemma
\begin{align*}
    \mathbb E\left[(F_i^{(l)}(x)-F_i^{(l)}(y))^2\right]
    &=\mathbb E\left[\liminf_{k\rightarrow\infty}(S_k(x)-S_k(y))^2\right]\\
    &\leq\liminf_{k\rightarrow\infty}\mathbb E\left[(S_k(x)-S_k(y))^2\right]\\
    &\leq\liminf_{k\rightarrow\infty}\frac{1}{i^\alpha}\left[\sigma_{w^{(l)}}^2c_{l-1}\sum_{j=1}^k\frac{1}{j^{2\alpha}}\right]\lVert x-y\rVert^2\\
    &=c_l\lVert x-y\rVert^2.
\end{align*}

Lastly, recall that by Assumption \ref{assumption_bounded_activation} the activation functions are Lipschitz continuous, and thus so is $v$ as a composition of Lipschitz functions. The claim of Property \ref{assumption_information_passes} for the finite width case now follows since $\mu_0$-almost surely, $v$ is Lipschitz continuous and thus differentiable almost everywhere by the Rademacher Theorem \citep[Theorem 3.1.6]{federer1969geometric}.
\end{proof}

\subsection{Lemma \ref{lem:node_swap}}
In networks with small widths, Algorithm \ref{alg:prior_swaps} gave acceptance rates of around $30\%$ (for $N^{(l)}=10$), which quickly declined as we included more nodes (e.g. $1\%$ acceptances for $N^{(l)}=100$.)
 This suggests that the NodeSwap algorithm is not well-defined in the infinite width limit, and this is indeed the statement of the next lemma. We will from now on write fraktal letters for the swapped nodes $f^{(\mathfrak l)}_{\mathfrak i}$ and $f^{(\mathfrak l)}_{\mathfrak i+1}$, and reserve $i$ and $l$ for general indices.
 
\begin{lem}\label{lem:node_swap}
The NodeSwap Algorithm \ref{alg:prior_swaps} which swaps the biases and weights associated with the nodes $f^{(\mathfrak l)}_{\mathfrak i}$ and $f^{(\mathfrak l)}_{\mathfrak i+1}$ is not well defined in the infinite width limit. 

For the finite width network, the acceptance ratio is given by 
\begin{align}\label{eq:acc_ratio}
    a^N(\theta^N,\vartheta^N)=\frac{\mu^N_0(\vartheta^N)}{\mu^N_0(\theta^N)}.
\end{align} 
\end{lem}
\begin{proof}
By \cite{tierney1998note}, one needs to check that the measures $\eta(d\theta,d\vartheta):=\mu(\theta)Q(\theta,d\vartheta)$ and $\eta^T(d\theta,d\vartheta):=\eta(d\vartheta,d\theta)=\mu(d\vartheta)Q(\vartheta,d\theta)$ are mutually absolutely continuous on a set $R\in(E\times E,\mathcal E\otimes\mathcal E)$, and mutually singular on $R^C$, where here $Q$ is the deterministic transition kernel, and $(E,\mathcal E)$ is the measurable space on which $\mu$ and $Q$ are defined.\footnote{We use a different notation to \cite{tierney1998note}: Our $\eta$ is his $\mu$, our $mu$ is his $\pi$, our $(\theta,\vartheta)$ is his $(x,y)$.}
The (deterministic) transition kernel $Q$ maps $\theta$ to $\vartheta$ by swapping the nodes $f^{(\mathfrak l)}_{\mathfrak{ij}}$ and $f^{(\mathfrak l)}_{\mathfrak i(\mathfrak j+1)}$ (or more precisely, their associated weights and biases) with probability 
\begin{align}
    \frac{1}{n}\times\frac{1}{i^\alpha}\frac{1}{\sum_{j=1}^{N^l}\frac{1}{j^\alpha}},
\end{align}
which is well defined as $N^l\rightarrow\infty$, and independent of $\theta$, such that it suffices to show that the measures $\mu(\theta)$ and $\mu^T(\theta)=\mu(\vartheta)$ are mutually absolutely continuous on a set $R_1\in(E,\mathcal E)$, and mutually singular on $R_1^C$. The likelihood is also invariant under the transformation $\theta\mapsto\vartheta$, and as it is integrable with respect to the prior by the assumptions in Section \ref{MH_pCN} \citep{stuart2010inverse}, we only need to show that the Gaussian measures $\mu_0(d\theta)$ and $\mu_0^T(d\theta)$ are absolutely continuous with respect to one another. Note that we can write these as 
\begin{align}
    \mu_0(d\theta)&=\mathcal N(0,\mathcal C)\\
    \mu_0^T(d\theta)&=\mathcal N(0,\tilde{\mathcal C}),
\end{align}
with diagonal (by assumption) covariance operators $\mathcal C$ and $\tilde{\mathcal C}$, where the latter arises from swapping the variances associated with the swapped nodes. To see what is going on exactly, we now change to the neural network notation, where the variances under $\mathcal C$ for the individual weights and biases were given by
\begin{align}
    W_{i,j}^{(1)}\sim\mathcal N\left(0,\frac{\sigma_{w^{(1)}}^2}{i^\alpha}\right),\quad
    W_{i,j}^{(l)}\sim\mathcal N\left(0,\frac{\sigma_{w^{(l)}}^2}{(ij)^\alpha}\right)~\text{for}~l=2\dots n+1,\quad
    B_{i}^{(l)}\sim\mathcal N\left(0,\frac{\sigma_{b^{(l)}}^2}{i^\alpha}\right).
\end{align}
The variances under $\tilde{\mathcal C}$ are the same for most weights and biases, changed are only those associated with the swap nodes (recall that we swap nodes $f^{(\mathfrak l)}_{\mathfrak{ij}}$ and $f^{(\mathfrak l)}_{\mathfrak i(\mathfrak j+1)}$). The only changed variances are
\begin{align}
    W_{\mathfrak i,j}^{(\mathfrak l)}\sim\mathcal N\left(0,\frac{\sigma_{w^{(\mathfrak l)}}^2}{((\mathfrak i+1)j)^\alpha}\right),\qquad W_{(\mathfrak i+1),j}^{(\mathfrak l)}\sim\mathcal N\left(0,\frac{\sigma_{w^{(\mathfrak l)}}^2}{(\mathfrak ij)^\alpha}\right),\qquad\forall j\in\mathbb N\\
    W_{i,\mathfrak j}^{(\mathfrak l+1)}\sim\mathcal N\left(0,\frac{\sigma_{w^{(\mathfrak l+1)}}^2}{(i(\mathfrak j+1))^\alpha}\right),\qquad W_{i(\mathfrak j+1)}^{(\mathfrak l+1)}\sim\mathcal N\left(0,\frac{\sigma_{w^{(\mathfrak l+1)}}^2}{(i\mathfrak j)^\alpha}\right),\qquad\forall i\in\mathbb N\\
    B_{\mathfrak i}^{(\mathfrak l)}\sim\mathcal N\left(0,\frac{\sigma_{b^{(l)}}^2}{(\mathfrak i+1)^\alpha}\right)
    \qquad 
    B_{\mathfrak i+1}^{(\mathfrak l)}\sim\mathcal N\left(0,\frac{\sigma_{b^{(l)}}^2}{\mathfrak i^\alpha}\right),
\end{align}
which corresponds to swapping all the weights going into the nodes, swapping all the weights leaving the nodes, and swapping the biases of the nodes, respectively (see Figure \ref{fig:NN} for an illustration).

We apply the Feldman-Hajek Theorem \citep[Theorem 2.25]{da2014stochastic} to prove that these two Gaussian measures are mutually singular, by showing that the operator $(\mathcal C^{-1/2}\tilde{\mathcal C}^{1/2})(\mathcal C^{-1/2}\tilde{\mathcal C}^{1/2})^*$ is \emph{not} a Hilbert-Schmidt operator. Due to the diagonality of $\mathcal C$ and $\tilde{\mathcal C}$ the operator would be a Hilbert-Schmidt operator if
\begin{align}
    \sum_{i=1}^\infty\left(\frac{\tilde\lambda_i^2}{\lambda_i^2}-1\right)^2<\infty.
\end{align}
We only need to check those terms where $\tilde\lambda_i\neq\lambda_i$. Again looking at only the eigenvalues corresponding to the weights going into the swapped nodes, and switching to the neural network parametrisation, we have
\begin{align}
    \sum_{j=1}^\infty\left(\frac{(\mathfrak ij)^\alpha}{((\mathfrak i+1)j)^\alpha}-1\right)^2+\sum_{j=1}^\infty\left(\frac{((\mathfrak i+1)j)^\alpha}{(\mathfrak ij)^\alpha}-1\right)^2=\infty,
\end{align}
such that the operator is \emph{not} a Hilbert-Schmidt operator, and the Gaussian measures are mutually singular.

For the interested reader, note that the other two conditions of the Feldman-Hajek Theorem \citep[Theorem 2.25]{da2014stochastic} are satisfied. First we show only that there exist constants $L$ and $U$ such that for any $\theta\in\ell^2$,
\begin{align}
    L|\tilde{\mathcal C}\theta|\leq|C\theta|\leq U|\tilde{\mathcal C}\theta|,
\end{align} 
which is equivalent to
\begin{align}\label{eq:LU_ineq}
    L\sum_{i=1}^\infty\left(\theta_i\tilde\lambda_i^2\right)^2\leq
    \sum_{i=1}^\infty\left(\theta_i\lambda_i^2\right)^2\leq
    U\sum_{i=1}^\infty\left(\theta_i\tilde\lambda_i^2\right)^2,
\end{align}
where $\lambda_i^2$ are the respective variances corresponding to the values. Firstly note that we only need to consider those terms for which $\tilde\lambda_i^2\neq\lambda^2_i$. Using the neural network parametrisation, we can split the problem in showing that \eqref{eq:LU_ineq} holds for A) all the weights going into the swapped nodes, B) all the weights leaving the swapped nodes, and C) swapping the biases.
Looking at the weights going into the swapped nodes, note that 
\begin{align*}
    \sum_{j=1}^\infty\left(\frac{w_{\mathfrak ij}^{(\mathfrak l)}}{((\mathfrak i+1)j)^\alpha}\right)^2+\sum_{j=1}^\infty\left(\frac{w_{(\mathfrak i+1)j}^{(\mathfrak l)}}{(\mathfrak ij)^\alpha}\right)^2
    &=\sum_{j=1}^\infty\left(\frac{w_{\mathfrak ij}^{(\mathfrak l)}}{(\mathfrak ij)^\alpha}\right)^2\left(\frac{(\mathfrak ij)^\alpha}{((\mathfrak i+1)j)^\alpha}\right)^2\\
    &\qquad+\sum_{j=1}^\infty\left(\frac{w_{(\mathfrak i+1)j}^{(\mathfrak l)}}{((\mathfrak i+1)j)^\alpha}\right)^2\left(\frac{((\mathfrak i+1)j)^\alpha}{(\mathfrak ij)^\alpha}\right)^2\\
    &\leq\sum_{j=1}^\infty\left(\frac{w_{\mathfrak ij}^{(\mathfrak l)}}{(\mathfrak ij)^\alpha}\right)^2+2^{2\alpha}\sum_{j=1}^\infty\left(\frac{w_{(\mathfrak i+1)j}^{(\mathfrak l)}}{((\mathfrak i+1)j)^\alpha}\right)^2\\
    &\leq2^{2\alpha}\left[\sum_{j=1}^\infty\left(\frac{w_{\mathfrak ij}^{(\mathfrak l)}}{(\mathfrak ij)^\alpha}\right)^2+\sum_{j=1}^\infty\left(\frac{w_{(\mathfrak i+1)j}^{(\mathfrak l)}}{((\mathfrak i+1)j)^\alpha}\right)^2\right],
\end{align*}
and 
\begin{align*}
    \sum_{j=1}^\infty\left(\frac{w_{\mathfrak ij}^{(\mathfrak l)}}{(\mathfrak ij)^\alpha}\right)^2+\sum_{j=1}^\infty\left(\frac{w_{(\mathfrak i+1)j}^{(\mathfrak l)}}{((\mathfrak i+1)j)^\alpha}\right)^2
    &=\sum_{j=1}^\infty\left(\frac{w_{\mathfrak ij}^{(\mathfrak l)}}{((\mathfrak i+1)j)^\alpha}\right)^2\left(\frac{((\mathfrak i+1)j)^\alpha}{(\mathfrak ij)^\alpha}\right)^2\\
    &\qquad+\sum_{j=1}^\infty\left(\frac{w_{(\mathfrak i+1)j}^{(\mathfrak l)}}{(\mathfrak ij)^\alpha}\right)^2\left(\frac{(\mathfrak ij)^\alpha}{((\mathfrak i+1)j)^\alpha}\right)^2\\
    &\leq2^{2\alpha}\sum_{j=1}^\infty\left(\frac{w_{\mathfrak ij}^{(\mathfrak l)}}{((\mathfrak i+1)j)^\alpha}\right)^2+\sum_{j=1}^\infty\left(\frac{w_{(\mathfrak i+1)j}^{(\mathfrak l)}}{(\mathfrak ij)^\alpha}\right)^2\\
    &\leq2^{2\alpha}\left[\sum_{j=1}^\infty\left(\frac{w_{\mathfrak ij}^{(\mathfrak l)}}{((\mathfrak i+1)j)^\alpha}\right)^2+\sum_{j=1}^\infty\left(\frac{w_{(\mathfrak i+1)j}^{(\mathfrak l)}}{(\mathfrak ij)^\alpha}\right)^2\right]
\end{align*}
such that for the weights going into the swapped nodes, \eqref{eq:LU_ineq} holds with $L=2^{-2\alpha}$ and $U=2^{2\alpha}$. Repeating the same argument for the weights leaving the swapped nodes and for the biases, shows that \eqref{eq:LU_ineq} holds in general with $L=2^{-2\alpha}$ and $U=2^{2\alpha}$.
The remaining condition of the Feldman-Hajek theorem addresses the difference of means, but as $\theta=0$ this is clearly in the Cameron-Martin space of the prior.\\

For the acceptance ratio in the finite width networks, observe that the likelihood does not depend on the labelling of the nodes and thus plays no role in the acceptance probability. Similarly, the transition kernel is symmetric, as nodes $f^{(\mathfrak l)}_{\mathfrak i}$ and $f^{(\mathfrak l)}_{\mathfrak i+1}$ are swapped with probability
\begin{align}
    q^N(\vartheta^N|\theta^N)&=\mathbb P(\{\text{layer }\mathfrak l\text{ gets chosen}\})\times\mathbb P(\{\text{node }\mathfrak i\text{ gets chosen}\}|\{\text{layer }\mathfrak l\text{ got chosen}\})\\
    &=\frac{1}{n+1}\times\frac{1}{i^\alpha}\frac{1}{\sum_{j=1}^{N^l}\frac{1}{j^\alpha}}.
\end{align}
For the finite dimensional case we thus get
\begin{align}\label{eq:nodeswapinvariance}
    a^N(\theta^N,\vartheta^N)=\frac{\mu^N_0(\vartheta^N)}{\mu^N_0(\theta^N)}\frac{\mathcal L^N(\vartheta^N)}{\mathcal L^N(\theta^N)}\frac{q^N(\theta^N|\vartheta^N)}{q^N(\vartheta^N|\theta^N)}=\frac{\mu^N_0(\vartheta^N)}{\mu^N_0(\theta^N)},
\end{align}
which is as required.
\end{proof}

\subsection{Proof of Theorem \ref{thm:bound_likelihood_derivative}\label{proof:thm:bound_likelihood_derivative}}
\begin{proof}[Proof of Theorem \ref{thm:bound_likelihood_derivative}.]
For a given data point $y=(x,a)$, let the actions be enumerates such that $a=1$. Let further $v=(v_1,\dots,v_M):=(v(\mathcal T(x,1)),\dots,v(\mathcal T(x,M)))$ be the vector of the value function evaluations relevant for the likelihood computation. The integral \eqref{likelihood_uncomputable} is trivially upper bounded by $1$. Define $\bar v=\max_j|v_j|$. For the lower bound, we use \eqref{likelihood_computable} to get 

\begin{align}
    \nonumber\eqref{likelihood_uncomputable} &\geq \frac{1}{\sigma} \int_{-\infty}^{\infty}\phi\left(\frac{t-v_1}{\sigma}\right)\frac{1}{2^{M-1}}\prod_{j=2}^M\mathbbm{1}_{\{\Phi((t-v_j)/\sigma)\geq1/2\}}dt \\
    \nonumber&= \frac{1}{\sigma2^{M-1}} \int_{-\infty}^{\infty}\phi\left(\frac{t-v_1}{\sigma}\right)\prod_{j=2}^M\mathbbm{1}_{\{t\geq v_j\}}dt\\
    \nonumber&\geq \frac{1}{\sigma2^{M-1}} \int_{-\infty}^{\infty}\phi\left(\frac{t-v_1}{\sigma}\right)\prod_{j=1}^M\mathbbm{1}_{\{t\geq v_j\}}dt\\
    \nonumber& \geq \frac{1}{\sigma2^{M-1}} \int_{\bar v}^\infty \phi\left(\frac{t-v_1}{\sigma}\right) dt \geq \frac{1}{\sigma2^{M-1}} \int_{2\bar v}^\infty \phi\left(\frac{t}{\sigma}\right) dt\\
    &\geq \frac{1}{\sigma2^{M+1}\sqrt{2\pi\sigma^2}\bar v} \exp(-(4\bar v^2/(2\sigma^2)).\label{bound_derivative_p}
\end{align}

Since $v$ is in a reproducing kernel Hilbert space $\mathcal H$, there exists for any $x\in\mathcal X$ a $C_x$ such that $|v(x)|\leq C_x \lVert v\rVert_{\mathcal H}$ for all $v\in\mathcal H$ \cite[Chapter 12]{wainwright2019high}, and taking $C=\max_{j\in \{1,\dots,M\}}C_{\mathcal T(x,j)}$, we have $\bar v\leq C\lVert v\rVert_{\mathcal H}$.

Taking logarithms of \eqref{bound_derivative_p}, we thus have
\begin{align}
    \ell(y|v,\sigma)&\overset{\text{\eqref{bound_derivative_p}}}{\geq} -\log(\sigma 2^M\sqrt{2\pi\sigma^2})-\log(\bar v)-\frac{\bar v^2}{2\sigma^2}\label{ineq:lower_bound_ell}\\
    &\geq -\log(\sigma 2^M\sqrt{2\pi\sigma^2})-\left(1+\frac{1}{2\sigma^2}\right)\cdot\bar v^2\nonumber\\
    &\geq -\log(\sigma 2^M\sqrt{2\pi\sigma^2})-C\cdot\left(1+\frac{1}{2\sigma^2}\right)\cdot \lVert v\rVert_{\mathcal H}^2\nonumber
\end{align}
showing that Assumption \ref{assumption_likelihood1} holds with $K=\max\left\{\log(\sigma 2^M\sqrt{2\pi\sigma^2}),C\cdot\left(1+\frac{1}{2\sigma^2}\right)\right\}$ and $p=2$.

To see that Assumption \ref{assumption_likelihood2} holds, assume that $\max\{\lVert u\rVert_{\mathcal H},\lVert v\rVert_{\mathcal H}\}<r$. Then, since the log-likelihood is continuously differentiable in $(v_1,\dots,v_M)$, for any $\bar{r}$ there exists a constant $C(\bar{r})$ such that for any vectors  $u_{1:M}$, $v_{1:M}$ with $\max_j|u_j|\leq\bar{r}$, $\max_j|v_j|\leq\bar{r}$ one has by the mean value theorem that
\begin{align}
     |\ell(y|u,\sigma)-\ell(y|v,\sigma)|
    &\leq C(\bar{r})\cdot \bigl(|u_1-v_1|+\dots+|u_M-v_M|\bigr)\label{eq:MVT}.
\end{align}
Using the RKHS property as before, we note that $\bar u=\max_j|u_j|<\max_j C_{\mathcal T(x,j)}r$ and $\bar v=\max_j|v_j|<\max_j C_{\mathcal T(x,j)}r$. We also use the fact that for any $x\in\mathcal X$ there exists a $C_x$ such that $|(u-v)(x)|\leq C_x \lVert u-v\rVert_{\mathcal H}$ for all $u,v\in\mathcal H$. Taking $\bar{r}=\max_j C_{\mathcal T(x,j)}r>0$ we thus get
\begin{align*}
    |\ell(y|u,\sigma)-\ell(y|v,\sigma)|
    &\overset{\ref{eq:MVT}}{\leq} C(\bar{r})\cdot \bigl(|u_1-v_1|+\dots+|u_M-v_M|\bigr)\\
    &= C(\bar{r})\cdot \bigl(|(u-v)(\mathcal T(x,1))|+\dots+|(u-v)(\mathcal T(x,M))|\bigr)\\
    &\leq C(\bar{r}) \cdot\sum_{j=1}^M C_{\mathcal T(x,j)} \cdot \lVert u-v\rVert_{\mathcal H},
\end{align*}
such that the assumption holds with $K(r)=C(\bar{r}) \cdot\sum_{j=1}^M C_{\mathcal T(x,j)}$.
\end{proof}

\subsection{Proof of Lemma \ref{lemma:likelihood}\label{proof:lemma:likelihood}}
\begin{proof}[Proof of Lemma \ref{lemma:likelihood}.]
Let $\theta=(w,b)$ be the collection of all weights and biases. Using the definition of the neural network \eqref{def:NNfunctions}, we let $\hat{x}=(1,x)\in\mathbb{R}^{d+1}$ and note that
\begin{align*}
    |f_i^{(1)}(x)|^2
    =|\langle (b_i^{(1)},w_{i,:}^{(1)}),\hat x\rangle|^2
    \leq \lVert \hat x\rVert_2^2 \cdot \lVert (b_i^{(1)},w_{i,:}^{(1)}) \rVert_2^2
    \leq  (1+\lVert x\rVert_2^2) \cdot \lVert (b_i^{(1)},w_{i,:}^{(1)}) \rVert_2^2
\end{align*}
by the Cauchy-Schwartz inequality (CSI). We now note that, regardless of the choice of $N^{(1)}$,
\begin{align*}
    \lVert f_{1:N^{(1)}}^{(1)}(x)\rVert_2^2 = \sum_{i=1}^{N^{(1)}}|f_i^{(1)}(x)|^2\leq  (1+\lVert x\rVert_2^2) \cdot \sum_{i=1}^{N^{(1)}}\lVert (b_i^{(1)},w_{i,:}^{(1)}) \rVert_2^2 \leq (1+\lVert x\rVert_2^2)\cdot \lVert\theta\rVert_{\ell^2}^2,
\end{align*}
such that the result holds also for the limit $N^{(1)}\rightarrow\infty$. 
For the higher layers, we use Assumption \ref{assumption_bounded_activation} and get for any $l$ that $|\zeta(f_i^{(l)}(x))|^2\leq |f_i^{(l)}(x)|^2$. We apply the CSI a few more times, and get that
\begin{align*}
    |f_i^{(l)}(x)|^2&= |b_i^{(l)}+\sum_{j=1}^{N^{(l-1)}}w_{i,j}^{(l)}\zeta(f_j^{(l-1)}(x))|^2
    \\
    &\leq (1+\lVert \zeta(f_{1:N^{(l-1)}}^{(l-1)}(x))\rVert_2^2) \cdot \lVert (b_i^{(l)},w_{i,:}^{(l)}) \rVert_2^2\\
    &\leq (1+\lVert f_{1:N^{(l-1)}}^{(l-1)}(x)\rVert_2^2) \cdot \lVert (b_i^{(l)},w_{i,:}^{(l)}) \rVert_2^2,
\end{align*}
and that 
\begin{align*}
    \lVert f_{1:N^{(l)}}^{(l)}(x)\rVert_2^2&=\sum_{i=1}^{N^{(l)}}|f_i^{(l)}(x)|^2\leq  (1+ \lVert \zeta(f_{1:N^{l-1}}^{(l-1)}(x))\rVert_2^2) \cdot \sum_{i=1}^{N^{(l)}}\lVert (b_i^{(l)},w_{i,:}^{(l)}) \rVert_2^2\\
    &\leq (1+\lVert f_{1:N^{l-1}}^{(l-1)}(x)\rVert_2^2)\cdot \lVert\theta\rVert_{\ell^2}^2.
\end{align*}
For any $\theta$ with $\lVert\theta\rVert_{\ell^2}^{2}<1$, we use induction and get that $|f_i^{(l)}(x)|^2\leq l+\lVert x\rVert^2$. If $\lVert\theta\rVert_{\ell^2}^{2}\geq1$, we get again by induction that $|f_i^{(l)}(x)|^2\leq (l+\lVert x\rVert^2)\cdot\lVert\theta\rVert_{\ell^2}^{2l}$; such that for any $\theta$,
\begin{align}\label{eq:bound_f_l}
    \lVert f_{1:N^{(l)}}^{(l)}(x)\rVert_2^2=\sum_{i=1}^{N^{(l)}}|f_i^{(l)}(x)|^2\leq (l+\lVert x\rVert^2)\cdot\bigl(1+\lVert\theta\rVert_{\ell^2}^{2l}\bigr),
\end{align}
in particular for $v(x)=f_1^{(n+1)}(x)$ we have
\begin{align*}
    |v(x)|^2\leq (n+1+\lVert x\rVert^2)\cdot\bigl(1+\lVert\theta\rVert_{\ell^2}^{2(n+1)}\bigr).
\end{align*}
Using the same bound for $\ell(y|v,\sigma)$ as in the proof of Theorem \ref{thm:bound_likelihood_derivative} given in \eqref{ineq:lower_bound_ell}, we get
\begin{align*}
    \ell(y|v,\sigma)&\geq -\log(\sigma 2^M\sqrt{2\pi\sigma^2})-\left(1+\frac{1}{2\sigma^2}\right)\cdot\bar v^2\\
    &\geq -\log(\sigma 2^M\sqrt{2\pi\sigma^2})-(n+1+\max_j\lVert\mathcal T(x,j)\rVert^2)\cdot\left(1+\frac{1}{2\sigma^2}\right)\cdot \bigl(1+\lVert \theta\rVert_{\ell^2}^{2(n+1)}\bigr),
\end{align*}
such that the result holds with $K=(n+1+\max_j\lVert\mathcal T(x,j)\rVert^2)\cdot\left(1+\frac{1}{2\sigma^2}\right)+\max\left\{\log(\sigma 2^M\sqrt{2\pi\sigma^2}),0\right\}$ and $p=2(n+1)$. Note that the constant $K$ is independent of the layer width and the result holds for networks of arbitrary width.

To prove Assumption \ref{assumption_likelihood2}, fix $r>0$ and consider the sequences $\theta,\tilde{\theta}\in\ell^{2}$ such that $\max\{\lVert \theta\}\rVert ,\lVert \tilde{\theta}\rVert\}\leq r$. Let $u=u_\theta$ be the neural network arising from the parameters $\theta$, and let $v=v_{\tilde{\theta}}$ be the neural network arising from the parameters $\tilde{\theta}$. The difference in the output of the final layers of the neural network is 
\begin{align*}
u(x)-v(x) & = b_{1}^{(n+1)}-\tilde{b}_{1}^{(n+1)}+\sum_{j=1}^{N^{(n)}}w_{1,j}^{(n+1)}\zeta(f_{j}^{(n)}(x))-\sum_{j=1}^{N^{(n)}}\tilde{w}_{1,j}^{(n+1)}\zeta(\tilde{f}_{j}^{(n)}(x))\\
 & =b_{1}^{(n+1)}-\tilde{b}_{1}^{(n+1)}+\sum_{j=1}^{N^{(n)}}\left(w_{1,j}^{(n+1)}-\tilde{w}_{1,j}^{(n+1)}\right)\zeta(f_{j}^{(n)}(x))\\
 &\qquad+\sum_{j=1}^{N^{(n)}}\tilde{w}_{1,j}^{(n+1)}\left(\zeta(f_{j}^{(n)}(x))-\zeta(\tilde{f}_{j}^{(n)}(x))\right)
\end{align*}
where the functions within the neural network defined by $\tilde{\theta}$ are distinguished by a tilde on each of them. We can bound the squared difference by
\begin{align*}
\frac{1}{2}\left(u(x)-v(x)\right)^{2} &\leq \left(b_{1}^{(n+1)}-\tilde{b}_{1}^{(n+1)}+\sum_{j=1}^{N^{(n)}}\left(w_{1,j}^{(n+1)}-\tilde{w}_{1,j}^{(n+1)}\right)\zeta(f_{j}^{(n)}(x))\right)^2\\
 &\qquad+\left(\sum_{j=1}^{N^{(n)}}\tilde{w}_{1,j}^{(n+1)}\left(\zeta(f_{j}^{(n)}(x))-\zeta(\tilde{f}_{j}^{(n)}(x))\right)\right)^2
\end{align*}
and using the CSI further by
\begin{align*}
&\leq\left(\left(b_{1}^{(n+1)}-\tilde{b}_{1}^{(n+1)}\right)^{2}+\sum_{j=1}^{N^{(n)}}\left(w_{1,j}^{(n+1)}-\tilde{w}_{1,j}^{(n+1)}\right)^{2}\right) \cdot \left(1+\sum_{j=1}^{N^{(n)}}\zeta(f_{j}^{(n)}(x))^{2}\right)\\
&\qquad +\left(\sum_{j=1}^{N^{(n)}}\left(\zeta(f_{j}^{(n)}(x))-\zeta(\tilde{f}_{j}^{(n)}(x))\right)^{2}\right) \cdot \left(\sum_{j=1}^{N^{(n)}}\left(\tilde{w}_{1,j}^{(n+1)}\right)^{2}\right)\\
 & \leq\left(\left(b_{1}^{(n+1)}-\tilde{b}_{1}^{(n+1)}\right)^{2}+\sum_{j=1}^{N^{(n)}}\left(w_{1,j}^{(n+1)}-\tilde{w}_{1,j}^{(n+1)}\right)^{2}\right) \cdot K(r,x,n)\\
 &\qquad+\left(\sum_{j=1}^{N^{(n)}}\left(f_{j}^{(n)}(x)-\tilde{f}_{j}^{(n)}(x)\right)^{2}\right) \cdot r^{2}\\
 &\leq \lVert\theta^{(n+1)}-\tilde{\theta}^{(n+1)}\rVert^2_{\ell^2}\cdot K(r,x,n)+\left(\sum_{j=1}^{N^{(n)}}\left(f_{j}^{(n)}(x)-\tilde{f}_{j}^{(n)}(x)\right)^{2}\right) \cdot r^{2},
\end{align*}
where the last inequality assumes $1+\sum_{j}\zeta(f_{j}^{(n)}(x))^{2}\leq K(r,x,n)$,
which will be verified next, and also uses the bound $\max\{\lVert \theta\rVert ,\lVert \tilde{\theta}\rVert \}\leq r.$
In \eqref{eq:bound_f_l} it was shown that 
\begin{align*}
1+\sum_{j=1}^{N^{(n)}}\zeta(f_{j}^{(n)}(x))^{2}&\leq
1+\sum_{j=1}^{N^{(n)}}(f_{j}^{(n)}(x))^{2}=1+\lVert f^{(n)}_{1:N^{(n)}}(x)\rVert_2^2\\
&\leq 1+ (n+\lVert x\rVert^2)\cdot\bigl(1+\lVert\theta\rVert_{\ell^2}^{2n}\bigr)\leq K(r,x,n)
\end{align*}
by setting $K(r,x,n):=1+ (n+\lVert x\rVert^2)\cdot\bigl(1+r^{2n}\bigr)$. 

The decomposition thus far articulates how $(u(x)-v(x))^2$  depends on the difference of the weights and biases of the output layer (layer $n+1$). We may similarly articulate how $\sum_j(f_j^{(n)}(x) - \tilde{f}_j^{(n)}(x))^2$ depends on the difference of the weights and biases of the previous layers. For example,
\begin{align*}
    \frac12 \Bigl(f_{i}^{(n)}(x)&-\tilde{f}_{i}^{(n)}(x)\Bigr)^{2}
    = \frac12 \left(b_i^{(n)}-\tilde{b}_i^{(n)} + \sum_{j=1}^{N^{(n-1)}}w_{i,j}^{(n)}\zeta(f_j^{(n-1)}(x))-\tilde{w}_{i,j}^{(n)}\zeta(\tilde{f}_j^{(n-1)}(x))\right)^2\\
    &\leq\left(\left(b_{i}^{(n)}-\tilde{b}_{i}^{(n)}\right)^{2} + \sum_{j=1}^{N^{(n-1)}}\left(w_{i,j}^{(n)}-\tilde{w}_{i,j}^{(n)}\right)^{2}\right) \cdot \left(1+\sum_{j=1}^{N^{(n-1)}}\left(f_{j}^{(n-1)}(x)\right)^{2}\right)\\
    &\qquad +\left(\sum_{j=1}^{N^{(n-1)}}\left(f_{j}^{(n-1)}(x)-\tilde{f}_{j}^{(n-1)}(x)\right)^{2}\right) \cdot \left(\sum_{j=1}^{N^{(n-1)}}\left(\tilde{w}_{i,j}^{(n)}\right)^{2}\right),
\end{align*} 
and summing over $i$ gives
\begin{align*}
    \frac12 \sum_{i=1}^{N^{(n)}}& \left(f_{i}^{(n)}(x)-\tilde{f}_{i}^{(n)}(x)\right)^{2}\\
    &\leq \sum_{i=1}^{N^{(n)}}\left(\left(b_{i}^{(n)}-\tilde{b}_{i}^{(n)}\right)^{2} + \sum_{j=1}^{N^{(n-1)}}\left(w_{i,j}^{(n)}-\tilde{w}_{i,j}^{(n)}\right)^{2}\right) \cdot \left(1+\sum_{j=1}^{N^{(n-1)}}\left(f_{j}^{(n-1)}(x)\right)^{2}\right)\\
    &\qquad +\left(\sum_{j=1}^{N^{(n-1)}}\left((f_{j}^{(n-1)}(x)-\tilde{f}_{j}^{(n-1)}(x)\right)^{2}\right) \cdot \sum_{i=1}^{N^{(n)}}\left(\sum_{j=1}^{N^{(n-1)}}\left(\tilde{w}_{i,j}^{(n)}\right)^{2}\right)\\
    &\leq \sum_{i=1}^{N^{(n)}}\left(\left(b_{i}^{(n)}-\tilde{b}_{i}^{(n)}\right)^{2} + \sum_{j=1}^{N^{(n-1)}}\left(w_{i,j}^{(n)}-\tilde{w}_{i,j}^{(n)}\right)^{2}\right) \cdot K(r,x,n-1) \\
    &\qquad +\left(\sum_{j=1}^{N^{(n-1)}}\left((f_{j}^{(n-1)}(x)-\tilde{f}_{j}^{(n-1)}(x)\right)^{2}\right) \cdot r^2\\
    &\leq \lVert\theta^{(n)}-\tilde{\theta}^{(n)}\rVert^2_{\ell^2} \cdot K(r,x,n-1) +\left(\sum_{j=1}^{N^{(n-1)}}\left((f_{j}^{(n-1)}(x)-\tilde{f}_{j}^{(n-1)}(x)\right)^{2}\right) \cdot r^2.
\end{align*}
In summary, we obtain $\frac{1}{2}\left(u(x)-v(x)\right)^{2}\leq\bar{K}(r,x,n)\lVert \theta-\tilde{\theta}\rVert_{\ell^2}^{2}$ for a constant $\bar K$ depending only on $r$, $x$, and $n$. In particular, when $\tilde{\theta}=0$ then $v=0$, which implies $\frac{1}{2}\left(u(x)\right)^{2}\leq\bar{K}(r,x,n)\lVert \theta\rVert_{\ell^2}^{2}$.

We conclude the proof similarly to the proof of Theorem \ref{thm:bound_likelihood_derivative}. Assume that $\max\{\lVert \theta\rVert_{\ell^2},\lVert \tilde{\theta}\rVert_{\ell^2}\}<r$, so that $\max_j|u_j|\leq r\sqrt{2\max_j\{ \bar{K}(r,\mathcal T(x,j),n)\}}$ and $\max_j|v_j|\leq r\sqrt{2\max_j\{ \bar{K}(r,\mathcal T(x,j),n)\}}$. 
Then using the mean value theorem, we note that for any $\bar{r}$ there exists a constant $C(\bar{r})$, such that for any vectors  $u_{1:M}$, $v_{1:M}$ with $\max_j|u_j|\leq\bar{r}$, $\max_j|v_j|\leq\bar{r}$ we have
\begin{align*}
    |\ell(y|u,\sigma)-\ell(y|v,\sigma)|
    &\leq C(\bar{r})\cdot \bigl(|u_1-v_1|^2+\dots+|u_M-v_M|^2\bigr)^{\frac12}\\
    &= C(\bar{r})\cdot \bigl(\sum_{j=1}^M(u(\mathcal T(x,j))-v(\mathcal T(x,j)))^2\bigr)^{\frac12}\\
    &\leq \sqrt{2} C(\bar{r}) \cdot \left(\sum_{j=1}^M \bar K(r,\mathcal T(x,j),n)\right)^{\frac12} \cdot \lVert \theta-\tilde{\theta}\rVert_{\ell^2}.
\end{align*}
The result holds by choosing $\bar{r}=r\sqrt{2\max_j\bar{K}(r,\mathcal T(x,j),n)}$.
\end{proof}

\subsection{Proof of Theorem \ref{thm:endomorphism_likelihood_derivative}\label{proof:thm:endomorphism_likelihood_derivative}}

\begin{proof}[Proof of Theorem \ref{thm:endomorphism_likelihood_derivative}.]

The equivalence of $\mathcal{N}(\mathcal{C}\mathcal{D}\ell(u),\mathcal{C})\simeq\mathcal{N}(0,\mathcal{C})$ $\mu_{0}$-almost surely for all $u$ will be shown by applying the Feldman-Hajek theorem \citep[Theorem 2.23]{da2014stochastic} which states that two Gaussian measures $\mathcal{N}(m_{1},Q)$ and $\mathcal{N}(m_{2},Q)$ are absolutely continuous with respect to one another if and only if $m_{1}-m_{2}\in Q^{1/2}(\mathcal{H})$ or $\sum_{i}(m_{1i}-m_{2i})^{2}/\lambda_{i}^{2}<\infty$ \citep{da2014stochastic} where $\lambda_{i}^{2}$ are the eigenvalues of $\mathcal{Q}$. For $m_{1}=\mathcal{C}\mathcal{D}\ell(u)$, $m_{2}=0$, and $\mathcal Q=\mathcal C$, this means showing
\begin{align}
\sum_{i=1}^{\infty}\frac{(\mathcal{C}\mathcal{D}\ell(u))_{i}^{2}}{\lambda_{i}^{2}}=\sum_{i=1}^{\infty}\frac{\lambda_{i}^{4}(\mathcal{D}\ell(u))_{i}^{2}}{\lambda_{i}^{2}}=\sum_{i=1}^{\infty}\lambda_{i}^{2}(\mathcal{D}\ell(u))_{i}^{2}\label{eq:var_Ss}
\end{align}
is finite for $\mu_{0}$-almost all $u$ \citep{da2014stochastic}. Note that $\mathcal{D}\ell(u)$ is the collection of partial derivatives with respect to each weight and bias parameter of the neural network.
We will show that the sequence of truncated sums of \eqref{eq:var_Ss} defines a submartingale that converges $\mu_0$-almost surely to a random variable with finite expectation.

We now specify the limiting neural network. As $\alpha>1$, we have $\mu_{0}$-almost surely, 
$\lVert W^{(n+1)}_{1,:} \rVert^2 + \lVert B^{(n+1)}_{:}\rVert^2+\sum_{l=1}^n \lVert W^{(l)}_{:,:} \rVert^2 + \lVert B^{(l)}_{:} \rVert^2< \infty$. The limiting neural network is thus defined to be, for $l>1$,
$f_i^{(l)}=\langle(1,\zeta(f^{(l-1)}_{:}),(B^{(l)}_i,W^{(l)}_{i,:})) \rangle_{\ell^2}$. Indeed $\sum_{i=1}^{\infty}(f_i^{(l)})^2<\infty$ and thus the definition is recursive.

Substituting both the eigenvalues of $\mathcal{C}$ and the derivatives with respect to the parameters of the neural network into Equation \eqref{eq:var_Ss} and truncating the sum gives
\begin{alignat}{1}
S_{s}= & \sigma_{b^{(n+1)}}^{2}\left(\frac{\partial\ell}{\partial B_{1}^{(n+1)}}\right)^{2}+\sum_{j=1}^{s}\frac{\sigma_{w^{(n+1)}}^{2}}{j^{\alpha}}\left(\frac{\partial\ell}{\partial W_{1,j}^{(n+1)}}\right)^{2}\nonumber \\
 & +\sum_{l=2}^{n}\sum_{i=1}^{s}\left[\frac{\sigma_{b^{(l)}}^{2}}{i^{\alpha}}\left(\frac{\partial\ell}{\partial B_{i}^{(l)}}\right)^{2}+\sum_{j=1}^{s}\frac{\sigma_{w^{(l)}}^{2}}{(ij)^{\alpha}}\left(\frac{\partial\ell}{\partial W_{i,j}^{(l)}}\right)^{2}\right] \nonumber\\
 & +\sum_{i=1}^{s}\left[\frac{\sigma_{b^{(1)}}^{2}}{i^{\alpha}}\left(\frac{\partial\ell}{\partial B_{i}^{(1)}}\right)^{2}+\sum_{j=1}^{d}\frac{\sigma_{w^{(1)}}^{2}}{i^{\alpha}}\left(\frac{\partial\ell}{\partial W_{i,j}^{(1)}}\right)^{2}\right].
 \label{eq:thm3_proof_Ss}
\end{alignat}
For the likelihood $\ell$ in (\ref{eq:loglikecontrol}) and $T=1$, we will show that $\lim_{s\rightarrow \infty} S_s$ exists and is finite $\mu_0$-almost surely so that the equivalence $\mathcal{N}(\mathcal{C}\mathcal{D}\ell(u),\mathcal{C})\simeq\mathcal{N}(0,\mathcal{C})$ follows (in fact we will show that $S_s$ converges to a $L^1$ random variable as $s\rightarrow\infty$); the case for $T>1$ follows similarly. 

To this end, observe that under the assumption of uniformly bounded partial derivatives of $\ell(a,\cdot)$ for all $a$, each partial derivative can be further bounded by 
\[
\left(\frac{\partial\ell}{\partial B_{i}^{(l)}}\right)^{2}\leq M\times c_t \times\sum_{k=1}^{M}\left(\frac{\partial u(x_{t}^{k})}{\partial B_{i}^{(l)}}\right)^{2}
\]
where $\sqrt{c_t}$ is the bound of partial derivatives of $\ell(a_{t},\cdot)$. We firstly calculate $\partial u/\partial W_{i,j}^{(l)}$ and $\partial u/\partial B_{i}^{(l)}$, for all $(i,j,l)$, where $u=u(x)$ is the output of the NN for an input $x\in\mathbb{R}^{d}$ - the input has been dropped for notational convenience. These derivatives can be cast as derivatives of $\partial u/\partial g_{i}^{(l)}$ since $u$ can be regarded as a function of $(g_{1}^{(l)},g_{2}^{(l)},\ldots)$ and only $g_{i}^{(l)}$ is a function of $W_{i,j}^{(l)}$ and $B_{i}^{(l)}$, that is $g_{i}^{(l)}=\zeta(f_{i}^{(l)})$ and $f_{i}^{(l)}=B_{i}^{(l)}+\sum_{k}W_{i,k}^{(l)}g_{k}^{(l-1)}$. Thus 
\begin{equation}
\frac{\partial u}{\partial W_{i,j}^{(l)}}=\frac{\partial u}{\partial g_{i}^{(l)}}\times\frac{\partial g_{i}^{(l)}}{\partial f_{i}^{(l)}}\times\frac{\partial f_{i}^{(l)}}{\partial W_{i,j}^{(l)}},\qquad\frac{\partial u}{\partial B_{i}^{(l)}}=\frac{\partial u}{\partial g_{i}^{(l)}}\times\frac{\partial g_{i}^{(l)}}{\partial f_{i}^{(l)}}. \label{eq:dv_dW}
\end{equation}
The next step is obtain a bound on $\partial u/\partial g_{k}^{(l)}$ for all $(l,k)$. To this end, assume all $\partial u/\partial g_{k}^{(l+1)}$ for $k=1,2,\ldots$ are available. The aim is to find $\partial u/\partial g_{k}^{(l)}$ (for all $k$) at the previous layer $l$:
\begin{alignat}{1}
\frac{\partial u}{\partial g_{k}^{(l)}} & =\sum_{j=1}^{\infty}\frac{\partial u}{\partial g_{j}^{(l+1)}}\frac{\partial g_{j}^{(l+1)}}{\partial g_{k}^{(l)}}\nonumber \\
 & =\sum_{j=1}^{\infty}\frac{\partial u}{\partial g_{j}^{(l+1)}}\frac{\partial\zeta(f_{j}^{(l+1)})}{\partial f_{j}^{l+1}}\frac{\partial f_{j}^{(l+1)}}{\partial g_{k}^{(l)}}\nonumber \\
 & =\sum_{j=1}^{\infty}\frac{\partial u}{\partial g_{j}^{(l+1)}}\times\frac{\partial\zeta(f_{j}^{(l+1)})}{\partial f_{j}^{(l+1)}}\times W_{j,k}^{(l+1)}
.\label{eq:true_dv}
\end{alignat}
Using Assumption \ref{assumption_bounded_activation}, we will employ the following non-negative upper bound $D_{k}^{(l)}$ for $\partial u/\partial g_{k}^{(l)}$, defined recursively as follows
\begin{alignat}{1}
D_{k}^{(n)} & :=\left|W_{1,k}^{(n+1)}\right|,\nonumber \\
D_{k}^{(l)} & :=\sum_{j=1}^{\infty}D_{j}^{(l+1)}\times\left|W_{j,k}^{(l+1)}\right|,\qquad l=n-1,\dots, 1\label{eq:bound_dv}
\end{alignat}
(These can be shown to be finite bounds as follows: firstly by the Cauchy-Schwarz inequality (CSI) we have $D_k^{(n-1)}=\langle D_j^{(n)},|W_{:,k}^{(n)}|\rangle_{\ell^2}\leq \lVert W_{1,:}^{(n+1)}\rVert\times\lVert W_{:,k}^{(n)}\rVert$. Also, it is square summable since  $\sum_{k=1}^\infty (D_{k}^{(n-1)})^2 \leq \lVert W_{1,:}^{(n+1)}\rVert^2\times\lVert W_{:,:}^{(n)}\rVert^2$. The remaining terms $D_k^{(l)}$ for $l<n-1$ can be studied similarly.)

Note that $D_{k}^{(l)}$ is independent of the collection of random variables $\left\{ W_{i,j}^{(m)}:m\leq l,\forall i,\forall j\right\} $, a property we will call on repeatedly in the study of the moments. Using,  $D_{k}^{(l)}=\lim_{s\rightarrow\infty}\sum_{j=1}^{s}D_{j}^{(l+1)}|W_{j,k}^{(l+1)}|$, we have 
\[
\left(D_{k}^{(l)}\right)^{2}=\left(\lim_{s\rightarrow\infty}\sum_{j=1}^{s}D_{j}^{(l+1)}|W_{j,k}^{(l+1)}|\right)^{2}=\lim_{s\rightarrow\infty}\left(\sum_{j=1}^{s}D_{j}^{(l+1)}|W_{j,k}^{(l+1)}|\right)^{2}
\] and thus
\[
\mathbb{E}\left\{ \left(D_{k}^{(l)}\right)^{2}\right\} =\lim_{s\rightarrow\infty}\mathbb{E}\left\{ \left(\sum_{j=1}^{s}D_{j}^{(l+1)}|W_{j,k}^{(l+1)}|\right)^{2}\right\} =C_{l}k^{-\alpha}
\] where the final result of $C_{l}k^{-\alpha}$ will be established now. Squaring the finite sum in the expectation gives

\begin{alignat*}{1}
\sum_{j=1}^{s}\left(D_{j}^{(l+1)}\left|W_{j,k}^{(l+1)}\right|\right)^{2}+2\sum_{j=1}^{s}\sum_{i=j+1}^{s}\left(D_{j}^{(l+1)}\left|W_{j,k}^{(l+1)}\right|\right)\left(D_{i}^{(l+1)}\left|W_{i,k}^{(l+1)}\right|\right).
\end{alignat*}
The expected value of the cross term can be bounded by
\begin{alignat*}{1}
 & \sum_{j=1}^{s}\sum_{i=j+1}^{s}\sqrt{\mathbb{E}\left[\left(D_{j}^{(l+1)}W_{j,k}^{(l+1)}\right)^2\right]}\sqrt{\mathbb{E}\left[\left(D_{i}^{(l+1)}W_{i,k}^{(l+1)}\right)^2\right]}\\
&=\sum_{j=1}^{s}\sum_{i=j+1}^{s}\sqrt{\mathbb{E}\left[\left(D_{j}^{(l+1)}\right)^2\right]}\sqrt{\mathbb{E}\left[\left(W_{j,k}^{(l+1)}\right)^2\right]}\sqrt{\mathbb{E}\left[\left(D_{i}^{(l+1)}\right)^2\right]}\sqrt{\mathbb{E}\left[\left(W_{i,k}^{(l+1)}\right)^2\right]}
\end{alignat*}
due to the independence of $D_{j}^{(l+1)}$ and $W_{j,k}^{(l+1)}$. The expected value of the sum of squares term can be similarly upper bounded by
\begin{alignat*}{1}
\sum_{j=1}^{s}\mathbb{E}\left[\left(D_{j}^{(l+1)}\left|W_{j,k}^{(l+1)}\right|\right)^2\right] & =\sum_{j=1}^{s}\mathbb{E}\left[\left(D_{j}^{(l+1)}\right)^2\right]\mathbb{E}\left[\left(\left|W_{j,k}^{(l+1)}\right|\right)^2\right].
\end{alignat*}
For $l=n-1$, substituting the definitions above gives
\begin{alignat*}{1}
 \sum_{j=1}^{s}\mathbb{E}\left[\left(D_{j}^{(n)}\left|W_{j,k}^{(n)}\right|\right)^2\right]
  =\sum_{j=1}^{s}\mathbb{E}\left[\left(W_{1,j}^{(n+1)}W_{j,k}^{(n)}\right)^2\right]
  =k^{-\alpha}\sigma_{w^{(n+1)}}^{2}\sigma_{w^{(n)}}^{2}\sum_{j=1}^{s}j^{-2\alpha}.
\end{alignat*}
For the cross terms, we get
\begin{alignat*}{1}
\sum_{j=1}^{s}\sum_{i=j+1}^{s}\sqrt{\mathbb{E}\left[\left(D_{j}^{(n)}W_{j,k}^{(n)}\right)^2\right]}\sqrt{\mathbb{E}\left[\left(D_{i}^{(n)}W_{i,k}^{(n)}\right)^2\right]}
 =k^{-\alpha}\sigma_{w^{(n+1)}}^{2}\sigma_{w^{(n)}}^{2}\sum_{j=1}^{s}\sum_{i=j+1}^{s}j^{-\alpha}i^{-\alpha}.
\end{alignat*}
Thus $\mathbb{E}\left[\left(\partial v/\partial g_{k}^{(n-1)}\right)^2\right]\leq\mathbb{E}\left[\left(D_{k}^{(n-1)}\right)^2\right]=C_{n-1}k^{-\alpha}$, where the constant $C_{n-1}$ does not depend on $s$. The result can be extrapolated to all $l<n$ by induction over the layer index to get
\begin{equation}
\mathbb{E}\left[\left(\frac{\partial u}{\partial g_{k}^{(l)}}\right)^2\right]\leq\mathbb{E}\left[\left(D_{k}^{(l)}\right)^2\right]=C_{l}k^{-\alpha}\label{eq:dv_moment_2}
\end{equation}
for all $x\in\mathcal{X}$. Combining \eqref{eq:dv_dW} and \eqref{eq:dv_moment_2} and using Assumption \ref{assumption_bounded_activation} as well as Theorem \ref{thm:prior} gives
\begin{alignat*}{1}
\mathbb{E}\left[\left(\frac{\partial u}{\partial W_{i,j}^{(l)}}\right)^2\right]
 \leq\mathbb{E}\left[\left(\frac{\partial u}{\partial g_{i}^{(l)}}\right)^2\left(\frac{\partial f_{i}^{(l)}}{\partial W_{i,j}^{(l)}}\right)^2\right] 
\leq\mathbb{E}\left[\left(D_{i}^{(l)}\right)^2\right]\mathbb{E}\left[\left(g_{j}^{(l-1)}\right)^2\right]
\leq C_{l}i^{-\alpha}\sigma_{l-1}^{2}j^{-\alpha}
\end{alignat*}
for all $x\in\mathcal{X}$, for constants $C_l$ not depending on $s$. Similarly,
$\mathbb{E}\left[\left(\partial u/\partial B_{i}^{(l)}\right)^2\right]\leq\mathbb{E}\left[\left(\partial u/\partial g_{i}^{(l)}\right)^2\right]\leq C_{l}i^{-\alpha}.$
Bringing together these results gives the following bound for $\mathbb{E}(S_{s})$
in (\ref{eq:thm3_proof_Ss})
\begin{alignat*}{1}
\mathbb{E}(S_{s})\leq & c_{t}M^2 \times\left[\sigma_{b^{(n+1)}}^{2}C_{n+1}+\sum_{j=1}^{s}\frac{\sigma_{w^{(n+1)}}^{2}}{j^{\alpha}}C_{n+1}\sigma_{n}^{2}j^{-\alpha}\right]\\
 & +c_{t}M^2 \times\sum_{l=2}^{n}\sum_{i=1}^{s}\left[\frac{\sigma_{b^{(l)}}^{2}}{i^{\alpha}}C_{l}i^{-\alpha}+\sum_{j=1}^{s}\frac{\sigma_{w^{(l)}}^{2}}{(ij)^{\alpha}}C_{l}i^{-\alpha}\sigma_{l-1}^{2}j^{-\alpha}\right]\\
 & +c_{t}M^2 \times \sum_{i=1}^{s}\left[\frac{\sigma_{b^{(1)}}^{2}}{i^{\alpha}}C_1i^{-\alpha}+\sum_{j=1}^{d}\frac{\sigma_{w^{(1)}}^{2}}{i^{\alpha}}C_1i^{-\alpha}x_j^2\right].
\end{alignat*}
As $S_{s}$ is a submartingale and its mean is bounded uniformly in  $s$. By the martingale convergence theorem, it converges almost surely to an $L^{1}$ random variable and thus the result follows.
\end{proof}

\section{Details on experimental setup}
We here give further details on the experimental setup for the Examples \ref{sec:exp_setup}. Both examples are included in the python package `gym' \citep{1606.01540}.

\subsection{Mountaincar}
The first example is the popular mountaincar problem. The state space is the $2$-dimensional domain $\mathcal X=[-1.2,0.6]\times[-0.07,0.07]$, where the first variable is the position $x_1$ of the car on a mountain slope, and the second variable represents its velocity $x_2$. The set of possible actions is $\mathcal A=\{-1,0,1\}$, representing exerting force to the left, not adding force, and exerting force to the right, respectively. The state transitions are deterministic, being given by Newtonian physics, and we refer the reader to the OpenAI documentation or to our code for the details. 

In the mountaincar problem, the reward is constant $r(x_1,x_2)=-1$ per step, until the car reaches the top of the mountain ($x_1\geq0.5$). The optimal policy is therefore to reach the mountaintop as quickly as possible. An optimal deterministic policy \citep{xiao2019} is given by
\begin{align*}
    \mu(x_1,x_2)=-1+2\mathbbm I\{\min(-0.09(x_1+0.25)^2&+0.03,0.3(x_1+0.9)^4-0.008)\leq x_2\\&\leq -0.07(x_1+0.38)^2+0.07\},
\end{align*}
and we generated state-action pairs by firstly drawing a random initial state in the valley of the mountain, $x\sim\mathcal U([-0.6 , -0.4])$, i.e. a uniform value between $-0.6$ and $-0.4$. The initial velocity is set to $0$. Starting from that state, we computed the action given the optimal policy given above. Once the flag was reached, a new initial state was drawn, and the process repeated until we had a total of $250$ observations. This gave a set of state-action pairs $\{(x_t,a_t)\}_{t=1}^{250}$, and we then took every fifth sample to obtain the final dataset $y=\{(x_{5t},a_{5t})\}_{t=1}^{50}$. This resulted in the state variables in $y$ covering the entire state space, such that we can expect to learn the value function in any region an agent might find themselves in. The likelihood \eqref{likelihood_computable} arises from this dataset $y$ and the noise level being set to $\sigma=0.1$.

In the simulations from the learned value functions, we again initialised  the state variable as $x\sim\mathcal U([-0.6 , -0.4])$ and set the velocity to $0$. We then simulated noise and used Equation \eqref{random_action_2} with the learned value function to pick an action. In Section \ref{sec:prior_comparison}, the used value function was taken as either a sample from the posterior or as the mean function; in Section \ref{sec:policy_learning} the used value function was the mean function from the posteriors. In all experiments, if the car didn't make it to the flag within $200$ time steps, we called this a failure and restarted the process from new initial conditions. 

\subsection{HalfCheetah}
To show that our algorithm works in a more complicated setting, we looked at the HalfCheetah example from the MuJoCo library \citep{todorov2012mujoco} where the state $x_t$ a $17$-dimensional vector. The original continuous actions space of the problem is $6$-dimensional.

An agent controlling the cheetah is to move it to forward while not exerting too much force: positive rewards are given for moving forward, and negative rewards are given for moving backwards, a further penalty is deducted for actions requiring a lot of force. A black box optimal policy for the HalfCheetah problem was provided in Berkeley's Deep Reinforcement Learning Course\footnote{CS294-112 HW 1: Imitation Learning, \url{https://github.com/berkeleydeeprlcourse/homework/tree/master/hw1}}, which we used to simulate state-action pairs. 

The initial state and velocity variables were drawn at random with distributions according to the python package `gym' \citep{1606.01540}. We discretised the action space to $M$ actions in the following way: an initial state was drawn, and the black box policy gave us an action, taking us to a new state via deterministic mapping. Iterating this process, the first $M$ actions were stored. From now on, we can use a discrete action space $\mathcal A_M$ consisting of these $M$ actions: at a state $x_t$ we compute $a_t$ as the action in $\mathcal A_M$ that minimises the Euclidean distance to the action computed by the black box policy. We found that $M=8$ actions were sufficient to get behaviour very similar to the one we got when using the continuous action space, and we thus fixed $\mathcal A=\mathcal A_8$. We refer to the action $a\in\mathcal A$ that minimises the Euclidean distance to the black box algorithm as `optimal'. To generate data, we firstly drew an initial state $x_1$, and then computed the optimal action $a_1$ using the procedure just described, and computed the next state using the state dynamics \eqref{eq:state_transition}. After $25$ steps, we restarted from a new initial state, and repeated this process another $4$ times until we had a total of $T=100$ data points. The reason we restarted occasionally was, as in the mountaincar example, to ensure that we cover a representative region of the state space. The dataset $y=\{(x_t,a_t)\}_{t=1}^{100}$ was used in the likelihood (\ref{likelihood_computable}), where we set the noise level to $\sigma=0.1$. In the experiments in Section \ref{sec:policy_learning}, an initial state is drawn, and the cheetah is controlled using Equation \eqref{random_action_2} over $100$ time steps.

\end{document}